\PassOptionsToPackage{scaled=0.89}{helvet}
\documentclass{lmcs}
\pdfoutput=1
\usepackage[utf8]{inputenc}

\makeatletter
\def\enddoc@text{%
}
\makeatother

\usepackage{lastpage}
\lmcsdoi{21}{3}{5}
\lmcsheading{}{\pageref{LastPage}}{}{}%
{Jan.~17,~2024}{Jul.~22,~2025}{}

\usepackage{mathtools}
\usepackage{nicefrac}
\usepackage{tikz}
\usetikzlibrary{automata,arrows,positioning}
\usepackage{bm}
\usepackage[all]{xy}
\usepackage{url}

\usepackage{wrapfig}
\usepackage[most]{tcolorbox}
\usepackage{accents}

\usepackage[Algorithmus]{algorithm}
\usepackage[end]{algpseudocode}

\allowdisplaybreaks

\usepackage{amssymb,amsmath,stmaryrd}
\usepackage[hidelinks]{hyperref}

\newcommand{\mystrut}{\raisebox{0ex}[2.5ex][1ex]{}}

\renewcommand{\int}{\ensuremath{\mathbb{Z}}}

\newcommand{\one}{\ensuremath{1}}
\newcommand{\zero}{\ensuremath{0}}

\newcommand{\comp}[1]{\ensuremath{\overline{#1}}}

\newcommand{\Pow}[1]{\ensuremath{\mathcal{P}(#1)}}
\newcommand{\Powf}[1]{\ensuremath{\mathcal{P}_f(#1)}}

\newcommand{\Ytop}[2]{[{#1}]^{#2}}
\newcommand{\Ybot}[2]{[{#1}]_{#2}}

\newcommand{\interval}[2]{\ensuremath{[{#1}, {#2}]}}

\newcommand{\norm}[1]{\ensuremath{|\!|{#1}|\!|}}

\newcommand{\mins}{\min\nolimits}
\newcommand{\maxs}{\max\nolimits}

\newcommand{\suc}{\eta}

\makeatletter
\newcommand{\superimpose}[2]{%
  {\ooalign{$#1\@firstoftwo#2$\cr\hfil$#1\@secondoftwo#2$\hfil\cr}}}
\makeatother

\newcommand{\monM}{\mathbb{M}}

\newcommand{\latL}{\mathbb{L}}

\newcommand{\arr}{\dashrightarrow}

\newcommand{\cupa}[1]{\mathop{\cup_{#1}}}

\newcommand{\udefix}{\textsf{UDEfix}}

\definecolor{dmagenta}{rgb}{0.81,0,0.81}
\definecolor{dcyan}{rgb}{0,0.6,0.6}
\definecolor{dgreen}{rgb}{0.09, 0.45, 0.27}
\definecolor{dred}{rgb}{0.55, 0.0, 0.0}

\newcommand{\dgreen}{\color{dgreen}}

\newcommand{\red}{\color{dred}}

\newcommand{\mystrutab}{\raisebox{0ex}[0.4cm]{}}
\newcommand{\mystrutbl}{\raisebox{-2ex}[0.4cm]{}}

\usepackage[appendix=append,bibliography=common]{apxproof}

\usepackage{etoolbox} 
\makeatletter
\apptocmd{\@enddocumenthook}{
    \vspace{-2cm}
    \nobreak%
    \insert\copyins{\hsize.58\textwidth
\vbox to 0pt{\vskip12 pt%
      \fontsize{6}{7\p@}\normalfont\upshape
      \everypar{}%
      \noindent\fontencoding{T1}%
  \headertextsf{This work is licensed under the Creative Commons
  Attribution License. To view a copy of this license, visit
  \texttt{https://creativecommons.org/licenses/by/4.0/} or send a
  letter to Creative Commons, 171 Second St, Suite 300, San Francisco,
    CA 94105, USA, or Eisenacher Strasse~2, 10777 Berlin, Germany}\vss}
      \par
      \kern\z@}%
}
\makeatother

\theoremstyle{plain}
\newtheoremrep{theorem}{Theorem}[section]
\newtheoremrep{lemma}[theorem]{Lemma}
\newtheoremrep{proposition}[theorem]{Proposition}
\newtheoremrep{corollary}[theorem]{Corollary}

\theoremstyle{definition}
\newtheoremrep{definition}[theorem]{Definition}
\newtheoremrep{example}[theorem]{Example}
\newtheoremrep{remark}[theorem]{Remark}

\newcommand{\C}{\mathbb{C}}
\newcommand{\D}{\mathbb{D}}
\newcommand{\A}{\mathbb{A}}
\newcommand{\Cf}{\mathbb{C}_f}
\newcommand{\Af}{\A_f}

\makeatletter
\DeclareRobustCommand{\cev}[1]{%
  \mathpalette\do@cev{#1}%
}
\newcommand{\do@cev}[2]{%
  \fix@cev{#1}{+}%
  \reflectbox{$\m@th#1\vec{\reflectbox{$\fix@cev{#1}{-}\m@th#1#2\fix@cev{#1}{+}$}}$}%
  \fix@cev{#1}{-}%
}
\newcommand{\fix@cev}[2]{%
  \ifx#1\displaystyle
    \mkern#23mu
  \else
    \ifx#1\textstyle
      \mkern#23mu
    \else
      \ifx#1\scriptstyle
        \mkern#22mu
      \else
        \mkern#22mu
      \fi
    \fi
  \fi
}
\makeatother

\title{A Monoidal View on Fixpoint Checks}

\author[P.~Baldan]{Paolo Baldan\lmcsorcid{0000-0001-9357-5599}}[a]
\address{Universit\`a di Padova, Italy}
\email{baldan@math.unipd.it}

\author[R.~Eggert]{Richard Eggert\lmcsorcid{0000-0002-9901-7392}}[b]

\author[B.~K\"onig]{Barbara K\"onig\lmcsorcid{0000-0002-4193-2889}}[b]

\author[T.~Matt]{Timo Matt\lmcsorcid{0009-0008-5570-6998}}[b]

\address{Universit\"at Duisburg-Essen, Germany}
\email{richard.eggert@uni-due.de, barbara\_koenig@uni-due.de}

\author[T.~Padoan]{Tommaso Padoan\lmcsorcid{0000-0001-7814-1485}}[c]
\address{Universit\`a di Trieste, Italy}
\email{tommaso.padoan@units.it}

\thanks{Partially supported by the DFG project SpeQt (project number
  434050016), by the Ministero dell’Universtà e della Ricerca
  Scientifica of Italy, under Grant No. 201784YSZ5, PRIN2017 – ASPRA,
  by project iNEST, No. J43C22000320006 and by the European Union -
  NextGenerationEU under the National Recovery and Resilience Plan
  (NRRP) - Call PRIN 2022 PNRR - Project P2022HXNSC ``Resource
  Awareness in Programming: Algebra, Rewriting, and Analysis''.}
  
\begin{document}

\maketitle

\begin{abstract}
  Fixpoints are ubiquitous in computer science as they play a
  central role in providing a meaning to recursive and cyclic
  definitions. Bisimilarity, behavioural metrics, termination
  probabilities for Markov chains and stochastic games are defined in
  terms of least or greatest fixpoints. Here we show that our recent
  work which proposes a technique for checking whether the fixpoint of
  a function is the least (or the largest) admits a natural
  categorical interpretation in terms of gs-monoidal categories.
  The technique is based on a construction that maps a function to a
  suitable approximation. We study the compositionality
    properties of this mapping and show that under some restrictions
    it can naturally be interpreted as a (lax) gs-monoidal functor.
  This guides the development of a tool, called {\udefix} that allows
  us to build functions (and their approximations) like a circuit out
  of basic building blocks and subsequently perform the fixpoints
  checks.
  We also show that a slight generalisation of the theory allows one
  to treat a new relevant case study: coalgebraic behavioural metrics
  based on Wasserstein liftings.
\end{abstract}

\section{Introduction}
\label{sec:introduction}

Fixpoints are fundamental in computer science: extremal (least or
greatest) fixpoints are commonly used to interpret recursive and
cyclic definitions. In a recent
work~\cite{bekp:fixpoint-theory-upside-down,BEKP:FTUD-journal} we
proposed a technique for checking whether the fixpoint of a function
is the least (or the largest). In this paper we show that such
technique admits a natural categorical interpretation in terms of
gs-monoidal categories. This allows us to provide a compositional
flavour to our technique and enables the realisation of a
corresponding tool {\udefix}.  

The theory
in~\cite{bekp:fixpoint-theory-upside-down,BEKP:FTUD-journal} can be
used in a variety of fairly diverse application scenarios, such as
bisimilarity~\cite{s:bisimulation-coinduction}, behavioural
metrics~\cite{dgjp:metrics-labelled-markov,b:prob-bisimilarity-distances,cbw:complexity-prob-bisimilarity,bbkk:coalgebraic-behavioral-metrics},
termination probabilities for Markov chains~\cite{bk:principles-mc}
and simple stochastic games~\cite{c:algorithms-ssg}.  It applies to
non-expansive functions of the kind $f\colon \monM^Y \to \monM^Y$,
where $\monM$ is a set of values and $Y$ is a finite set. More
precisely, the set of values $\monM$ is an MV-chain, i.e., a totally
ordered complete lattice endowed with suitable operations of sum and
complement, which allow for the definition of a natural notion of
non-expansiveness for $\monM$-valued functions. A prototypical example
of MV-algebra, largely used in the examples, is the interval $[0,1]$
with truncated sum.
Roughly, the idea consists in mapping the semantic functions of
interest $f\colon \monM^Y\to \monM^Y$ to corresponding approximations
$f^{\#}$ over (a subset of) $\Pow{Y}$, the powerset of $Y$. While
$\monM$ is in general infinite (or very large), the set $Y$ is
typically finite in a way that the fixpoints of $f^{\#}$ can be
computed effectively and provide information on the fixpoints of the
original function. In particular they allow us to decide, whether a
given fixpoint is indeed the least one (or dually greatest one).

In this paper, we show that the approximation framework and its
compositionality properties can be naturally interpreted in
categorical terms using gs-monoidal categories. In essence gs-monoidal
categories describe graph-like structures with dedicated input and
output interfaces, operators for disjoint union (tensor), duplication
and termination of wires, quotiented by appropriate
axioms. Particularly useful are gs-monoidal functors that preserve
such operators and hence naturally describe compositional operations.

More concretely, we introduce two gs-monoidal categories, that we
refer to as the concrete category and the category of approximations,
where the concrete functions and their approximations, respectively,
live as arrows. We then define a mapping $\#$ from the concrete
category to the category of approximations and show that, whenever, as
in~\cite{BEKP:FTUD-journal}, we assume finiteness of the underlying
set $Y$, the mapping $\#$ turns out to be a gs-monoidal functor.

In the general case, for functions $f\colon \monM^Y \to \monM^Y$,
where $Y$ is not necessarily finite, the mapping $\#$ is only known to
be the union of lax functors. However, we observe that we can
characterise it again as a gs-monoidal functor if we restrict it to
the subcategory of the concrete category where arrows are
reindexings. While this might seem to be a severe restriction, we will
see that it is sufficient to cover the intended applications.

We prove a number of properties of $\#$ that enable us to give a
recipe for finding approximations for a special type of functions:
predicate liftings as those introduced for coalgebraic modal
logic~\cite{p:coalgebraic-logic,s:coalg-logics-limits-beyond-journal}. This
allows us to include a new case study in the machinery for fixpoint
checking: coalgebraic behavioural metrics, based on Wasserstein
liftings.

Besides shedding further light on the theoretical approximation
framework of~\cite{BEKP:FTUD-journal}, the results in the paper guide
the realisation of a tool, called {\udefix} which realises a number of
fixpoints checks (given a post- resp.\ pre-fixpoint, is below the
least resp.\ above the greatest fixpoint? If it is a fixpoint, is it
the least or the greatest?).
Leveraging the visual string diagrammatic language of gs-monoidal
categories, {\udefix} models system functions as (hyper)graphs whose
edges represent suitable basic building blocks. This also extends to
approximations, thanks to the fact that the mapping $\#$ can be seen
as a gs-monoidal functor, enabling a compositional construction of the
approximation function.

The paper is organised as follows. In Section~\ref{sec:motivation} we
provide some high-level motivation.  In
Section~\ref{sec:preliminaries} we introduce some preliminaries and,
in particular, we review the theory if fixpoint checks
from~\cite{BEKP:FTUD-journal}.  In
Section~\ref{sec:categories-functor} we introduce two (gs-monoidal)
categories, $\C$ and $\A$, the so-called concrete category and
category of approximations, and investigate the properties of the
approximation map $\# : \C \to \A$.  In Section~\ref{sec:liftings} we
show how to handle predicate liftings and then, in
Section~\ref{sec:beh-metrics}, we exploit such results to treat
behavioural metrics. In Section~\ref{sec:gs-mon}, we show that, in the
finitary case, the categories $\C$, $\A$ are indeed gs-monoidal and
$\#$ is gs-monoidal functor. The tool {\udefix} is discussed in
Section~\ref{sec:tool}. Finally, in Section~\ref{sec:conclusion} we
draw some conclusions.

This article is an extended version of the paper~\cite{BEKMP:MVFPC}
presented at ICGT 2023. With respect to the conference version, the
present paper contains additional explanations and examples and full
proofs of the results.
    
\section{Motivation}
\label{sec:motivation}

We start by giving a high-level overview over the approach, followed
by a small worked-out example. The general idea follows the
upside-down theory of fixpoint checks
from~\cite{bekp:fixpoint-theory-upside-down}. We are given a monotone
function\footnote{The set of functions from $X$ to $Y$ is denoted by
  $Y^X$.}  $f\colon [0,1]^Y\to [0,1]^Y$ (later $[0,1]$ will be
replaced by a general MV-algebra) and we also assume that $f$ is
non-expansive with respect to a suitable norm.  By
Knaster-Tarski~\cite{t:lattice-fixed-point} we have a guarantee that
$f$ has fixpoints -- among them a least fixpoint $\mu f$ and a
greatest fixpoint $\nu f$ -- but there is no guarantee that there is a
unique fixpoint, indeed there might be several of them.

Now, given a fixpoint $a\in [0,1]^Y$ of $f$, our aim is to check
whether $a$ is the least fixpoint, i.e., $a=\mu f$. (Or dually,
whether it is the greatest fixpoint, i.e., $a=\nu f$.)

The check proceeds by determining from $f$ and $a$ a so-called
approximation $f_\#^a\colon \mathcal{P}(Y)\to
\mathcal{P}(Y)$. Intuitively, the approximation describes how function
$f$ propagates decreases of its argument $a$: for $Y' \subseteq Y$, the set
$f_\#^a(Y')$ contains those elements $y \in Y$ where the value of
$f(a)$ decreases by some fixed amount $\delta$ whenever we decrease
$a$ on $Y'\subseteq Y$ by $\delta$. 
Then one determines the
greatest fixpoint of this approximation ($\nu f_\#^a$) and observes
that $a$ coincides with $\mu f$ if and only if $\nu f_\#^a$ is the
empty set. Otherwise $a$ is too large and intuitively there is still
``wiggle room'' and potential for decrease. This procedure works under
conditions that will be made precise later in
Section~\ref{ss:fix-check}.

The technique is schematised in Figure~\ref{fig:schema}. A feature of
this approach is the fact that the function $f$ can be assembled
compositionally from basic functions via composition and disjoint
union. Such a decomposition can be represented via a string
diagram. Furthermore, approximations can be determined compositionally,
leading to a string diagram of the same type, but with approximations
instead of concrete functions.

\begin{figure}
  \centering
  \scalebox{0.5}{\input{schema.tex}}
  \caption{A graphical overview over the fixpoint checking method}
  \label{fig:schema}
\end{figure}

\medskip

We conclude with an example which illustrates the notions and tools
discussed above for fixpoint checks.  We
consider (unlabelled) Markov chains $(S,T,\eta)$, where $S$
is a finite set of states, $T\subseteq S$ is a set of terminal states
and $\eta\colon S\backslash T\to \mathcal{D}(S)$ assigns to each
non-terminal state a probability distribution over its successors.  We
denote by $p_s = \eta(s)$ the probability distribution associated with
state $s$.

We are interested in the termination probability of a given state of
the Markov chain, which can be computed by taking the least fixpoint
of a function $\mathcal{T}\colon [0,1]^S \to [0,1]^S$:
\begin{eqnarray*}
  && \mathcal{T} : [0,1]^S \to [0,1]^S \\
  && \mathcal{T}(t)(s) = \left\{
    \begin{array}{ll}
      1 & \mbox{if $s\in T$} \\
      \sum\limits_{s'\in S} \suc(s)(s') \cdot t(s') &
      \mbox{otherwise}
    \end{array}
  \right.
\end{eqnarray*}

\begin{table}
  \begin{tabular}{|c|c|c|c|c|c|} 
    \hline
    Function & $c_k$ & $g^*$  & $\min_u$ &
    $\max_u$ & $\mathrm{av}_D = \tilde{\mathcal{D}}$ \mystrut \\
    & $k\colon Z \to \monM$ & $g\colon Z \to Y$  &
    $u\colon Y\to Z$ & $u\colon Y\to
    Z$ & $\monM=[0,1]$, $Z = \mathcal{D}(Y)$  \\
    \hline
    Name & constant & reindexing & minimum & maximum & expectation \\
    \hline 
    $a\mapsto \dots$ & $k$ & $a\circ g$ & $\lambda
    z.\min\limits_{u(y)=z} a(y)$ & $\lambda
    z.\max\limits_{u(y)=z} a(y)$ & $\lambda z.\sum\limits_{y\in
      Y} z(y)\cdot a(y)$  \mystrut \\
    \hline
  \end{tabular}

  \medskip
  
  \caption{Basic functions of type $\monM^Y\to\monM^Z$, $a\colon Y\to\monM$.}
  \label{tab:basic-functions}
\end{table}

We observe that the function $\mathcal{T}$ can be decomposed as
\[
  \mathcal{T} = (\eta^*\circ \tilde{D})\otimes c_k
\]
where $\otimes$ stands for disjoint union (over functions) and we use
some of the basic functions given in
Table~\ref{tab:basic-functions}. We depict this decomposition
diagrammatically in Figure~\ref{fig:decomp-termination}.  For the
Markov chain in Figure~\ref{fig:markov-chain2} the parameters
instantiate as follows (cf. Table~\ref{tab:basic-functions}):
\begin{itemize}
\item $c_k\colon [0,1]^\emptyset\to [0,1]^T$, $k\colon T\to [0,1]$
  with $T=\{u\}$ and $k(u) = 1$
\item $\tilde{\mathcal{D}}\colon [0,1]^S \to [0,1]^{\mathcal{D}(S)}$
\item $\eta^* \colon [0,1]^{\mathcal{D}(S)} \to [0,1]^{S\setminus T}$
  where $\eta\colon S\setminus T\to {\mathcal{D}(S)}$ with
  $\eta(s) = p_s$, where $p_x(y) = p_x(u) = \nicefrac{1}{2}$,
  $p_y(z) = p_z(y) = 1$ and all other values are $0$. 
\end{itemize}
Since $S$ is finite we could alternatively restrict to a finite subset
$D$ of $\mathcal{D}(S)$.

\begin{figure}
  \centering
  \begin{tikzpicture}[every node/.style={inner xsep=0.2em, minimum height=2em, outer sep=0, rectangle},scale=1.2]
    \node (c) at (0,0) [draw] {$c_k$};
    \node (init1) at (-1.5,0) {};
    \node (init2) at (-1.5,-1.2) {};
    \node (D) at (0,-1.2) [draw] {$\tilde{\mathcal{D}}$};
    \node (eta) at (2.4,-1.2) [draw] {$\eta^*$};
    \node (end1) at (4.5,0) {};
    \node (end2) at (4.5,-1.2) {};
    \draw [-] (c) to node [above, pos=0.5]{\footnotesize $[0,1]^T$} (end1);
    \draw [-] (init1) to node [below, pos=0.43] {\footnotesize $[0,1]^\emptyset$} (c);
    \draw [-] (init2) to node [below, pos=0.43] {\footnotesize
      $[0,1]^S$} (D);
    \draw [-] (D) to node [below] {\footnotesize $[0,1]^{\mathcal{D}(S)}$} (eta);
    \draw [-] (eta) to node [below] {\footnotesize
      $[0,1]^{S\backslash T}$}
    (end2);
  \end{tikzpicture}
  \caption{Decomposition of the fixpoint function $\mathcal{T}$ for
    computing termination probabilities.}
  \label{fig:decomp-termination}
\end{figure}

The function $\mathcal{T}$ is a monotone function on a complete
lattice, hence it has a least fixpoint by Knaster-Tarski's fixpoint
theorem~\cite{t:lattice-fixed-point}.  Furthermore, it is
non-expansive, allowing us to use the fixpoint checking techniques
from \cite{BEKP:FTUD-journal}.

We will illustrate this on the concrete instance (Markov chain in
Figure~\ref{fig:markov-chain2}). The state set is $S = \{x,y,u,z\}$
and $u$ is the only terminal state. The least fixpoint
$\mu\mathcal{T}$ of $\mathcal{T}$ is given in
Figure~\ref{fig:markov-chain2} in green (left) and the greatest
fixpoint $\nu\mathcal{T}$ in red (right). These are two of the
infinitely many fixpoints of $\mathcal{T}$. The cycle on the left
(including $y,z$) can be seen as some kind of vicious cycle, where
$y,z$ convince each other erroneously that they terminate with a
probability that might be too high. Such cycles have a coinductive
flavour, hence the computation of the greatest fixpoint of the
approximation.
\begin{figure}[t]
  \centering
  \begin{tikzpicture}[->, >=stealth, nodes={minimum size=1.8em},
    node distance=0.8cm]
    \node [circle,draw](x) [label=below:{\dgreen
      $\nicefrac{1}{2}$}/{\red $1$}] {$x$}; \node
    [circle,draw,accepting](t) [right=of x] [label=below:{\dgreen
      $1$}/{\red $1$}] {$u$}; \node [circle,draw](y) [left=of x]
    [label=below:{\dgreen $0$}/{\red $1$}] {$y$}; \node
    [circle,draw](z) [left=of y][label=below:{\dgreen $0$}/{\red
      $1$}] {$z$};
    \draw [->,thick] (x) to node [above]{$\nicefrac{1}{2}$} (t);
    \draw [->,thick] (x) to node [above]{$\nicefrac{1}{2}$} (y);
    \draw [->,thick,bend left] (y) to node [below]{$1$} (z); \draw
    [->,thick,bend left] (z) to node [above]{$1$} (y);
  \end{tikzpicture}
  \caption{A Markov chain with two fixpoints of $\mathcal{T}$
    (right)}
  \label{fig:markov-chain2}
\end{figure}
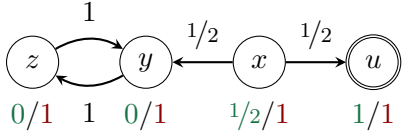

Now let $a = \nu \mathcal{T}$ be the greatest fixpoint of
$\mathcal{T}$, i.e., the function $S\to [0,1]$ that maps every state
to~$1$. In this case we can associate $\mathcal{T}$ with an
approximation $\mathcal{T}_\#^a$ on subsets of $S$ that
records the propagation of decrease. Given $S' \subseteq S$, the
approximation $\mathcal{T}_\#^a$ is as follows:
\[
  \mathcal{T}_\#^a(S') = \{s \in S \mid a(s) > 0, s \notin T,
  \mathit{supp}(\eta(s)) \subseteq S'\}.
\]
where $\mathit{supp}(p)$ denotes the \emph{support} of a probability
distribution $p\colon S\to[0,1]$, i.e., the set of states $v$ for which
$p(v) > 0$.
Intuitively, if we decide to
decrease the $a$-values of all states in $S'$ by some small amount
$\delta$, the states in $\mathcal{T}_\#^a(S')$ will also decrease
their values by $\delta$ after applying $\mathcal{T}$. This is true
for all states that have non-zero value, are not terminal and whose
successors are \emph{all} included in $S'$.

In the example $\mathcal{T}_\#^a(\{y,z\}) = \{y,z\}$ and $\{y,z\}$
is indeed the greatest fixpoint of the approximation. Since it is
non-empty, we  deduce that $a$ is not the least
fixpoint. Furthermore we could now subtract a small value (for
details on how to obtain this value see~\cite[Proposition
4.5]{BEKP:FTUD-journal}) from $a(y)$, $a(z)$ to obtain a smaller
pre-fixpoint, from where one can continue to iterate to the least
fixpoint (see also~\cite{bekp:lattice-strategy-iteration}).

We anticipate that in our tool {\udefix} we can draw a diagram as in
Figure~\ref{fig:decomp-termination}, from which the approximation and
its greatest fixpoint are automatically computed in a compositional
way, allowing us to perform such fixpoint checks.
A more complex case study in the domain of behavioural
metrics will be provided in Section~\ref{sec:case-study}.  

\section{Preliminaries}
\label{sec:preliminaries}

This section reviews some background used throughout the paper. We
first introduce some basics of lattices and MV-algebras, the
domain where the functions of interest take values. Then we recap some
results from~\cite{BEKP:FTUD-journal} useful for detecting if a
fixpoint of a given function is the least (or greatest).

\subsection{Lattices and MV-algebras}
\label{ss:MV}

A \emph{partially ordered} set $(P, \sqsubseteq)$ is often denoted
simply as $P$, omitting the order relation. For $x, y \in P$, we write
$x \sqsubset y$ when $x \sqsubseteq y$ and $x \neq y$.
For a function $f : X \to P$, we will write $\arg\min_{x\in X'} f(x)$
to denote the (possibly empty) set of elements where $f$ reaches the
minimum on domain $X'$, i.e.,
\begin{center}
  $\arg\min_{x\in X'} f(x) = \{ x \in X' \mid \forall y \in X.\, f(x) \sqsubseteq f(y)\}$.
\end{center}
Abusing the notation, we will write $z = \arg\min_{x\in X} f(x)$
instead of $z \in \arg\min_{x\in X} f(x)$.

\begin{definition}[complete lattice]
  A \emph{complete lattice} is a partially ordered set
  $(\latL, \sqsubseteq)$ such that each subset $X \subseteq \latL$
  admits a join $\bigsqcup X$ and a meet $\bigsqcap X$. A complete
  lattice $(\latL, \sqsubseteq)$ always has a least element
  $\bot = \bigsqcap \latL$ and a greatest element
  $\top = \bigsqcup \latL$.
\end{definition}

A prototypical example for a complete lattice is the \emph{powerset}
$\mathcal{P}(Y)$ of a given set $Y$, where the partial order is
subset inclusion ($\subseteq$) and join and meet are given by union and
intersection. The set of finite subsets of $X$ is written $\Powf{X}$.

A function $f : \latL \to \latL$ is \emph{monotone} if for all
$l, l' \in \latL$, if $l \sqsubseteq l'$ then
$f(l) \sqsubseteq f(l')$. By Knaster-Tarski's
theorem~\cite[Theorem~1]{t:lattice-fixed-point}, any monotone function
on a complete lattice has a least fixpoint $\mu f$ and a greatest
fixpoint $\nu f$, characterised as the meet of all pre-fixpoints
$\mu f = \bigsqcap \{ l \mid f(l) \sqsubseteq l \}$ and, dually, a
greatest fixpoint $\nu f = \bigsqcup \{ l \mid l \sqsubseteq f(l) \}$,
characterised as the join of all post-fixpoints.  We denote by
$\mathit{Fix}(f)$ the set of all fixpoints of $f$.

For a set $Y$ and a complete lattice $\latL$, the set of functions
$\latL^Y = \{ f \mid f : Y \to \latL \}$ with pointwise order (for
$a, b \in \latL^Y$, $a \sqsubseteq b$ if $a(y) \sqsubseteq b(y)$ for
all $y\in Y$), is a complete lattice.

The semantic functions of interest in the paper will take values on
special lattices, induced by a suitable class of commutative monoids
with a complement operation.

\begin{definition}[MV-algebra~\cite{Mun:MV}]
  An \emph{MV-algebra} is a tuple
  $\monM = (M, \oplus, \zero, \comp{(\cdot)})$ where
  $(M, \oplus, \zero)$ is a commutative monoid and
  $\comp{(\cdot)} : M \to M$ maps each element to its \emph{complement},
  such that, if we let $\one = \comp{\zero}$ and \emph{subtraction}
  $x \ominus y = \comp{\comp{x} \oplus y}$, then
  for all $x, y \in M$ it holds that
  \begin{enumerate}
  \item $\comp{\comp{x}} = x$;
  \item $x \oplus \one = \one$;    
  \item $(x \ominus y) \oplus y = (y \ominus x) \oplus x$.
  \end{enumerate}
\end{definition}

MV-algebras can be endowed with a partial order, the so-called
\emph{natural order}, defined for $x,y\in M$, by $x \sqsubseteq y$
if $x \oplus z= y$ for some $z \in M$. When $\sqsubseteq$ is total,
$\monM$ is called an \emph{MV-chain}. We will often write $\monM$ instead of
$M$.

The natural order gives an MV-algebra a lattice structure where
$\bot = \zero$, $\top =\one$, $x \sqcup y = (x \ominus y) \oplus y$
and
$x \sqcap y = \comp{\comp{x} \sqcup \comp{y}} = x \ominus (x \ominus y)$. We call the MV-algebra \emph{complete} if it is a
complete lattice, which is not true in general, e.g.,
$([0,1] \cap \mathbb{Q}, \leq)$.

\begin{example}
  \label{ex:mv-chains}
  A prototypical MV-algebra is $([0,1],\oplus,0,\comp{(\cdot)})$ where
  $x\oplus y = \min\{x+y,1\}$, $\comp{x} = 1-x$ and
  $x\ominus y = \max\{0,x-y\}$ for $x,y\in [0,1]$.  The natural order
  is $\le$ (less or equal) on the reals.  Another example is
  $K = (\{0,\dots,k\},\oplus,0,\comp{(\cdot)})$ where
  $n\oplus m = \min\{n+m,k\}$, $\comp{n} = k-n$ and
  $n\ominus m = \max\{ n-m,0\}$ for $n,m\in \{0,\dots,k\}$. Both
  MV-algebras are complete and MV-chains.
\end{example}

\subsection{A Theory of Fixpoint Checks}
\label{ss:fix-check}

We briefly recap the theory from~\cite{BEKP:FTUD-journal}. Given a
function $f : \monM^Y \to \monM^Y$, where $\monM$ is an MV-algebra,
the theory provides results useful for checking whether a fixpoint of
$f$ is the least or the greatest fixpoint.
For gaining some intuition, one can think that, as in
Section~\ref{sec:motivation}, $Y$ is the set of states of a Markov
Chain, $\monM = [0,1]$ and functions $\monM^Y = [0,1]^Y$ provide the
probability of termination of each state. Alternatively, $Y$ could be
the set of pairs of states of a system, $\monM = [0,1]$ and functions
$\monM^Y = [0,1]^Y$ provide the behavioural distance between states.
The technique is for instance useful in a situation where one can
obtain a fixpoint for the function of interest, e.g., via strategy
iteration, but it is unclear whether the fixpoint is the least (or the
largest).

While the theory in~\cite{BEKP:FTUD-journal} was restricted to
functions over $\monM^Y$ where $Y$ is finite, for the purposes of the
present paper we actually need a generalisation working for functions
with an infinite domain. Hence, hereafter $Y$ and $Z$ denote possibly
infinite sets.

\begin{definition}[norm and non-expansive functions]
  Given $a \in \monM^Y$ we define its \emph{norm} as
  $\norm{a} = \bigsqcup \{ a(y) \mid y \in Y\}$.  A function
  $f: \monM^Y\to \monM^Z$ is \emph{non-expansive} if for all
  $a, b \in \monM^Y$ it holds
  $\norm{f(b) \ominus f(a)} \sqsubseteq \norm{b \ominus a}$.
\end{definition}

It can be seen that non-expansive functions are monotone. A number of
standard operators are non-expansive (e.g., constants, reindexing, max
and min over a relation, average in Table~\ref{tab:basic-functions}),
and non-expansiveness is preserved by composition and disjoint union
(see~\cite[Theorem 5.2]{BEKP:FTUD-journal}).

Let $f: \monM^Y\to \monM^Y$, $a \in \monM^Y$.  For a non-expansive
endo-function $f: \monM^Y\to \monM^Y$ and $a \in \monM^Y$, the theory
in~\cite{BEKP:FTUD-journal} provides a so-called $a$-approximation
$f_\#^a$ of $f$, which is an endo-function over a suitable subset of
$\Pow{Y}$. Intuitively, given some $Y'$, the set $f_\#^a(Y')$ contains
the points where a decrease of the values of $a$ on the points in $Y'$
``propagates'' through the function $f$.
Understanding that no decrease can be propagated allows one to
establish when a fixpoint of a non-expansive function $f$ is actually
the least one, and, more generally, when a (post-)fixpoint of $f$ is
above the least fixpoint.

The above intuition is formalised using tools from abstract
interpretation~\cite{cc:ai-unified-lattice-model,CC:TLA}. In
particular, we define a pair of functions, which, under suitable
conditions, form a Galois connection.  For $a, b \in \monM^Y$, let
$[a,b] = \{c\in\monM^Y \mid a\sqsubseteq c\sqsubseteq b\}$ and let
\[ \Ytop{Y}{a} = \{ y\in Y\mid a(y) \neq 0 \}. \] Intuitively
$\Ytop{Y}{a}$ contains those elements where we have ``wiggle room'',
i.e., a decrease is in principle feasible. Then for
$0\sqsubset\delta \in \monM$ define
$\alpha^{a,\delta} : \Pow{\Ytop{Y}{a}} \to \interval{a \ominus
  \delta}{a}$ and
$\gamma^{a,\delta} : \interval{a \ominus \delta}{a} \to
\Pow{\Ytop{Y}{a}}$, as follows: for $Y' \in \Pow{\Ytop{Y}{a}}$ and
$b \in \interval{a \ominus \delta}{a}$, we have
\[
  \alpha^{a,\delta}(Y') = a \ominus \delta_{Y'} \qquad  \qquad
  \gamma^{a,\delta}(b) = \{ y\in \Ytop{Y}{a} \mid a(y) \ominus b(y)
  \sqsupseteq \delta \}.
\]
where we write $\delta_{Y'}$ for the function
defined by $\delta_{Y'}(y) = \delta$ if $y \in Y'$ and
$\delta_{Y'}(y) = \zero$, otherwise.

\begin{center}
  \begin{tikzpicture}[->]
    \node (y) {$\Pow{\Ybot{Y}{a}}$};
    \node [right=of y] (b) 
    {$\interval{a}{a+\delta}$}; 
    \draw [->,thick,bend left] (y) to node [above]{$\alpha_{a,\delta}$} (b); 
    \draw [->,thick,bend left] (b) to node [below]{$\gamma_{a,\delta}$}
    (y);
  \end{tikzpicture}
\end{center}
One can see that for sufficiently small $\delta$ the pair above is
indeed a Galois connection.

Now, in order to allow for a compositional approach, the approximation
is defined for non-expansive functions, where domain and codomain
are possibly distinct.

\begin{definition}[approximation]
  Given $f: \monM^Y\to \monM^Z$ and $\delta \in \monM$, define
  $f^{a,\delta}_\# \colon \Pow{\Ytop{Y}{a}} \to \Pow{\Ytop{Z}{f(a)}}$ as
  $f_\#^{a,\delta} = \gamma^{f(a),\delta} \circ f \circ
  \alpha^{a,\delta}$.
  We then define the \emph{$a$-approximation} of $f$
  as
  \[ f_\#^a = \bigcup_{\delta \sqsupset 0} f_\#^{a,\delta}.\]
\end{definition}

Intuitively, for $Y' \subseteq \Ytop{Y}{a}$, we have that
$f_\#^{a,\delta}(Y')$ is the set of points to which $f$ propagates a
decrease of the function $a$ with value $\delta$ on the subset $Y'$.

For finite sets $Y$ and $Z$ all functions $f_\#^{a,\delta}$, for
$\delta$ below some bound, are equal. Here, the $a$-approximation is
given by $f_\#^a = f_\#^{a,\delta}$ for such a $\delta$.

As stated above, the set $f_\#^a(Y')$ contains the points where a
decrease of the values of $a$ on the points in $Y'$ is propagated by
applying the function $f$.  The greatest fixpoint of $f_\#^a$ gives us
the subset of $Y$ where such a decrease is propagated in a cycle (a
so-called ``vicious cycle''). Whenever $\nu f_\#^a$ is non-empty, one
can argue that $a$ cannot be the least fixpoint of $f$ since we can
decrease the value of $a$ at all elements of $\nu f_\#^a$, obtaining a
smaller pre-fixpoint. Interestingly, for non-expansive functions, also
the converse holds, i.e., emptiness of the greatest fixpoint of
$f_\#^a$ implies that $a$ is the least fixpoint. This is summarised by
the following result from~\cite{BEKP:FTUD-journal}.

\begin{theorem}[soundness and completeness for fixpoints]
  \label{th:fixpoint-sound-compl}
  Let $\monM$ be a complete
  MV-chain, $Y$ a finite set and $f : \monM^Y\to \monM^Y$ be a
  non-expansive function.
  Let $a \in \monM^Y$ be a fixpoint of $f$. Then 
  $\nu f^a_{\#} = \emptyset$ if and only if $a = \mu f$.
  
  If $a$ is not the least fixpoint and thus
  $\nu f^a_{\#} \neq \emptyset$ then there is
  $0 \sqsubset \delta \in \monM$ such that
  $a \ominus \delta_{\nu f^a_{\#}}$ is a pre-fixpoint of $f$.
\end{theorem}

Using the above theorem we can check whether some fixpoint $a$ of $f$
is the least fixpoint. Whenever $a$ is a fixpoint, but not yet the
least fixpoint of $f$, it can be decreased by a fixed value in $\monM$
(see~\cite[Proposition 4.5]{BEKP:FTUD-journal} for the details) on the points in
$\nu f^a_{\#}$ to obtain a smaller pre-fixpoint. In this way we obtain
$a' \sqsubset a$ such that $f(a') \sqsubseteq a'$ and can continue
fixpoint iteration from there.

This results in the following biconditional proof rule (where
``biconditional'' means that it can be used in both directions).

\[
  \frac{a = f(a) \qquad \nu f_\#^a = \emptyset}{a = \mu f}
\]

When $a \in \monM^Y$ is not a fixpoint, but a post-fixpoint of $f$
(i.e., $a\sqsubseteq f(a)$), a restriction of
the $a$-approximation of $f$ leading to a sound (but not biconditional)
rule.

\begin{lemma}[soundness for post-fixpoints]
  Let $\monM$ be a complete MV-chain, $Y$ a finite set and
  $f : \monM^Y\to \monM^Y$ a non-expansive function, $a \in \monM^Y$
  such that $a \sqsubseteq f(a)$.  Define
  \[ \Ytop{Y}{a=f(a)} = \{ y \in \Ytop{Y}{a} \mid a(y) = f(a)(y) \} \]
  and restrict the approximation
  $f^a_{\#}: \Pow{\Ytop{Y}{a}} \to \Pow{\Ytop{Y}{f(a)}}$ to an
  endo-function
  \begin{eqnarray*}
    f_*^a : \Ytop{Y}{a=f(a)} & \to & \Ytop{Y}{a=f(a)} \\
    f_*^a(Y') & = & f^a_{\#}(Y') \cap \Ytop{Y}{a=f(a)}
  \end{eqnarray*}
  If $\nu f^a_* = \emptyset$ then $a \sqsubseteq \mu f$.
\end{lemma}

Written more compactly, we obtain the following proof rule:

\[
  \frac{a \sqsubseteq f(a) \qquad \nu f_*^a = \emptyset}{a \sqsubseteq
    \mu f}
\]

The above theory can be easily be dualised to checking greatest fixpoints.

\subsection{Approximations for Basic Functions}
\label{sec:approximations-basic}

We next provide the approximation for a number of non-expansive
functions which can be used as basic building blocks for constructing
the semantics functions of interest (Some of them have already been
considered in Table~\ref{tab:basic-functions}.)

In Table~\ref{tab:basic-functions-approximations} these basic
non-expansive functions are listed together with their
approximations. In particular, $\mathcal{D}(Y)\subseteq [0,1]^Y$
denotes the set of finitely supported probability distributions, i.e.,
functions $\beta : Y \to [0,1]$ with finite support such that
$\sum_{y \in Y} \beta(y) = 1$.

Note that functions $\min$ and $\max$ are slightly generalised with respect to
Table~\ref{tab:basic-functions} as they can be parameterised by
relations instead of functions.
We stress, in particular, that the approximation of reindexing is
given by the inverse image, a fact that will be extensively used in the
paper.

\begin{table}[h]
  \small
  \begin{center}
    \begin{tabular}{|l|l|l|}
      \hline function $f$ & definition of $f$ & $f_\#^a(Y')$ \\
      \hline\hline
      $c_k$ & $f(a) = k$ & $\emptyset$ \mystrutbl \\        
      ($k\in\monM^Z$) && \\ \hline
      $u^*$ & $f(a) = a \circ u$ & $u^{-1}(Y')$ \mystrutbl \\
      ($u\colon Z\to Y$) & & \\ \hline
      $\mins_\mathcal{R}$ &
      $f(a)(z) = \min\limits_{y\mathcal{R}z} a(y)$ &
      $\{z \in \Ytop{Z}{f(a)} \mid\ \arg\min\limits_{y\in\mathcal{R}^{-1}(z)} a(y)\cap Y' \neq \emptyset\}$ \mystrutbl \\       
      ($\mathcal{R}\subseteq Y\times Z$) & & \qquad        
      \\ \hline
      $\maxs_\mathcal{R}$ &
      $f(a)(z) = \max\limits_{y\mathcal{R}z} a(y)$ &
      $\{z \in \Ytop{Z}{f(a)} \mid\ \arg\max\limits_{y\in\mathcal{R}^{-1}(z)} a(y) \subseteq
      Y'\}$
      \mystrutbl \\
      ($\mathcal{R}\subseteq Y\times Z$) &&       
      \\
      \hline \mystrutab$\tilde{\mathcal{D}}$ \quad ($\monM = [0,1]$, &
      $f(a)(p) = \sum\limits_{y\in Y} p(y)\cdot a(y)$ &
      $\{p \in \Ytop{D}{f(a)} \mid \mathit{supp}(p) \subseteq Y'\}$ \\
      $Z = D \subseteq \mathcal{D}(Y)$) & &
      \\
      \hline
    \end{tabular}
  \end{center}
  \caption{\normalsize Basic functions $f\colon \monM^Y\to \monM^Z$
    (constant, reindexing, minimum, maximum, average) and their approximations
    $f^a_\#\colon \Pow{\Ytop{Y}{a}} \to
    \Pow{\Ytop{Z}{f(a)}}$. 
  }
  \label{tab:basic-functions-approximations}
\end{table}

\section{A Categorical View of the Approximation Framework}
\label{sec:categories-functor}

In this section we argue that the framework
from~\cite{BEKP:FTUD-journal}, summarised in the previous section, can
be naturally reformulated in a categorical setting. In particular,
here we study the compositionality properties of the operation 
mapping a function $f$ to its approximation $f_\#^a$ (for a given
fixpoint $a$ of $f$). We show that, under some constraints, it
can be characterised as a functor and, in general, as a union of (lax)
functors.

We will use some standard notions from category theory, in particular
categories, functors and natural transformations. The definition of
(strict) gs-monoidal categories will be spelled out in detail later in
Definition~\ref{def:gs-monoidal}.

We first define a concrete category $\C$ whose arrows are the
non-expansive functions for which we seek the least (or greatest)
fixpoint and a category $\A$ whose arrows are the corresponding
approximations.  Recall that, as discussed in the previous section,
given a non-expansive function $f : \monM^Y \to \monM^Z$, the
approximation of $f$ is relative to a fixed map $a \in \monM^Y$. Hence
objects in $\C$ are intuitively pairs
$\langle \monM^Y, a \in \monM^Y \rangle$ or
$\langle Y, a \in \monM^Y \rangle$. Since $a \in \monM^Y$ determines
$\monM^Y$ as its domain, for simplifying the notation, we leave the
first component implicit and let objects in $\C$ be elements
$a \in \monM^Y$ and an arrow from $a \in \monM^Y$ to $b \in \monM^Z$
is a non-expansive function $f : \monM^Y \to \monM^Z$ required to map
$a$ into $b$ ($f(a)=b$). The approximations instead live in a
different category $\A$. Recall that the approximation $f^a_{\#}$ is
of type $\Pow{[Y]^a} \to \Pow{[Z]^b}$. Since the domain and codomain
are again dependent on maps $a$ and $b$, we still employ elements of
$\monM^Y$ as objects, but arrows are functions between powersets.

\begin{definition}[concrete category and category of approximations]
  \label{def:cats-functor-hash}
  We define two categories, the \emph{concrete category} $\C$ and the
  \emph{category of approximations} $\A$ and a mapping
  $\#\colon \C\to\A$.
  \begin{itemize}
  \item The concrete category $\C$ has as objects maps $a \in\monM^Y$
    where $Y$ is a (possibly infinite) set.
    Given $a \in \monM^Y$, $b \in \monM^Z$ an arrow $f: a \arr b$ is a
    non-expansive function $f\colon \monM^Y\to \monM^Z$, such that
    $f(a)=b$.
  
  \item The category of approximations $\A$ has again maps
    $a \in \monM^Y$ as objects.
    Given $a \in \monM^Y$, $b \in \monM^Z$ an arrow $g: a \arr b$
    is a monotone (with respect to inclusion) function
    $g\colon \Pow{[Y]^a}\to \Pow{[Z]^b}$.
    Arrow composition and identities are the obvious ones.
  
  \item The approximation maps $\#^\delta \colon \C \to \A$ (for
    $\delta \sqsupset 0$) and $\#\colon \C \to \A$ are defined as
    follows: for an object $a\in \monM^Y$, we let
    $\#(a) = \#^\delta(a) = a$ and, given an arrow $f : a \arr b$, we
    let $\#^\delta(f) = f_\#^{a,\delta}$ and
    $\#(f) = \bigcup_{\delta\sqsupset 0} \#^\delta(f) = f_\#^a$.
  \end{itemize}
\end{definition}

Note that categorical arrows are represented as dashed ($\arr$). By
definition each such categorical arrow consists of an underlying
function whose domain and codomain are sets. The underlying functions,
being set-valued, are represented as usual (using the notation $\to$).

\begin{lemmarep}[well-definedness]
  \label{le:cate}
  The categories $\C$ and $\A$ are well-defined and the $\#^\delta$
  are lax functors, i.e., identities are preserved and
  $\#^\delta(f) \circ \#^\delta(g)\subseteq \#^\delta(f\circ g)$ for
  composable arrows $f,g$ in $\C$.
\end{lemmarep}

\begin{proof}
  \mbox{}

  \begin{enumerate}
  \item \emph{$\C$ is a well-defined category}: Given arrows
    $f\colon a\arr b$ and $g\colon b\arr c$ then $g\circ f$ is
    non-expansive (since non-expansiveness is preserved by composition) and
    $(g\circ f)(a) = g(b) = c$, thus $g\circ f\colon a\arr c$.
    Associativity holds and the identities are the units of composition
    as for standard function composition.

  \item \emph{$\A$ is a well-defined category}: Given arrows
    $f\colon a\arr b$ and $g\colon b\arr c$ then $g\circ f$ is
    monotone (since monotonicity is preserved by composition) and hence
    $g\circ f\colon a\arr c$.
 
    Again associativity and the fact that the identities are units is
    immediate. 

  \item \emph{$\#^\delta\colon \C\to \A$ is a lax functor}: we first
    check that identities are preserved. Let $U\subseteq [Y]^a$, then
    \begin{align*}
      \#^\delta(\mathit{id}_a)(U) &=
      (\mathit{id}_a)^{a,\delta}_\#(U) \\
      &= \{ y\in [Y]^{\mathit{id}_a(a)}
      \mid \mathit{id}_a(a)(y)
      \ominus \mathit{id}_a(a\ominus \delta_{U})(y)  \sqsupseteq \delta\}\\
      &= \{ y\in [Y]^{a} \mid a(y)
      \ominus
      (a\ominus \delta_{U})(y) \sqsupseteq \delta\}\\
      &=U = \mathit{id}_{a}(U) = \mathit{id}_{\#^\delta(a)}(U).
    \end{align*}
    where in the second last line we use the fact that $U\subseteq [Y]^a$.
    
    Let $a\in \monM^Y$, $b\in\monM^Z$, $c\in \monM^V$,
    $f\colon a\arr b$, $g\colon b\arr c$ be arrows in $\C$ and
    $Y'\subseteq [Y]^a$. Then
    \begin{align*}
      (\#^\delta(g) \circ \#^\delta(f))(Y') 
      &= g_\#^{b,\delta} ( f_\#^{a,\delta}(Y')) \\
      &= (\gamma^{c,\delta} \circ g \circ \alpha^{b,\delta}
      \circ \gamma^{b,\delta} \circ f \circ
      \alpha^{a,\delta})(Y') \\
      &\subseteq (\gamma^{g(f(a)),\delta} \circ g
      \circ f \circ \alpha^{a,\delta})(Y') \\
      &=
      (g\circ f)_\#^{a,\delta}(Y') \\
      &= \#^\delta(g\circ f) (Y')
   \end{align*}
  The inequality holds since for $c\in\monM^Z$:
  \[ \alpha^{b,\delta}(\gamma^{b,\delta}(c))
    = \alpha^{b,\delta}(\{y\in Y\mid b(y)\ominus c(y)\sqsupseteq
    \delta\}) = b\ominus \delta_{\{y\in Y\mid b(y)\ominus c(y)\sqsupseteq
      \delta\}} \sqsupseteq c. \]
  Then the inequality follows from the
  antitonicity of $\gamma^{c,\delta}$. (Remember that we are working
  with a contra-variant Galois connection.) \qedhere
\end{enumerate}
\end{proof}

Note that while $\#$ clearly also preserves identities, the question
whether it is a lax functor (or even a proper functor) is currently
open. It is however the union of lax functors.

We next observe that $\#$ is a functor if we restrict to suitable
subcategories of $\mathbb{C}$, i.e., the subcategory where arrows are
reindexings and the one where objects are maps on finite sets.
This allows to recover compositionality in those cases in which it is
required by the intended applications (see
Sections~\ref{sec:liftings}-\ref{sec:gs-mon}).

\begin{definition}[reindexing subcategory]
  We denote by $\C^*$ the sub-category of $\C$ that contains all
  objects and where arrows are restricted to reindexings, i.e., given
  objects $a \in \monM^Y$, $b \in \monM^Z$ we consider only arrows
  $f : a \arr b$ such that $f = g^*$ for some $g\colon Z\to Y$ (hence,
  in particular, $b = g^*(a) = a\circ g$).
\end{definition}

We prove the following auxiliary lemma that basically shows that
reindexings are preserved by $\alpha,\gamma$:

\begin{lemmarep}
  \label{lem:fibredness}
  Given $a\in \monM^Y$, $g\colon Z \to Y$ and
  $0\sqsubset \delta \in \monM$, then we have
  \begin{enumerate}
  \item
    $\alpha^{a\circ g,\delta} \circ g^{-1} = g^*\circ
    \alpha^{a,\delta}$
  \item
    $\gamma^{a\circ g,\delta} \circ g^* = g^{-1}\circ
    \gamma^{a,\delta}$
  \end{enumerate}
  This implies that for two $\C$-arrows $f\colon a\arr b$,
  $h\colon b\arr c$, it holds that $\#(h\circ f) = \#(h)\circ \#(f)$
  whenever $f$ or $h$ is a reindexing, i.e., is contained in
  $\C^*$. 
\end{lemmarep}

\begin{proof}
  ~
  \begin{enumerate}
  \item Let $Y'\subseteq \Ytop{Y}{a}$. Then
    \begin{align*}
      g^*(\alpha^{a,\delta}(Y')) &= g^*(a\ominus \delta_{Y'}) = (a\ominus
      \delta_{Y'})\circ g = a\circ g \ominus \delta_{Y'}\circ g \\
      &= a\circ g \ominus \delta_{g^{-1}(Y')} = \alpha^{a\circ
        g,\delta}(g^{-1}(Y')) 
    \end{align*}
    where we use that $(\delta_{Y'}\circ g)(z) = \delta$ if $g(z)\in
    Y'$, equivalent to $z\in g^{-1}(Y')$, and $0$ otherwise. Hence
    $\delta_{Y'}\circ g = \delta_{g^{-1}(Y')}$.
  \item Let $b\in \monM^Y$ with $a\ominus \delta\sqsubseteq b\sqsubseteq
    a$. Then
    \begin{align*}
      \gamma^{a\circ g,\delta}\circ g^*(b) &=
      \{z\in Z\mid a(g(z))\ominus
      b(g(z)) \sqsupseteq \delta \} = \{z\in Z\mid
      g(z)\in\gamma^{a,\delta}(b)\} \\
      &= g^{-1}(\gamma^{a,\delta}(b))
    \end{align*}

  \end{enumerate}
  It is left to show that $\#(h\circ f) = \#(h)\circ \#(f)$ whenever
  $f$ or $h$ is a reindexing. Recall that on reindexings it holds that
  $\#(g^*) = g^{-1}$.

  Let $a\in\monM^Y$, $b\in\monM^Z$, $c\in\monM^W$ and assume first
  that $f$ is a reindexing, i.e., $f=g^*$ for some $g\colon Z\to
  Y$. Let $Y'\subseteq [Y]^a$, then
  \begin{align*}
    \#(h\circ f) &= (h\circ f)_\#^a = \bigcup_{\delta\sqsupset 0} 
    (\gamma^{h(f(a)),\delta}\circ h \circ f
    \circ \alpha^{a,\delta})(Y') \\
    &= \bigcup_{\delta\sqsupset 0} 
    (\gamma^{h(f(a)),\delta}\circ h \circ g^*
    \circ \alpha^{a,\delta})(Y') \\
    &= \bigcup_{\delta\sqsupset 0} 
    (\gamma^{h(f(a)),\delta}\circ h \circ \alpha^{a\circ g,\delta})(g^{-1}(Y')) && (1) \\
    &= \bigcup_{\delta\sqsupset 0} 
    (\gamma^{h(f(a)),\delta}\circ h \circ \alpha^{f(a),\delta})(\#(g^*)(Y')) \\
    &= (\#(h)\circ \#(f))(Y')
  \end{align*}
  Now we assume that $h$ is a reindexing, i.e., $h=g^*$ for some
  $g\colon W\to Z$. Let again $Y'\subseteq [Y]^a$, then:
  \begin{align*}
    \#(h\circ f) &= (h\circ f)_\#^a = \bigcup_{\delta\sqsupset 0} 
    (\gamma^{h(f(a)),\delta}\circ h \circ f
    \circ \alpha^{a,\delta})(Y') \\
    &= \bigcup_{\delta\sqsupset 0} 
    (\gamma^{f(a)\circ g,\delta}\circ g^* \circ f
    \circ \alpha^{a,\delta})(Y') \\
    &= \bigcup_{\delta\sqsupset 0} g^{-1}(
    (\gamma^{f(a),\delta} \circ f
    \circ \alpha^{a,\delta})(Y'))  &&\text{(2)}\\
    &= g^{-1}(\bigcup_{\delta\sqsupset 0} 
    (\gamma^{f(a),\delta} \circ f
    \circ \alpha^{a,\delta})(Y')) &&\text{[preimage preserves union]} \\
     &= \#(g^*)(\bigcup_{\delta\sqsupset 0} 
    (\gamma^{f(a),\delta} \circ f
    \circ \alpha^{a,\delta})(Y')) \\
     &= (\#(h)\circ \#(f))(Y') \tag*{\qedhere}
    \end{align*}
\end{proof}

Then, as an immediate corollary, we obtain the functoriality of $\#$
in the corresponding subcategory.

\begin{corollary}[approximation functor for reindexing categories]
  \label{cor:reindexing-functor}
  The approximation map $\#\colon \C\to \A$ restricts to
  $\#\colon \C^* \to \A$, which is a (proper) functor.
\end{corollary}

We next focus on the subcategory of $\C$ where we consider as objects
only maps $a : Y \to \monM$ over a finite set $Y$.

\begin{definition}[finitary subcategories]
  We denote by $\Cf$,
  $\Af$ the full sub-categories of
  $\C,\A$ where objects are of the kind
  $a \in \monM^Y$ for a \emph{finite} set $Y$.
\end{definition}

\begin{lemmarep}[approximation functor for finitary categories]
  \label{lem:lax-functor-proper-finite}
  The approximation map $\#\colon \C\to \A$ restricts to
  $\# \colon \Cf\to\Af$, which is a (proper) functor.
\end{lemmarep}

\begin{proof}
  Clearly the restriction to categories based on finite sets is well-defined.
  
  We show that $\#$ is a (proper) functor. Let $a\in\monM^Y$,
  $b\in\monM^Z$, $c\in\monM^V$, $f\colon a\arr b$, $g\colon b\arr c$
  and $Y'\subseteq [Y]^a$. Then
  \begin{align*}
    \#(g\circ f) = (g\circ f)_\#^{a} = g_\#^{f(a)} \circ f_\#^a =
    g_\#^b \circ f_\#^a = \#(g) \circ \#(f),
  \end{align*}
  The second inequality above is a consequence of the compositionality
  result in~\cite[Proposition~D.3]{BEKP:FTUD-journal}. This
  requires finiteness of the sets $Y,Z,V$.
  
  The rest follows from Lemma~\ref{le:cate}.
\end{proof}

We will later show in Theorem~\ref{thm:gs-monoidal-functor} that $\#$
also respects the monoidal operation $\otimes$ (disjoint union of
functions). Using this, we can now apply the framework
to an example.

\begin{example}
  \label{ex:termination-mc-comp}
  We revisit the example from Section~\ref{sec:motivation}. Remember
  that the function $\mathcal{T}$, whose least fixpoint is termination
  probability, can be written as follows:
  \[ \mathcal{T}= (\eta^*\circ \tilde{\mathcal{D}}) \otimes c_k\] 

  Let $a=\nu\mathcal{T}$ be the greatest fixpoint of $\mathcal{T}$,
  i.e., $\mathcal{T}(a)=a$. Hence we can view
  $\mathcal{T}\colon a\arr a$ as a concrete arrow in $\C$ in the sense
  of Definition~\ref{sec:categories-functor}. Furthermore
  $\#\mathcal{T} = \mathcal{T}_\#^a\colon a\arr a$ is the
  corresponding approximation, living in the category of
  approximations $\A$.
  
  We can compute $\mathcal{T}_\#^a$ compositionally. The subfunctions
  of $\mathcal{T}$ given above are also arrows in $\Cf$ for
  appropriate finite domains and codomains, which we refrain from
  spelling out explicitly.
  Now:
  \[ \mathcal{T}_\#^a = \#\mathcal{T} = \#((\eta^*\circ \tilde{\mathcal{D}})
    \otimes c_k) = \#(\eta^*)\circ \#(\tilde{\mathcal{D}}) \otimes
    \#(c_k) \]
\end{example}

This view enables us to obtain an approximation $\mathcal{T}_\#^a$
compositionally out of the approximations of the subfunctions.

\section{Predicate Liftings}
\label{sec:liftings}

In this section we show how predicate
liftings~\cite{p:coalgebraic-logic,s:coalg-logics-limits-beyond-journal}
can be integrated into our theory.
We will characterise predicate liftings which are non-expansive and
derive their approximations. This will then be used in
Section~\ref{sec:beh-metrics} for treating coalgebraic behavioural
metrics.

\subsection{Predicate Liftings and their Properties}

Roughly speaking, predicates are seen as maps from a set $Y$ to a
suitable set of truth values $\mathcal{V}$. Then, given a functor $F$,
a predicate lifting is an operation which transforms predicates over
$Y$ to predicates over $F Y$.

One of the simplest examples of a predicate lifting is given by
the diamond ($\Diamond$) operator from modal logic. In this case
$F=\mathcal{P}$ (powerset functor) and $\mathcal{V}=\{0,1\}$. Given a
predicate $q\colon Y\to \{0,1\}$,
this is mapped to $\Diamond(q)\colon \mathcal{P}(Y)\to\{0,1\}$ where
$\Diamond(q)(Y')=1$ iff there exists $y\in Y'$ with $q(y)=1$.
Another typical example is expectation where $F=\mathcal{D}$ (distribution functor) and $\mathcal{V}=[0,1]$
(see also
Example~\ref{ex:distribution-predlift} below). In this case a random
variable $r\colon Y\to[0,1]$ is mapped to a function of type
$\mathcal{D}(Y)\to [0,1]$ where $p\mapsto \mathbb{E}_p[r]$.

Predicate liftings have been studied for predicates valued over
arbitrary quantales $\mathcal{V}$ (see,
e.g.,~\cite{bkp:up-to-behavioural-metrics-fibrations}), i.e., complete
lattices with an associative operator that distributes over arbitrary
joins.
It can be shown that every complete MV-algebra is a quantale with
respect to $\oplus$ and the inverse of the natural order. This result
can be easily derived from~\cite{di05:_algeb_lukas}. (See
Lemma~\ref{le:quantale} in the appendix for an explicit proof.) Hence
here we can work with predicates of the kind $Y \to \monM$ where
$\monM$ is a complete MV-algebra.

\begin{toappendix}
  \begin{lemma}[complete MV-algebras are quantales]
  \label{le:quantale}
  Let $\monM$ be a complete MV-algebra. Then
  $(\monM, \oplus, \sqsupseteq)$ is a unital commutative quantale,
  i.e., a quantale with neutral element for $\oplus$.
\end{lemma}

\begin{proof}
  We know $\monM$ is a complete lattice. Binary meets are given by
  \begin{equation}
    \label{eq:meet}
    x \sqcap y = \comp{\comp{x \oplus \comp{y}} \oplus \comp{y}}.
  \end{equation}
  Moreover $\oplus$ is
  associative and commutative, with $0$ as neutral element.
  
  It remains to show that $\oplus$ distributes with respect to
  $\sqcap$ (note that $\sqcap$ is the join for the reverse order),
  i.e., that for all $X \subseteq \monM$ and $a \in \monM$, it holds
  \[
    a \oplus \bigsqcap X = \bigsqcap \{ a \oplus x \mid x  \in X \}
  \]
  Clearly, since $\bigsqcap X \leq x$ for all $x \in X$ and $\oplus$
  is monotone, we have
  $a \oplus \bigsqcap X \sqsubseteq \bigsqcap \{ a \oplus x \mid x \in
  X \}$. In order to show that $a \oplus \bigsqcap X$ is the greatest
  lower bound, let $z$ be another lower bound for
  $\{ a \oplus x \mid x \in X \}$, i.e., $z \sqsubseteq a \oplus x$
  for all $x \in X$. Then observe that for $x \in X$, using
  \eqref{eq:meet}, we get
  \begin{center}
    $x \sqsupseteq x \sqcap \comp{a} = \comp{\comp{(x \oplus a)}
      \oplus a} \sqsupseteq \comp{\comp{z} \oplus a} = z \ominus a$
  \end{center}
  Therefore $\bigsqcap X \sqsupseteq z \ominus a$ and thus
  \begin{center}
    $a \oplus  \bigsqcap X \sqsupseteq a \oplus (z \ominus a) \sqsupseteq z$
  \end{center}
  as desired.
\end{proof}
\end{toappendix}

\begin{definition}[predicate lifting]
  Given a functor $F\colon\mathbf{Set}\to\mathbf{Set}$, a 
  \emph{predicate lifting} is a family of functions
  $\tilde{F}_Y\colon \monM^Y \to \monM^{FY}$ (where $Y$ is a set),
  such that for all $g\colon Z\to Y$, $a\colon Y \to \monM$ it holds that
  $(Fg)^*(\tilde{F}_Y(a))  = \tilde{F}_Z(g^*(a))$. 
\end{definition}

In words, predicate liftings must commute with reindexings. The index
$Y$ will be omitted if clear from the context. 
It can be seen that such predicate liftings are in one-to-one
correspondence to so called \emph{evaluation maps}
$\mathit{ev} \colon F\monM \to \monM$.  Given $\mathit{ev}$, we define
the corresponding lifting to be
$\tilde{F}(a) = \mathit{ev}\circ Fa\colon FY \to \monM$, where
$a\colon Y\to \monM$. Conversely, given a lifting $\tilde{F}$, we
obtain $\mathit{ev} = \tilde{F}(\mathit{id}_\monM)$.

A lifting $\tilde{F}$ is well-behaved if (i) $\tilde{F}$ is monotone;
(ii) $\tilde{F}(0_Y) = 0_{FY}$ where $0$ is the constant $0$-function;
(iii)
$\tilde{F}(a \oplus b) \sqsubseteq \tilde{F}(a) \oplus \tilde{F}(b)$
for $a,b\colon Y\to \monM$; (iv) $F$ preserves weak pullbacks. We need
well-behavedness in order to prove the next result and in
Section~\ref{sec:beh-metrics}. 

In order to use the theory of fixpoint checks
from~\cite{BEKP:FTUD-journal} we need to have not only monotone, but
non-expansive liftings. We next provide a characterisation of such
liftings, following \cite{ws:metrics-fuzzy-lax-journal}.

\begin{lemmarep}[non-expansive predicate lifting]
  \label{lem:pred-lifting-nonexpansive}
  Let $\mathit{ev}\colon F\monM\to \monM$ be an evaluation map and
  assume that its corresponding lifting
  $\tilde{F}\colon \monM^Y\to \monM^{FY}$ is well-behaved. Then
  $\tilde{F}$ is non-expansive iff for all
  $\delta \in \monM$ it holds that
  $\tilde{F}\delta_{Y} \sqsubseteq
  \delta_{FY}$.
\end{lemmarep}

\begin{proof}
  The proof is inspired
  by~\cite[Lemma~3.9]{ws:metrics-fuzzy-lax-journal} and uses the fact
  that a monotone function $f\colon \monM^Y\to\monM^Z$ is
  non-expansive iff $f(a\oplus \delta)\sqsubseteq f(a)\oplus \delta$
  for all $a,\delta$.

  \smallskip
  \noindent
  ($\Rightarrow$) Fix a set $Y$ and assume that
  $\tilde{F}\colon \monM^Y\to \monM^{FY}$ is non-expansive.  Then
  \[ \tilde{F}(\delta) = \tilde{F}(0\oplus \delta) \sqsubseteq
    \tilde{F}(0)\oplus \delta \sqsubseteq 0 \oplus \delta = \delta \]

  \smallskip
  \noindent
  ($\Leftarrow$) Now assume that $\tilde{F}(\delta) \sqsubseteq
  \delta$. Then, using the lemma referenced above,
  \[ \tilde{F}(a\oplus \delta) \sqsubseteq \tilde{F}(a)\oplus
    \tilde{F}(\delta) \sqsubseteq \tilde{F}(a) \oplus \delta \]

  Above we write $\delta$ for both $\delta_Y$, $\delta_{FY}$
  and both deductions rely on the fact that $\tilde{F}$ is
  well-behaved.
\end{proof}

\begin{example}[Finitely supported distributions]
  \label{ex:distribution-predlift}
  Consider the (finitely supported)
  distribution functor $\mathcal{D}$ that maps a set $X$ to all maps
  $p\colon X\to [0,1]$ that have finite support and satisfy
  $\sum_{x\in X} p(x) = 1$. (Here $\monM = [0,1]$.)  A possible
  evaluation map is $\mathit{ev}\colon \mathcal{D}[0,1]\to [0,1]$
  defined by $\mathit{ev}(p) = \sum_{r\in[0,1]} r\cdot p(r)$, where
  $p$ is a distribution on $[0,1]$. This results in the predicate
  lifting $\tilde{D}_Y\colon [0,1]^Y \to [0,1]^{\mathcal{D}(Y)}$ with
  $\tilde{\mathcal{D}}_Y(r)(p) = \mathbb{E}_p[r]$ 
  discussed at the beginning of the section (expectation).

  It is easy to see that $\tilde{\mathcal{D}}$ is well-behaved and
  non-expansive. The latter follows from
  $\tilde{\mathcal{D}}(\delta_Y) = \delta_{\mathcal{D}Y}$.
\end{example}

\begin{example}[Finite powerset]
  \label{ex:powerset-predlift}
  Consider the finite powerset functor
  $\mathcal{P}_f$ with  the evaluation map
  $\mathit{ev} \colon \mathcal{P}_f \monM \to \monM$, defined for
  finite $S\subseteq \monM$ as $\mathit{ev} (S) = \max S$, where
  $\max \emptyset = 0$. This results in the predicate lifting
  $\tilde{P}_Y\colon [0,1]^Y \to [0,1]^{\mathcal{P}(Y)}$,
  $\mathcal{P}_Y(p)(Y') = \bigsqcup_{y\in Y'} p(y)$.

  The lifting $\tilde{\mathcal{P}}_f$ is well-behaved
  (see~\cite{bbkk:coalgebraic-behavioral-metrics}) and
  non-expansive. To show the latter, observe that
  $\tilde{\mathcal{P}}_f(\delta_Y) =
  \delta_{\mathcal{P}_f(Y)\backslash\{\emptyset\}} \sqsubseteq
  \delta_{\mathcal{P}_f(Y)}$.
\end{example}

Non-expansive predicate liftings can be seen as functors
$\tilde{F}\colon \C^* \to \C^*$. To be more precise,
$\tilde{F}$ maps an object $a\in \monM^Y$ to
$\tilde{F}(a) \in \monM^{FY}$ and an arrow
$g^* : a \arr a\circ g$, , where $g\colon Z\to Y$, to
$(Fg)^* : \tilde{F}a \arr \tilde{F}(a \circ g)$.

\subsection{Approximations of Predicate Liftings}
\label{sec:approximation-pred-lifting}

We now study approximations of predicate liftings. It involves the
approximation functor $\#\colon \C^*\to\A$ (restricted to $\C^*$) and
the predicate lifting $\tilde{F}\colon \C^*\to \C^*$ introduced
before. We start with a result about an auxiliary natural transformation.

\begin{propositionrep}
  \label{prop:nat-transf}
  Let $\tilde{F}$ be a (non-expansive) predicate lifting.  There is a
  natural transformation $\beta \colon \# \Rightarrow \# \tilde{F}$
  between functors $\#,\#\tilde{F}\colon \C^*\to \A$, whose
  components, for $a\in \monM^Y$, are
  $\beta_a\colon a \arr \tilde{F}(a)$ in $\A$, defined by
  $\beta_a(U) = \tilde{F}^{a}_\#(U)$ for $U\subseteq [Y]^a$.

  That is, the following diagrams commute for every $g\colon Z \to Y$
  (the diagram on the left indicates the formal arrows, while the one
  on the right reports the underlying functions).
  \begin{center}
    \scalebox{1}{
      \begin{tikzpicture}	
        \node (1) at (-4,1.5) [] {\begin{tabular}{c}
            $\#(a)$
          \end{tabular} };
        
        \node (2) at (0,1.5) [] {\begin{tabular}{c}
            $\#(a\circ g)$
          \end{tabular} };
        \node (3) at (-4,0) [] {\begin{tabular}{c}
            $\#(\tilde{F}a)$
          \end{tabular} };
        \node (4) at (0,0) [] {\begin{tabular}{c}
            $\#(\tilde{F}(a\circ g))$
          \end{tabular} };
        
        \draw[->,dashed] (1) to node [above]  {$\#(g^*) $} (2);
        \draw[->,dashed] (1) to node [left]  {$\beta_a$} (3);
        \draw[->,dashed] (2) to node [right]
        {$\beta_{a\circ g}$} (4);
        \draw[->,dashed] (3) to node [above]  {$\#(\tilde{F}(g^*))  $} (4);
        
      \end{tikzpicture}
      \hspace{1cm}
      \begin{tikzpicture}	
        \node (1) at (-3,1.5) [] {\begin{tabular}{c}
            $\Pow{[Y]^a}$
          \end{tabular} };
        
        \node (2) at (1,1.5) [] {\begin{tabular}{c}
            $\Pow{[Z]^{a\circ g}}$
          \end{tabular} };
        \node (3) at (-3,0) [] {\begin{tabular}{c}
            $\Pow{[FY]^{\tilde{F}(a)}}$
          \end{tabular} };
        \node (4) at (1,0) [] {\begin{tabular}{c}
            $\Pow{[FZ]^{\tilde{F}(a\circ g)}}$
          \end{tabular} };
        
        \draw[->] (1) to node [above]  {$g^{-1}$} (2);
        \draw[->] (1) to node [left]  {$\tilde{F}^{a}_\#$} (3);
        \draw[->] (2) to node [right]
        {$\tilde{F}^{a\circ g}_\#$} (4);
        \draw[->] (3) to node [above]  {$(Fg)^{-1}$} (4);
        
      \end{tikzpicture}
    }
  \end{center}
\end{propositionrep}

\begin{proof}
  We first define a natural transformation
  $\eta\colon \mathrm{Id}_{\C^*}\Rightarrow \tilde{F}$ (between the
  identity functor and $\tilde{F}$) with components
  $\eta_a\colon a\arr \tilde{F}(a)$ (for $a\in\monM^Y$) by defining
  $\eta_a(b) = \tilde{F}(b)$ for $b\in\monM^Y$. The $\eta_a$ are
  non-expansive by assumption. In addition, $\eta$ is natural due to
  the definition of a predicate lifting, i.e.,
  $(Fg)^*\circ \tilde{F} = \tilde{F}\circ g^*$ for $g\colon Z\to Y$.

  Now we apply $\#$ and use the fact that $\#$ is functorial, even for
  the full category $\C$, whenever one of the two arrows to which $\#$
  is applied is a reindexing (see
  Lemma~\ref{lem:fibredness}). Furthermore we observe that
  $\beta = \#(\eta)$.  This immediately gives commutativity of the
  diagram on the left. (The diagram on the right just displays the
  underlying functions.)
\end{proof}

We next characterise $\tilde{F}^d_\#(Y')$. We rely on the fact that
$d$ can be decomposed into $d = \pi_1\circ\bar{d}$, where the
projection $\pi_1$ is independent of $d$ and $\bar{d}$ is dependent on
$Y'$, and exploit the natural transformation in
Proposition~\ref{prop:nat-transf}.

\begin{propositionrep}[Approximations for predicate liftings]
  \label{prop:char-abstr-predlift}
  Let $\tilde{F}$ be a predicate lifting. We fix $Y'\subseteq Y$ and
  let $\chi_{Y'}\colon Y\to \{0,1\}$ be its characteristic
  function. Furthermore let $a\colon Y\to \monM$ be a predicate. Let
  $\pi_1\colon \monM\times\{0,1\} \to \monM$,
  $\pi_2\colon \monM\times\{0,1\} \to \{0,1\}$ be the projections and
  define $\bar{a}\colon Y\to \monM\times \{0,1\}$ via
  $\bar{a}(y) = (a(y),\chi_{Y'}(y))$ as the mediating morphism into
  the product (see diagram below).

  \begin{center}
    \begin{tikzpicture}	
      \node (Y) at (0,0) [] {$Y$};
      \node (B) at (2,0) [] {$\{0,1\}$};
      \node (M) at (-2,0) [] {$\monM$};
      \node (C) at (0,-2) [] {$\monM\times \{0,1\}$};
      
      \draw[->] (Y) to node [above]  {$a$} (M);
      \draw[->] (Y) to node [above]  {$\chi_{Y'}$} (B);
      \draw[->] (Y) to node [right] {$\bar{a}$} (C);
      \draw[->] (C) to node [left]  {$\pi_1$} (M);
      \draw[->] (C) to node [right]  {$\pi_2$} (B);
    \end{tikzpicture}
  \end{center}

  Then
  \[ \tilde{F}_\#^{a}(Y') = (F\bar{a})^{-1}
    (\tilde{F}_\#^{\pi_1}((\monM\backslash\{0\})\times \{1\})). \]
\end{propositionrep}

\begin{proof}
  Let $a\in \monM^Y$ and $Y'\subseteq \Ytop{Y}{a}$. Note that
  $\bar{a}^{-1}((\monM\backslash\{0\})\times \{1\}) = Y'$ and
  $a=\pi_1\circ \bar{a}$, thus by Proposition~\ref{prop:nat-transf}:
  \begin{align*}
    \tilde{F}_\#^a(Y') 
    &= \tilde{F}_\#^{\pi_1\circ \bar{a}}(\bar{a}^{-1}((\monM\backslash\{0\})\times \{1\})) \\
    &= (F\bar{a})^{-1} (\tilde{F}_\#^{\pi_1}( (\monM\backslash\{0\})\times \{ 1\})) \tag*{\qedhere}
  \end{align*}
\end{proof}

Here
$\tilde{F}_\#^{\pi_1}((\monM\backslash\{0\})\times \{1\})\subseteq
F(\monM\times \{0,1\})$ is independent of $a$ and has to be determined
only once for every predicate lifting $\tilde{F}$. We will show how
this set looks like for our example functors. We first consider the
distribution functor.

\begin{lemmarep}
  \label{lem:char-hash-distr}
  Consider the lifting of the distribution functor presented in
  Example~\ref{ex:distribution-predlift} and let $\monM = [0,1]$. Then we have
  \[ \tilde{\mathcal{D}}_\#^{\pi_1}( (0,1]\times \{1\} ) = \{ p\in \mathcal{D}Z
    \mid \mathit{supp}(p) \in (0,1]\times \{1\} \}. \]
\end{lemmarep}

\begin{proof}
  Let $\delta > 0$.
  We define
  \[ \tilde{\pi}^\delta_1 := \alpha^{\pi_1,\delta}((0,1]\times
    \{1\}) \] where $\tilde{\pi}^\delta_1(x,0) = x$,
  $\tilde{\pi}^\delta_1 (x,1) = x\ominus \delta$ for $x\in
  [0,1]$. Note that
  $[\mathcal{D}Z]^{\tilde{\mathcal{D}}\pi_1} = \{ p\in DZ \mid \exists
  (x,b)\in \mathit{supp}(p) \text{ with } x \ge 0 \}$. Now
  \begin{align*}
    \tilde{\mathcal{D}}_\#^{\pi_1,\delta}((0,1]\times\{1\}) &= \{ p\in [\mathcal{D}Z]^{\tilde{\mathcal{D}}\pi_1} \mid \tilde{\mathcal{D}}\pi_1 (p) \ominus \tilde{\mathcal{D}}(\tilde{\pi}^\delta_1) (p)\ge \delta \} \\
    &= \{ p\in [\mathcal{D}Z]^{\tilde{\mathcal{D}}\pi_1} \mid \bigl( \sum_{x\in [0,1]} x\cdot p(x,0) \oplus \sum_{x\in [0,1]} x\cdot p(x,1) \bigr)\\
    &\qquad\quad\ominus \bigl( \sum_{x\in [0,1]} x\cdot p(x,0) \oplus \sum_{x\in [0,1]} (x\ominus \delta)\cdot p(x,1) \bigr)\ge \delta \}  \\
    &= \{ p\in [\mathcal{D}Z]^{\tilde{\mathcal{D}}\pi_1} \mid
    \sum_{x\in[0,\delta)} x\cdot p(x,1) + \sum_{x\in[\delta,1]} \delta
    \cdot p(x,1) \ge \delta \}  \\
    &= \{ p\in [\mathcal{D}Z]^{\tilde{\mathcal{D}}\pi_1} \mid
    \mathit{supp}(p) \in [\delta,1]\times \{1\} \}.
  \end{align*}
  Where the second last equality uses the fact that $x\ominus
  (x\ominus \delta) = \delta$ if $x\ge \delta$ and $x$ otherwise.
  
  Now, we obtain
  \begin{align*} \tilde{\mathcal{D}}_\#^{\pi_1}( (0,1]\times \{ 1\} ) = \bigcup_{\delta
      \sqsupset 0} \tilde{\mathcal{D}}_\#^{\pi_1,\delta}( (0,1]\times \{ 1\} ) =
    \{ p\in \mathcal{D}Z \mid \mathit{supp}(p) \in (0,1]\times \{1\}
    \}. \tag*{\qedhere}\end{align*}   
\end{proof}

This means intuitively that a decrease or ``slack'' can exactly be
propagated for elements whose probabilities are strictly larger than
$0$.

We now turn to the powerset functor.

\begin{lemmarep}
  \label{lem:char-hash-powerset}
  Consider the lifting of the finite powerset functor from
  Example~\ref{ex:powerset-predlift} with arbitrary $\monM$. Then we
  have
  \[ (\tilde{\mathcal{P}}_f)_\#^{\pi_1}((\monM\backslash\{0\})\times \{1\}) = \{
    S\in [\mathcal{P}_f Z]^{\tilde{\mathcal{P}}_f\pi_1} \mid \exists
    (s,1)\in S~ \forall (s',0)\in S:~s\sqsupset s' \}. \]
\end{lemmarep}

\begin{proof}
  Let $\delta \sqsupset 0$ and define $\tilde{\pi}_1^\delta$ as in the
  proof of Lemma~\ref{lem:char-hash-distr}. Then
  \begin{align*}
    (\tilde{\mathcal{P}}_f)_\#^{\pi_1,\delta}((\monM\backslash\{0\})\times\{1\})
    &= \{ S\in [\mathcal{P}_f Z]^{\tilde{\mathcal{P}}_f\pi_1} \mid \tilde{\mathcal{P}}_f\pi_1 (S) \ominus \tilde{\mathcal{P}}_f(\tilde{\pi}^\delta_1) (S)\sqsupseteq \delta \} \\
    &= \{ S\in [\mathcal{P}_f Z]^{\tilde{\mathcal{P}}_f \pi_1} \mid \max_{(s,b)\in S} s \ominus (\max_{(s,b)\in S} s\ominus b\cdot \delta)\sqsupseteq \delta \}  \\
    &= \{ S\in [\mathcal{P}_f Z]^{\tilde{\mathcal{P}}_f \pi_1} \mid \exists (s,1)\in S~\forall (s',0)\in S:~s\ominus \delta \sqsupseteq s' \}
  \end{align*}
  For the last step we note that this condition ensures that the
  second maximum equates to $\max_{(s,b)\in S} s \ominus \delta$
  which is required for the inequality to hold. Now, we obtain
  \begin{align*}
    (\tilde{\mathcal{P}}_f)_\#^{\pi_1}( (\monM\backslash\{0\})\times \{ 1\}) &=
    \bigcup_{\delta \sqsupset 0} \tilde{F}_\#^{\pi_1,\delta}(
    (\monM\backslash\{0\})\times \{ 1\}) \\
    &= \{ S\in
    [\mathcal{P}_f Z]^{\tilde{\mathcal{P}}_f \pi_1} \mid \exists (s,1)\in
    S~\forall (s',0)\in S:~s\sqsupset s' \}. \tag*{\qedhere}
  \end{align*}
\end{proof}

The idea is that the maximum of a set $S$ decreases if we decrease at
least one its values and all values which are not decreased are
strictly smaller.

\begin{remark}
  \label{rem:embed}
  Note that $\#$ is a functor on the subcategory $\Cf$ (see
  Lemma~\ref{lem:lax-functor-proper-finite}), while some liftings
  (e.g., the one for the distribution functor) involve infinite sets,
  for which we would lose compositionality.  In this case, given a
  finite set $Y$, we will actually focus on a finite subset
  $D\subseteq FY$. (This is possible since we work with coalgebras with
  finite state space that map only into finitely many elements of
  $Y$.)  Then we consider $\tilde{F}_Y \colon \monM^Y\to \monM^{FY}$
  and $e\colon D\hookrightarrow FY$ (the embedding of $D$ into
  $FY$). We set $f = e^*\circ \tilde{F}_Y\colon \monM^Y\to
  \monM^D$. Given $a\colon Y\to\monM$, we view $f$ as an arrow
  $a\arr \tilde{F}_Y(a)\circ e$ in $\C$.  The approximation adapts to
  the ``reduced'' lifting, which can be seen as follows
  (cf. Lemma~\ref{lem:fibredness}, which shows that $\#$ preserves
  composition if one of the arrows is a reindexing):
  \[ f_\#^a = \#(f) = \#(e^*\circ \tilde{F}_Y) = \#(e^*)\circ
    \#(\tilde{F}_Y) = e^{-1}\circ \#(\tilde{F}_Y) = \#(\tilde{F}_Y) \cap
    D. \]
\end{remark}

\section{Wasserstein Lifting and Behavioural Metrics}
\label{sec:beh-metrics}

In this section we use the results about predicate liftings from the
previous section to show how the framework for fixpoint checking can
be used to deal with coalgebraic behavioural metrics.

We build on~\cite{bbkk:coalgebraic-behavioral-metrics}, where an
approach is proposed for canonically defining a behavioural
pseudo-metric for coalgebras of a functor
$F\colon\mathbf{Set}\to\mathbf{Set}$, that is, for functions of the
form $\xi\colon X\to FX$ where $X$ is a set. Intuitively $\xi$
specifies a transition system whose branching type is given by $F$.
Our aim is to determine the behavioural distance of two states.  Given
a coalgebra $\xi$, the idea is to endow $X$ with a distance function
$d_\xi\colon X\times X\to \monM$ defined as the least fixpoint of the
map $d \mapsto d^F \circ (\xi\times \xi)$ where $\_^F$ lifts a
distance function $d \colon X \times X\to \monM$ to
$d^F\colon FX\times FX\to \monM$. Here we focus on a generalisation of
the Wasserstein or Kantorovich distances \cite{v:optimal-transport} to
the categorical setting and then explain how they integrate into the
fixpoint checking framework. Such distances are parametric on a
suitable notion of predicate lifting, hence we will need the results
from the previous section.

We remark that we will use some functors $F$, for which $FY$ is
infinite, even if $Y$ is finite. This is in fact the reason why the
categories $\C$ and $\A$ also include infinite sets. However note that
the resulting fixpoint function will be always defined for finite
sets, although intermediate functions might not conform to this. Hence
the restricted compositionality results of
Section~\ref{sec:categories-functor} are sufficient
(cf. Remark~\ref{rem:embed} and Definition~\ref{def:finc}).  

\subsection{Wasserstein Lifting}

We first recap the definition of the generalised Wasserstein lifting
from \cite{bbkk:coalgebraic-behavioral-metrics}.
Hereafter, $F$ denotes a fixed endo-functor on $\mathbf{Set}$ and
$\xi\colon X\to FX$ is a coalgebra over a finite set $X$. We also fix
a well-behaved non-expansive predicate lifting $\tilde{F}$.
Recall that \emph{pseudo-metrics} are distance functions satisfying:
(i) reflexivity: $\forall x\in X\,.\,d(x,x)=0$; (ii) symmetry:
$\forall x,y\in X\,.\,d(x,y)=d(y,x)$; (iii) triangle inequality:
$\forall x,y,z\in X\,.\,d(x,z)\sqsubseteq d(x,y)\oplus d(y,z)$.  Here
we allow arbitrary distance functions $d\colon X\times X\to \monM$ and
do not restrict to pseudo-metrics.

In order to define a Wasserstein lifting for this functor, a first
ingredient is that of a coupling. Given $t_1,t_2\in FX$ a
\emph{coupling} of $t_1$ and $t_2$ is an element $t\in F(X\times X)$,
such that $F{\pi_i}(t) = t_i$ for $i=1,2$, where
$\pi_i\colon X\times X\to X$ are the projections. We write
$\Gamma (t_1,t_2)$ for the set of all such couplings.

\begin{definition}[Wasserstein lifting]
  \label{def:wasserstein}
  The Wasserstein lifting
  $\_^F\colon \monM^{X\times X}\to \monM^{FX\times FX}$ is defined for
  $d\colon X\times X\to \monM$ and $t_1,t_2\in FX$ as
  \begin{align*}
    d^F(t_1,t_2) = \inf_{t\in \Gamma(t_1,t_2)}
    \tilde{F}d(t)
  \end{align*}
\end{definition}

\begin{example}
  \label{ex:intuition-wasserstein}   
  We consider the Wasserstein lifting in the concrete case where $F$
  equals the distribution functor $\mathcal{D}$. The function
  $\tilde{\mathcal{D}}$ is obtained as the predicate lifting of
  $\mathcal{D}$ (see Section~\ref{sec:liftings} and
  Table~\ref{tab:basic-functions}).
  
  For this instance, the lifting corresponds to the well-known 
  Kantorovich or Wasserstein lifting~\cite{v:optimal-transport}. In fact,
  it gives the solution of a transport problem, where we
  interpret $p_1,p_2$ as the supply respectively demand at each point
  $x\in X$. Transporting a unit from $x_1$ to $x_2$ costs $d(x_1,x_2)$
  and $t$ is a transport plan (= coupling) whose marginals are
  $p_1,p_2$. In other words $d^\mathcal{D}(p_1,p_2)$ can be seen as
  the cost of the optimal transport plan, moving the supply $p_1$ to
  the demand $p_2$.

  In more detail: given $d$, we obtain
  $d^\mathcal{D}\colon \mathcal{D}(X)\times \mathcal{D}(X)\to [0,1]$
  as
  \begin{eqnarray*}
    d^\mathcal{D}(p_1,p_2) & = & \inf \{\tilde{D}d(t) \mid t \in
    \Gamma (p_1,p_2) \} \\
    & = & \inf \{ \sum_{x_1,x_2\in X} d(x_1,x_2)\cdot t(x_1,x_2) \mid t
    \in \Gamma (p_1,p_2) \}
  \end{eqnarray*}
  where
  $\Gamma(p_1,p_2)$ is the set of couplings of $p_1,p_2$ (i.e.,
  distributions $t\colon X\times X\to [0,1]$ such that
  $\sum_{x_2\in X} t(x_1,x_2) = p_1(x_1)$ and
  $\sum_{x_1\in X} t(x_1,x_2) = p_2(x_2)$).
\end{example}

It can be seen that for  well-behaved $\tilde{F}$, the lifting preserves
pseudo-metrics
(see~\cite{bbkk:coalgebraic-behavioral-metrics,bkp:up-to-behavioural-metrics-fibrations}).

In order to make the theory for fixpoint checks effective we will need
to restrict to a subclass of liftings.

\begin{definition}[finitely coupled lifting]
  \label{def:finc}
  We call a lifting $\tilde{F}$ \emph{finitely coupled}  
  if for all $X$ and $t_1,t_2\in FX$ there exists a
  finite
  $\Gamma'(t_1,t_2) \subseteq \Gamma(t_1,t_2)$, which can be computed
  given $t_1,t_2$, such that
  $ \inf_{t\in \Gamma(t_1,t_2)} \tilde{F}d(t) = \min_{t\in
    \Gamma'(t_1,t_2)} \tilde{F}d(t)$ for all $d$.
\end{definition}

We hence ask that the infimum in Definition~\ref{def:wasserstein} is
actually a minimum. Observe that whenever the infimum above is a
minimum, there is trivially such a finite $\Gamma'(t_1,t_2)$. We
however ask that $\Gamma'(t_1,t_2)$ is independent of $d$ and there is
an effective way to determine it.
  
The lifting in Example~\ref{ex:powerset-predlift} (for the finite
powerset functor) is obviously finitely coupled.  For the lifting
$\mathcal{\tilde{\mathcal{D}}}$ from
Example~\ref{ex:distribution-predlift} we note that the set of
couplings $t\in \Gamma(t_1,t_2)$ forms a polytope with a finite number
of vertices, which can be effectively computed and $\Gamma'(t_1,t_2)$
consists of these vertices. The infimum (minimum) is obtained at one
of these
vertices~\cite[Remark~4.5]{bblm:on-the-fly-exact-journal}. This allows
us to always reduce to a finite set of couplings for finite-state
systems.

\subsection{Decomposing the Behavioural Metrics Function}
\label{sec:comp-rep}

As mentioned above, for a coalgebra $\xi\colon X\to FX$ the
behavioural pseudo-metric $d: X \times X \to \monM$ is the
least fixpoint of the behavioural metrics function
$\mathcal{W} = (\_^F) \circ (\xi\times \xi)$ where $(\_^F)$ is the
Wasserstein lifting.

The Wasserstein lifting can be decomposed as
$\_^F = \min\nolimits_u \circ \tilde{F}$ where
$\tilde{F} : \monM^{X\times X} \to \monM^{F(X\times X)}$ is a
predicate lifting -- which we require to be non-expansive
(cf. Lemma~\ref{lem:pred-lifting-nonexpansive}) -- and $\min_u$ is the
minimum over the coupling function
$u\colon F(X\times X) \to FX\times FX$ defined as
$u(t) = (F{\pi_1}(t),F{\pi_2}(t))$, which means that
$\min_u \colon \monM^{F(X\times X)}\to \monM^{FX\times FX}$ (see
Table~\ref{tab:basic-functions}).

Therefore the behavioural metrics function can be expressed as
\begin{center}
  $\mathcal{W} = (\xi\times\xi)^*\circ \min_u\circ \tilde{F}$
\end{center}

Explicitly, for $d \in [0,1]^{X\times X}$ and $x, y \in X$,
\begin{eqnarray*}
  \mathcal{W}(d)(x,y)
  & = & \min\nolimits_{u} \circ \tilde{F}(d)(\xi (x),\xi (y))
  = \min_{u(t) = (\xi(x),\xi(y))}\tilde{F}(d)(t)\\
  & = & \min_{t\in \Gamma(\xi(x),\xi(y))} \tilde{F}d(t)
  = d^F(\xi(x),\xi(y))
\end{eqnarray*}

Note that the fixpoint equation for behavioural metrics is sometimes
equipped with a discount factor that reduces the effect of deviations
in the (far) future and ensures that the fixpoint is unique by
contractivity of the function. Here we focus on the undiscounted case
where the fixpoint equation may have several solutions.

\subsection{A Worked-out Case Study on Behavioural Metrics}
\label{sec:case-study}

As a case study we consider probabilistic transition systems (Markov
chains) with labelled states. These are given by a finite set of
states $X$, a function $\eta\colon X\to \mathcal{D}X$ mapping each
state $x\in X$ to a probability distribution on $X$ and a labelling
function $\ell\colon X\to \Lambda$, where $\Lambda$ is a fixed set of
labels (for examples see Figure~\ref{fig:prob-ts}). Two such systems
are depicted in Figure~\ref{fig:prob-ts}.

  This is represented by a coalgebra
  $\xi:X\to \Lambda\times \mathcal{D}X$ for the functor
  $FX = \Lambda \times \mathcal{D}(X)$, where $\Lambda$ is a fixed set
  of labels.

  In particular, we are interested in computing behavioural metrics
  for such systems. We let $\monM = [0,1]$ and consider the
  Wasserstein lifting for $\mathcal{D}$ explained earlier in
  Example~\ref{ex:intuition-wasserstein} that lifts a distance
  function $d\colon X\times X\to [0,1]$ to
  $d^{\mathcal{D}}\colon \mathcal{D}(X)\times \mathcal{D}(X)\to
  [0,1]$.

  For instance, the best transport plan for the system on the
  left-hand side of Figure~\ref{fig:prob-ts} and the distributions
  $\eta(1),\eta(2)$ (where $\eta(1)(3) = \nicefrac{1}{2}$,
  $\eta(1)(4) = \nicefrac{1}{2}$, $\eta(2)(3) = \nicefrac{1}{3}$,
  $\eta(2)(4) = \nicefrac{2}{3}$) is $t$ with
  $t(3,3) = \nicefrac{1}{3}$, $t(3,4) = \nicefrac{1}{6}$,
  $t(4,4) = \nicefrac{1}{2}$ and $0$ otherwise.

  For the functor $FX = \Lambda \times \mathcal{D}(X)$ we observe that
  couplings of $(a_1,p_1), (a_2,p_2)\in FX$ only exist if $a_1=a_2$
  and -- if they do not exist -- the distance is the infimum of an
  empty set, hence $1$. If $a_1=a_2$, couplings correspond to the
  usual Wasserstein couplings of probability distributions $p_1,p_2$.

\begin{figure}
  \centering
  \scalebox{1}{

  \begin{tikzpicture}
    \node (S1) at (0,-2) [circle,draw]{$1$}; 			
    \node (S1L) at (0,-1.4) {$A$}; 			
    \node (S3) at (-1.6,-2.5) [circle,draw]{$3$};
    \node (S3L) at (-1.6,-1.9) {$B$};
    \node (S4) at (1.6,-2.5) [circle,draw]{$4$};
    \node (S4L) at (1.6,-1.9) {$C$};
    \node (S2) at (0,-3) [circle,draw]{$2$}; 			
    \node (S2L) at (0,-3.6) {$A$}; 			
    \draw  [->] (S1) to node [above]{$\nicefrac{1}{2}$} (S3);
    \draw  [->] (S1) to node [above]{$\nicefrac{1}{2}$} (S4);	
    \draw  [->] (S2) to node [below]{$\nicefrac{1}{3}$} (S3);
    \draw  [->] (S2) to node [below]{$\nicefrac{2}{3}$} (S4);		
    \path
    (S3) edge [loop left] node {$1$} (S3);			 
    \path
    (S4) edge [loop right] node {$1$} (S3);			 
  \end{tikzpicture}
  \qquad \qquad
  \raisebox{10pt}
  {\begin{tikzpicture}
    \node (S1) at (0,-2) [circle,draw]{$1$}; 			
    \node (S1L) at (0,-1.4) {$A$}; 			
    \node (S2) at (3,-2) [circle,draw]{$2$}; 			
    \node (S2L) at (3,-1.4) {$A$}; 			
    \draw  [->] (S1) to [bend left] node [above]{$1$} (S2);
    \draw  [->] (S2) to [bend left] node [below]{$1$} (S1);	
  \end{tikzpicture}}
  }
 
  \caption{Two probabilistic transition systems.}
  \label{fig:prob-ts}
\end{figure}

Hence, the behavioural metric is defined as the least fixpoint of
the function 
\begin{align*}
  \mathcal{B} \colon [0,1]^{X\times X} &\to [0,1]^{X\times
    X} \\
  \mathcal{B}(d)(x_1,x_2) &=
  \begin{cases}
    1 &\mbox{if }
    \ell(x_1)\neq \ell(x_2) \\
    d^\mathcal{D}(\eta(x_1),\eta(x_2))
    &\mbox{otherwise}
  \end{cases}
\end{align*}

Based on the decomposition of $\mathcal{W}$ explained in
Section~\ref{sec:comp-rep}, the function $\mathcal{B}$ can be\linebreak 
written as
\[ \mathcal{B} = \max\nolimits_\rho\circ (c_k \otimes \mathcal{W}) =
  \max\nolimits_\rho\circ (c_k \otimes (\eta\times \eta)^*\circ
  \min\nolimits_u \circ \tilde{\mathcal{D}}), \] where we use the
functions given in Table~\ref{tab:basic-functions-approximations}
(Section~\ref{sec:preliminaries}). More concretely, the types of the
components and the parameters $k,u,\rho$ are given as follows, where
$Y = X\times X$:

\begin{itemize}
\item $c_k\colon [0,1]^\emptyset \to [0,1]^{Y}$ where $k(x,x') = 1$
  if $\ell(x)\neq\ell(x')$ and $0$ otherwise.
\item $\tilde{\mathcal{D}}\colon [0,1]^{Y} \to [0,1]^{\mathcal{D}(Y)}$.
\item
  $\min\nolimits_u\colon [0,1]^{\mathcal{D}(Y)} \to
  [0,1]^{\mathcal{D}(X)\times \mathcal{D}(X)}$ where
  $u\colon \mathcal{D}(Y) \to \mathcal{D}(X)\times \mathcal{D}(X)$,
  $u(t) = (p,q)$ with $p(x) = \sum_{x'\in X} t(x,x')$,
  $q(x) = \sum_{x'\in X} t(x',x)$.
\item
  $(\eta\times \eta)^*\colon [0,1]^{\mathcal{D}(X)\times
    \mathcal{D}(X)} \to [0,1]^{Y}$.
\item $\max\nolimits_\rho\colon [0,1]^{Y + Y} \to [0,1]^{Y}$ where
  $\rho\colon Y + Y \to Y$ is the obvious map from the coproduct
  (disjoint union of sets) $Y + Y$ to $Y$.
\end{itemize}

This decomposition is depicted diagrammatically in
Figure~\ref{fig:decomp-beh-metric}.

\begin{figure}
  \centering
  \begin{tikzpicture}[every node/.style={inner xsep=0.2em, minimum height=2em, outer sep=0, rectangle},scale=1.2]
    \node (c) at (0,0) [draw] {$c_k$};
    \node (init) at (-1.5,-1.2) {};
    \node (D) at (0,-1.2) [draw] {$\tilde{\mathcal{D}}$};
    \node (min) at (2.4,-1.2) [draw] {$\min\nolimits_u$};
    \node (delta) at (6,-1.2) [draw] {$(\eta\times\eta)^*$};
    \node (max) at (8.4,-0.6) [draw, minimum height=5.7em] {$\max\nolimits_\rho$};
    \node (end) at (10.2,-0.6) {};
    \draw [-] (c) to node [above, pos=0.5]{\footnotesize $[0,1]^{Y}$} (max.125);
    \draw [-] (init) to node [below, pos=0.43] {\footnotesize $[0,1]^{Y}$} (D);
    \draw [-] (D) to node [below] {\footnotesize $[0,1]^{\mathcal{D}(Y)}$} (min);
    \draw [-] (min) to node [below] {\footnotesize $[0,1]^{\mathcal{D}(X)\times\mathcal{D}(X)}$} (delta);
    \draw [-] (delta) to node [below] {\footnotesize $[0,1]^{Y}$}
    (max.233);
    \draw [-] (max) to node [above, pos=0.57]{\footnotesize $[0,1]^{Y}$} (end);
  \end{tikzpicture}
  \caption{Decomposition of the fixpoint function $\mathcal{B}$ for computing behavioural metrics.}
  \label{fig:decomp-beh-metric}
\end{figure}

If we consider only finite state spaces, we can, as explained in
Remark~\ref{rem:embed} and after Definition~\ref{def:finc}, restrict
ourselves to finite subsets of $\mathcal{D}(X)$ and $\mathcal{D}(Y)$.

By giving a transport plan as above, it is possible to provide an
upper bound for the Wasserstein lifting and hence there are strategy
iteration algorithms that can approach a fixpoint from above. The
problem with these algorithms is that they might get stuck at a
fixpoint that is not the least. Hence, it is essential to be able to
determine whether a given fixpoint is indeed the smallest one
(see, e.g.,~\cite{bblmtb:prob-bisim-distance-automata-journal}).

Consider, for instance, the transition system in
Figure~\ref{fig:prob-ts} on the right. It contains two states $1,2$ on
a cycle. In fact these two states should be indistinguishable and
hence, if $d = \mu \mathcal{B}$ is the least fixpoint of
$\mathcal{B}$, then $d(1,2) = d(2,1) = 0$. However, the metric $a$
with $a(1,2) = a(2,1) = 1$ ($0$ otherwise) is also a fixpoint and the
question is how to determine that it is not the least.

For this, we use the techniques outlined in
Section~\ref{ss:fix-check}.  We associate $\mathcal{B}$ with an
approximation $\mathcal{B}_\#^a$ on subsets of $X\times X$ such that,
given $Y'\subseteq X\times X$, the set $\mathcal{B}_\#^a(Y')$
intuitively contains all pairs $(x_1,x_2)$ such that, decreasing
function $a$ by some value $\delta$ over $Y'$, resulting in a function
$b$ (defined as $b(x_1,x_2) = a(x_1,x_2)\ominus\delta$ if
$(x_1,x_2)\in Y'$ and $b(x_1,x_2) = a(x_1,x_2)$ otherwise) and
applying $\mathcal{B}$, we obtain a function $\mathcal{B}(b)$, where
the same decrease takes place at $(x_1,x_2)$ (i.e.,
$\mathcal{B}(b)(x_1,x_2) = \mathcal{B}(a)(x_1,x_2) \ominus \delta$).
We will later discuss in detail how to compute $\mathcal{B}_\#^a$. Here, it can be seen that $\mathcal{B}_\#^a(\{(1,2)\}) = \{(2,1)\}$, since a
decrease at $(1,2)$ will cause a decrease at $(2,1)$ in the next
iteration. In fact the greatest fixpoint of $\mathcal{B}_\#^a$, which
here is $\{(1,2),(2,1)\}$, gives us those elements that have a
potential for decrease (intuitively there is ``slack'' or ``wiggle
room'') and form a ``vicious cycle''.

By the results outlined in Section~\ref{ss:fix-check} it holds that
$a$ is the least fixpoint of $\mathcal{B}$ iff the the greatest
fixpoint of $\mathcal{B}_\#^a$ is the empty set.

As discussed earlier, we can here indeed exploit the fact that the
approximation is compositional, i.e., $\mathcal{B}_\#^a$ can be built
out of the approximations of $\max\nolimits_\rho$, $c_k$,
$(\delta\times \delta)^*$, $\min\nolimits_u$, $\tilde{\mathcal{D}}$
(see Table~\ref{tab:basic-functions-approximations}).  

\subsection{Approximation of the Behavioural Metrics Function}
\label{sec:approximation-wasserstein}

While we previously did not spell out how to compute the approximation
of the Wasserstein lifting $\mathcal{W}$, we explain this here in
detail. We rely on the decomposition of $\mathcal{W}$ as given in
Section~\ref{sec:comp-rep}, which can be used to derive its
$d$-approximation (for a given distance function $d$).

\begin{propositionrep}
  \label{prop:approx-wasserstein}
  Let $F\colon \mathbf{Set}\to \mathbf{Set}$ be a functor and let
  $\tilde{F}$ be a corresponding predicate lifting (for $\monM$-valued
  predicates). Assume that $\tilde{F}$ is non-expansive and finitely
  coupled and fix a coalgebra $\xi\colon X\to FX$, where $X$ is
  finite.  Let $Y=X\times X$. For $d\in \monM^Y$ and
  $Y'\subseteq \Ytop{Y}{d}$ we have
  \[ 
    \mathcal{W}_\#^d (Y') = \{ (x,y)\in \Ytop{Y}{d} \mid\ \exists t\in
    \tilde{F}_\#^d(Y'), u(t)=(\xi(x),\xi(y)), \tilde{F}d(t) =
    \mathcal{W}(d)(x,y) \}.
  \]
\end{propositionrep}

\begin{proof}
  We first remark that since $X$ is finite and $\tilde{F}$ is finitely
  coupled it is sufficient to restrict to finite subsets of $F(X\times
  X)$ and $FX\times FX$ (cf. Remark~\ref{rem:embed}). In other words
  $\mathcal{W}$ can be obtained as composition of functions living in
  $\C_f$, hence $\#$ is a proper functor and approximations
  can obtained compositionally. We exploit this fact in the following.

  More concretely, we restrict $u$ to $u\colon V\to W$, where
  $V\subseteq F(X\times X)$, $W\subseteq FX\times FX$. We require that
  $W$ contains all pairs $(\xi(x),\xi(y))$ for $x,y\in X$ and
  $V = \bigcup_{(t_1,t_2)\in W} \Gamma'(t_1,t_2)$. Hence both $V,W$
  are finite.

  The function $\tilde{F}$ is restricted accordingly to a map
  $\monM^Y\to \monM^V$ as explained in Remark~\ref{rem:embed}.
  
  For $d\in \monM^Y$ and $Y'\subseteq \Ytop{Y}{d}$ we have, using the fact
  that $\mathcal{W} = (\xi\times\xi)^*\circ \min_u\circ \tilde{F}$,
  compositionality and the approximations listed in
  Table~\ref{tab:basic-functions-approximations}:
  \begin{align*}
    \mathcal{W}_\#^d (Y') = \{ (x,y)\in \Ytop{Y}{d}\mid\ & (\xi (x),\xi (y))
    \in (\min\nolimits_{u})_\#^{\tilde{F}(d)}
    (\tilde{F}^d_\#(Y')\cap V)\} \\
    = \{ (x,y)\in \Ytop{Y}{d}\mid\ & \arg\min_{t\in u^{-1}(\xi (x),\xi (y))}
    \tilde{F}(d)(t) \cap
    \tilde{F}^d_\#(Y') \cap V \not= \emptyset\} \\
    = \{ (x,y)\in \Ytop{Y}{d}\mid\  & \exists t\in \tilde{F}_\#^d(Y')\cap V, u(t)=(\xi(x),\xi(y)),\\
    &\tilde{F}d(t) = \min\nolimits_{t'\in V}\tilde{F}d(t') \} \\
    = \{ (x,y)\in \Ytop{Y}{d}\mid\  & \exists t\in \tilde{F}_\#^d(Y')\cap V, u(t)=(\xi(x),\xi(y)),\\
    &\tilde{F}d(t) = \min\nolimits_{t'\in
      \Gamma'(\xi(x),\xi(y))}\tilde{F}d(t') \} \\
    = \{ (x,y)\in \Ytop{Y}{d}\mid\ & \exists t\in \tilde{F}_\#^d(Y'),
    u(t)=(\xi(x),\xi(y)), \tilde{F}d(t) = \min\nolimits_{t'\in
      \Gamma'(\xi(x),\xi(y))}\tilde{F}d(t')\!\} \\
    = \{ (x,y)\in \Ytop{Y}{d}\mid\ & \exists t\in \tilde{F}_\#^d(Y'),
    u(t)=(\xi(x),\xi(y)), \tilde{F}d(t) = \mathcal{W}(d)(x,y) \}
  \end{align*}
  The first equality is based on Remark~\ref{rem:embed} and uses the
  fact that the approximation for the restricted $\tilde{F}$ maps $Y'$
  to $\tilde{F}^d_\#(Y')\cap V$.

  The second-last inequality also needs explanation, in particular, we
  have to show that the set on the second-last line is included in the one on
  the previous line, although we omitted the intersection with $V$.

  Hence let $(x,y)\in \Ytop{Y}{d}$ such that there exists
  $s\in\tilde{F}^d_\#(Y')$, $u(s) = (\xi(x),\xi(y))$ and
  $\tilde{F}d(s) = \min\nolimits_{t'\in
    \Gamma'(\xi(x),\xi(y))}\tilde{F}d(t')$. We have to show that there
  exists a $t$ with the same properties that is also included in $V$.

  The fact that $s\in \tilde{F}^d_\#(Y')$ implies that
  $\tilde{F}(d)(s) \ominus \tilde{F}(d\ominus \delta_{Y'})(s)
  \sqsupseteq \delta$ for small enough $\delta$, using the fact that
  $\tilde{F}_\#^d = \gamma^{\tilde{F}(d),\delta} \circ \tilde{F} \circ
  \alpha^{d,\delta}$ (for an appropriate value $\delta$).

  Since the minimum of the Wasserstein lifting is always reached in
  $\Gamma'(\xi(x),\xi(y))$, independently of the argument, there
  exists $t\in \Gamma'(\xi(x),\xi(y))\subseteq V$ (hence $u(t) =
  (\xi(x),\xi(y))$), such that
  \[ \tilde{F}(d\ominus \delta_{Y'})(t) = \min\nolimits_{t'\in
      \Gamma(\xi(x),\xi(y))}\tilde{F}(d\ominus \delta_{Y'})(t'). \]
  This implies that
  $\tilde{F}(d\ominus \delta_{Y'})(t) \sqsubseteq \tilde{F}(d\ominus
  \delta_{Y'})(s)$ (since $s\in \Gamma(\xi(x),\xi(y))$). From the
  assumption
  $\tilde{F}d(s) = \min\nolimits_{t'\in
    \Gamma'(\xi(x),\xi(y))}\tilde{F}d(t')$ we obtain
  $\tilde{F}d(s)\sqsubseteq \tilde{F}d(t)$. Hence, using the fact that
  $\ominus$ is monotone in the first and antitone in the second
  argument, we have:
  \[ \delta\sqsubseteq \tilde{F}(d)(s) \ominus \tilde{F}(d\ominus
    \delta_{Y'})(s) \sqsubseteq \tilde{F}(d)(t) \ominus
    \tilde{F}(d\ominus \delta_{Y'})(t) \sqsubseteq \delta. \] The last
  inequality follows from non-expansiveness. Hence
  \[ \tilde{F}(d)(s) \ominus \tilde{F}(d\ominus \delta_{Y'})(s) =
    \tilde{F}(d)(t) \ominus \tilde{F}(d\ominus \delta_{Y'})(t) =
    \delta, \] which in particular implies that
  $t\in\tilde{F}^d_\#(Y')$.

  In order to conclude we have to show that
  $\tilde{F}d(t) = \min \nolimits_{t'\in
    \Gamma'(\xi(x),\xi(y))}\tilde{F}d(t')$. We first observe that in
  an MV-chain $\monM$, whenever $x\sqsubseteq y$ (for $x,y\in \monM$) we
  can infer that $(y\ominus x)\oplus x = y$ (this follows for instance
  from Lemma~2.4(6) in~\cite{BEKP:FTUD-journal} and duality). The
  inequality
  $\tilde{F}(d\ominus \delta_{Y'}) \sqsubseteq \tilde{F}(d)$ holds by
  monotonicity and we can conclude that
  \begin{align*}
    \tilde{F}d(t) & = (\tilde{F}(d)(t) \ominus \tilde{F}(d\ominus
    \delta_{Y'})(t)) \oplus \tilde{F}(d\ominus \delta_{Y'})(t)) \\
    & = (\tilde{F}(d)(s) \ominus \tilde{F}(d\ominus \delta_{Y'})(s))
    \oplus \tilde{F}(d\ominus \delta_{Y'})(t)) \\
    & \sqsubseteq
    (\tilde{F}(d)(s) \ominus \tilde{F}(d\ominus \delta_{Y'})(s))
    \oplus \tilde{F}(d\ominus \delta_{Y'})(s)) = \tilde{F}(d)(s).
  \end{align*}
  The other inequality $\tilde{F}d(s)\sqsubseteq \tilde{F}d(t)$ holds
  anyway and hence $\tilde{F}d(t) = \tilde{F}d(s)$.  This
  finally implies, as desired, that
  \begin{align*} \tilde{F}d(t) = \tilde{F}d(s) = \min \nolimits_{t'\in
      \Gamma'(\xi(x),\xi(y))}\tilde{F}d(t'). \tag*{\qedhere}\end{align*}
\end{proof}

Intuitively, the definition of $\mathcal{W}_\#^d (Y')$ in
Proposition~\ref{prop:approx-wasserstein} says that $(x,y)$ is
contained in this set, whenever an optimal coupling for the successors
of $x,y$ (i.e., a coupling reaching the minimum in the Wasserstein
lifting) is contained in $\tilde{F}_\#^d(Y')$. Note that
$\tilde{F}_\#^d$ has already been characterised earlier in
Section~\ref{sec:approximation-pred-lifting}.

We now illustrate this result by two simple and concrete examples
involving coalgebras over the distribution and powerset
functor.

\begin{example}
  We continue with the case study from Section~\ref{sec:case-study} on
  probabilistic transition systems. We omit labels, i.e.\ $\Lambda$ is
  a singleton, which implied $\mathcal{B} = \mathcal{W}$.

  Let $X=\{x,y,z\}$. We define a coalgebra
  $\xi \colon X \to \Lambda\times \mathcal{D}X$ via $\xi (x)(x) = 1$,
  $\xi (y)(y) = \xi(y)(z)= \nicefrac{1}{2}$ and $\xi (z)(z)= 1$. All
  other distances are $0$.  
  
  \begin{center}
    \begin{tikzpicture}
      \node (Y) at (4,0) [circle,draw]{$y$}; 			
      \node (X) at (0,0) [circle,draw]{$x$};  
      \node (Z) at (8,0) [circle,draw]{$z$}; 
      \draw  [->] (Y) to [below] node {$\nicefrac{1}{2}$} (X);
      \draw  [->] (Y) to [below] node {$\nicefrac{1}{2}$} (Z);		
      \path (X) edge [loop below] node {$1$} (X);
      \path (Z) edge [loop below] node {$1$} (Z);	
    \end{tikzpicture}
  \end{center}

  Since all states have the same label, they are in fact
  probabilistically bisimilar and hence have behavioural distance $0$,
  given by the least fixpoint of $\mathcal{W}$.  Now consider the
  pseudo-metric $d\colon X\times X \to [0,1]$ with
  $d(x,y) = d(x,z) = d(y,z) = 1$ and $0$ for the reflexive pairs. This
  is a also a fixpoint of $\mathcal{W}$ ($d = \mathcal{W}(d)$), but
  it clearly over-estimates the true behavioural metric.

  We can detect this by computing the greatest fixpoint of
  $\mathcal{W}_\#^d$, which is
  \[ Y' = \{(x,y),(y,x),(x,z),(z,x),(y,z),(z,y)\} \neq\emptyset, \]
  containing the pairs that still have slack in $d$ and whose
  distances can be reduced. We explain why $Y' = \mathcal{W}_\#^d(Y')$
  by focusing on the example pair $(x,y)$ and check that
  $(x,y)\in \mathcal{W}_\#^d(Y')$. For this we use the definition of
  $\mathcal{W}_\#^d$ given in
  Proposition~\ref{prop:approx-wasserstein}.
  
  A valid coupling $t\in \mathcal{D}(X\times X)$ of $\xi (x),\xi(y)$
  is given by $t(x,y)=t(x,z) = \nicefrac{1}{2}$. It satisfies
  $u(t) = (\xi(x),\xi(y))$ and is optimal since it is the only one. We
  obtain the Wasserstein lifting
  \[ \mathcal{W}(d)(x,y) =d^{\mathcal{D}}(\xi(x),\xi (y)) =
    \min_{t'\in \Gamma(\xi(x),\xi (y))} \tilde{\mathcal{D}} d(t') =
    \tilde{\mathcal{D}} d(t) = \nicefrac{1}{2} \cdot 1 +
    \nicefrac{1}{2} \cdot 1 = 1. \] It is left to show that
  $t\in \tilde{\mathcal{D}}_\#^d(Y')$, for which we use the
  characterisation in
  Proposition~\ref{prop:char-abstr-predlift}. We have
  $\bar{d}(x,y) =\bar{d}(x,z) = (1,1)$ where
  \[ \mathcal{D}\bar{d}(t) = p \in \mathcal{D}Z \text{ with } p(1,1) =
    \nicefrac{1}{2}+\nicefrac{1}{2}= 1. \] From
  Lemma~\ref{lem:char-hash-distr} we obtain
  $p \in \tilde{\mathcal{D}}_\#^{\pi_1}((0,1]\times\{1\})$ and by
  definition $t\in (\mathcal{D}\bar{d})^{-1} (p)$, i.e.
  $t\in \tilde{\mathcal{D}}_\#^d(Y')$.  So we can conclude that
  $(x,y) \in \mathcal{W}_\#^d(Y')$.  
\end{example}

\begin{example}
  We now consider an example in the non-deterministic setting. Let
  $X=\{x,y\}$ and a coalgebra $\xi \colon X \to \mathcal{P}X$ be given
  by $\xi (x) = \{x,y\}$, $\xi(y) = \{x\}$.

  \begin{center}
    \begin{tikzpicture}
      \node (Y) at (4,0) [circle,draw]{$y$}; 			
      \node (X) at (0,0) [circle,draw]{$x$};  
      \draw  [->] (X) to [bend left=20] (Y);
      \draw  [->] (Y) to [bend left=20] (X);		
      \path (X) edge [loop below] (X);
    \end{tikzpicture}
  \end{center}

  Since all states have successors, they are in fact bisimilar and
  hence have behavioural distance $0$, given by the least fixpoint of
  $\mathcal{W}$.  Now consider the pseudo-metric
  $d\colon X\times X \to [0,1]$ with
  $d(x,y) = d(y,x) = \nicefrac{1}{2}$ and $0$ for the reflexive
  pairs. This is a also a fixpoint of $\mathcal{W}$
  ($d = \mathcal{W}(d)$), but it clearly over-estimates the true
  behavioural metric.

  We can detect this by computing the greatest fixpoint of
  $\mathcal{W}_\#^d$, which is
  \[ Y' = \{(x,y),(y,x)\} \neq \emptyset, \] containing the pairs that still have
  slack in $d$ and whose distances can be reduced. We explain why
  $Y' = \mathcal{W}_\#^d(Y')$ by focusing on the example pair $(x,y)$
  and check that $(x,y)\in \mathcal{W}_\#^d(Y')$. For this we use the
  definition of $\mathcal{W}_\#^d$ given in
  Proposition~\ref{prop:approx-wasserstein}.
  
  A valid (and optimal) coupling $t\in \mathcal{P}(X\times X)$ of
  $\xi (x),\xi(y)$ is given by $t = \{ (x,x),(y,x)\}$.  It satisfies
  $u(t) = (\xi(x),\xi(y))$. We obtain the Wasserstein lifting
  \[ \mathcal{W}(d)(x,y) =d^{\mathcal{P}}(\xi(x),\xi (y)) = \min_{t\in
      \Gamma(\xi(x),\xi (y))} \tilde{\mathcal{P}} d(t) =
    \tilde{\mathcal{P}} d(t) = \max\{0,\nicefrac{1}{2}\} =
    \nicefrac{1}{2}. \] It is left to show that
  $t\in \tilde{\mathcal{P}}_\#^d(Y')$, for which we use the
  characterisation in
  Proposition~\ref{prop:char-abstr-predlift}. We have
  $\bar{d}(x,x) =\bar{d}(y,y) = (0,0)$,
  $\bar{d}(x,y) =\bar{d}(y,x) = (\nicefrac{1}{2},1)$ and
  \[ \mathcal{P}\bar{d}(t) = S=\{ (0,0),(\nicefrac{1}{2},1)\} \] From
  Lemma~\ref{lem:char-hash-powerset} we obtain
  $S \in
  \tilde{\mathcal{P}}_\#^{\pi_1}((\monM\backslash\{0\})\times\{1\})$
  and $t\in (\mathcal{P}\bar{d})^{-1} (S)$, i.e.
  $t\in \tilde{\mathcal{P}}_\#^d(Y')$.  So we can conclude that
  $(x,y) \in \mathcal{W}_\#^d(Y')$.
\end{example}

\section{GS-Monoidality}
\label{sec:gs-mon}

We will now show that the categories $\Cf$ and $\Af$ can be turned
into gs-monoidal categories, making $\#$ a gs-monoidal functor. This
will give us a method to assemble functions and their approximations
compositionally and this will form the basis for the
tool.
We first define gs-monoidal categories in detail
(cf.~\cite[Definition~7]{gh:inductive-graph} and~\cite{fgtc:gs-monoidal-oplax}).

\begin{definition}[gs-monoidal categories]
  \label{def:gs-monoidal}
  A \emph{strict gs-monoidal category} is a strict symmetric monoidal
  category, where $\otimes$ denotes the tensor and $e$ its unit and
  symmetries are given by $\rho_{a,b}\colon a\otimes b\to b\otimes a$.
  For every object $a$ there exist morphisms
  $\nabla_a\colon a\to a\times a$ (\emph{duplicator}) and
  $!_a\colon a\to e$ (\emph{discharger}) satisfying the axioms given
  below. (See also their visualisations as string diagrams in
  Figure~\ref{fig:gs-axioms}.)

\medskip

\begin{small}
\begin{tabular}{ll}
  \hline
  \begin{tabular}{l}
  \emph{functoriality of tensor}\\[.5mm]
  $(g\otimes g') \circ (f\otimes f') = (g\circ f) \otimes (g'\circ f')$\\
  $\mathit{id}_{a\otimes b} = \mathit{id}_a \otimes \mathit{id}_b$\\[.5mm]
  \hline
  \emph{monoidality}\\[.5mm]
  $(f\otimes g) \otimes h = f \otimes (g\otimes h)$\\
  $f\otimes \mathit{id}_e=f = \mathit{id}_e\otimes f$\\[.5mm]
  \end{tabular}
  &
  \begin{tabular}{l}
  \emph{naturality}\\[.5mm]
    $(f'\otimes f) \circ \rho_{a,a'} = \rho_{b,b'} \circ (f\otimes f')$\\
    \ \\[.5mm]
  \hline
  \emph{symmetry}\\[.5mm]
  $\rho_{e,e} = \mathit{id}_e$ \qquad
  $\rho_{b,a}\circ \rho_{a,b} = \mathit{id}_{a\otimes b}$\\
  $\rho_{a\otimes b,c} = (\rho_{a,c}\otimes \mathit{id}_b \circ(\mathit{id}_a\otimes \rho_{b,c}) $\\
  \end{tabular}\\[.5mm]
  \hline
  \hline
  \multicolumn{2}{l}{\ \emph{gs-monoidality}}\\[.5mm]
  \multicolumn{2}{l}{\ $!_e = \nabla_e = \mathit{id}_e$}\\[.5mm]
  \begin{tabular}{l}
  \hline
  \emph{coherence axioms}\\[.5mm]
  $(\mathit{id}_a\otimes \nabla_a) \circ \nabla_a = (\nabla_a\otimes \mathit{id}_a)\circ \nabla_a$\\
  $\mathit{id}_a=(\mathit{id}_a\otimes !_a) \circ \nabla_a$\\
  $\rho_{a,a}\circ \nabla_a = \nabla_a$\\[.5mm]
  \end{tabular}
  &
  \begin{tabular}{l}
  \hline
  \emph{monoidality axiom}s\\[.5mm]
  $!_{a\otimes b} = !_a \otimes !_b$\\
  $(\mathit{id}_a\otimes \rho_{a,b} \otimes
  \mathit{id}_b) \circ (\nabla_a\otimes \nabla_b) =
  \nabla_{a\otimes b}$\\
  (or, equiv.
  $\nabla_a\otimes \nabla_b = (\mathit{id}_a\otimes
  \rho_{b,a} \otimes \mathit{id}_b) \circ
    \nabla_{a\otimes b}$
  \end{tabular}\\
  \hline 
\end{tabular}
\end{small}

\medskip

A functor $F : \C\to \D$, where $\C$ and $\D$ are gs-monoidal categories, is \emph{gs-monoidal} if the following holds:

\medskip

  \begin{small}
    \begin{tabular}{l@{\hspace{1cm}}l@{\hspace{1cm}}l}
    \hline
    \begin{tabular}{l}
    \emph{monoidality}\\[.5mm]
    $F(e) = e'$\\
    $F(a\otimes b) = F(a) \otimes' F(b)$\\
    \end{tabular}
    &
    \begin{tabular}{l}
    \emph{symmetry}\\[.5mm]
    $F(\rho_{a,b}) = \rho'_{F(a),F(b)}$\\
     \\
    \end{tabular}
    &
    \begin{tabular}{l}
    \emph{gs-monoidality}\\[.5mm]
    $F(!_a) = !'_{F(a)}$\\
     $F(\nabla_a) = \nabla'_{F(a)}$\\
    \end{tabular}\\
    \hline 
  \end{tabular}
  \end{small}
  \medskip
  
\noindent  
where the primed operators are from category $\D$, the others from $\C$.
\end{definition}

Note that the visualisations of the axioms in
Figure~\ref{fig:gs-axioms} match the images in
Figure~\ref{fig:decomp-termination} and
Figure~\ref{fig:decomp-beh-metric}. However, instead of labelling the
wires with the types of objects as in the previous figures, we here
label the wires with a fixed object.

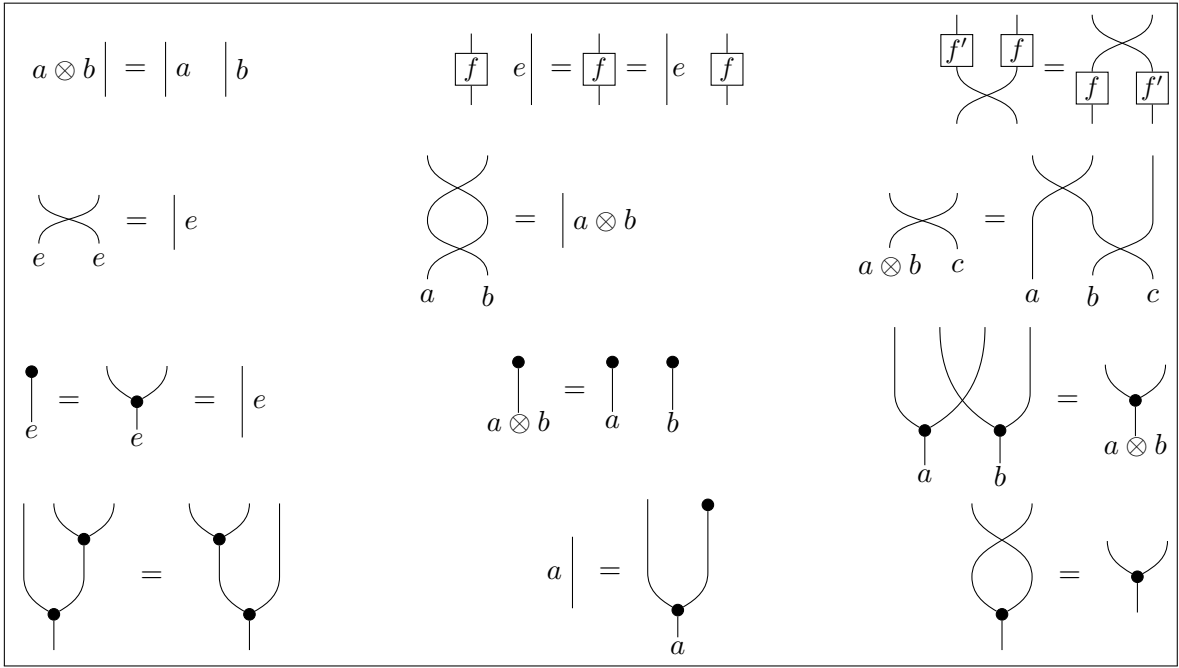
\begin{figure}[h]
  \begin{center}
  \hspace*{-5pt}
  \fbox{
  \hspace*{-8pt}
  \parbox{\textwidth}{    
  \begin{tikzpicture}[baseline={(current bounding box.center)}]
    \node (x) at (0,0) {};
    \node (x1) at (1,0) {};
    \node (eq) at (1.5,0) {$=$};
    \node (y) at (2,-0.4) {};
    \node (y1) at (3,-0.4) {};
    \node (z) at (2,0.4) {};
    \node (z1) at (3,0.4) {};
    \draw (x) -- node[above]{$a\otimes b$} (x1);
    \draw (y) -- node[above]{$b$} (y1);
    \draw (z) -- node[above]{$a$} (z1);
  \end{tikzpicture}
  \hfill\vrule\hfill
  \begin{tikzpicture}[baseline={(current bounding box.center)}, every node/.style={inner sep=0, outer xsep=0}]
    \node (a) at (0,0) {};
    \node[minimum height=1.1em, minimum width=1.1em, draw] (f) at (0.5,0) {\small $f$};
    \node (b) at (1,0) {};
    \node (e) at (0,0.8) {};
    \node (e1) at (1,0.8) {};
    \node (eq) at (1.5,0.4) {$=$};
    \node (a1) at (2,0.4) {};
    \node[minimum height=1.1em, minimum width=1.1em, draw] (f1) at (2.5,0.4) {\small $f$};
    \node (b1) at (3,0.4) {};
    \node (eq0) at (3.5,0.4) {$=$};
    \node (e0) at (4,0) {};
    \node (e01) at (5,0) {};
    \node (a0) at (4,0.8) {};
    \node[minimum height=1.1em, minimum width=1.1em, draw] (f0) at (4.5,0.8) {\small $f$};
    \node (b0) at (5,0.8) {};
    \draw (a) -- (f) -- (b);
    \draw (e) -- node[below, inner sep=2pt]{$e$} (e1);
    \draw (a1) -- (f1) -- (b1);
    \draw (a0) -- (f0) -- (b0);
    \draw (e0) -- node[above, inner sep=2pt]{$e$} (e01);
  \end{tikzpicture}
  \hfill\vrule\hfill
  \begin{tikzpicture}[baseline={(current bounding box.center)}, every node/.style={inner sep=0, outer xsep=0}]
    \node (a) at (0,0) {};
    \node (a1) at (0,0.8) {};
    \node[minimum height=1.1em, minimum width=1.1em, draw] (f1) at (1,0) {\small $f'$};
    \node[minimum height=1.1em, minimum width=1.1em, draw] (f) at (1,0.8) {\small $f$};
    \node (b1) at (1.5,0) {};
    \node (b) at (1.5,0.8) {};
    \node (eq) at (2,0.4) {$=$};
    \node (a0) at (2.5,0) {};
    \node (a01) at (2.5,0.8) {};
    \node[minimum height=1.1em, minimum width=1.1em, draw] (f0) at (3,0) {\small $f$};
    \node[minimum height=1.1em, minimum width=1.1em, draw] (f01) at (3,0.8) {\small $f'$};
    \node (b01) at (4,0) {};
    \node (b0) at (4,0.8) {};
    \draw (a) to[in=180, out=0] (f) -- (b);
    \draw (a1) to[in=180, out=0] (f1) -- (b1);
    \draw (a0) -- (f0) to[in=180, out=0] (b0);
    \draw (a01) -- (f01) to[in=180, out=0] (b01);
  \end{tikzpicture}
  \end{center}
  \smallskip\hrule\smallskip
  \begin{center}
  \begin{tikzpicture}[baseline={(current bounding box.center)}]
    \node (a) at (0,0) {$e$};
    \node (b) at (0,0.8) {$e$};
    \node (b1) at (1.5,0) {};
    \node (a1) at (1.5,0.8) {};
    \node (eq) at (2,0.4) {$=$};
    \node (e) at (2.5,0.4) {};
    \node (e1) at (3.5,0.4) {};
    \draw (a) to[out=0, in=180] (a1);
    \draw (b) to[out=0, in=180] (b1);
    \draw (e) -- node[above]{$e$} (e1);
  \end{tikzpicture}
  \hfill\vrule\hfill
  \begin{tikzpicture}[baseline={(current bounding box.center)}]
    \node (a) at (0,0) {$b$};
    \node (b) at (0,0.8) {$a$};
    \node[inner sep=0, outer sep=0] (b1) at (1,0) {};
    \node[inner sep=0, outer sep=0] (a1) at (1,0.8) {};
    \node (a2) at (2,0) {};
    \node (b2) at (2,0.8) {};
    \node (eq) at (2.5,0.4) {$=$};
    \node (ab) at (3,0.4) {};
    \node (ab1) at (4,0.4) {};
    \draw (a) to[out=0, in=180] (a1.center) to[out=0, in=180] (a2);
    \draw (b) to[out=0, in=180] (b1.center) to[out=0, in=180] (b2);
    \draw (ab) -- node[above]{$a\otimes b$} (ab1);
  \end{tikzpicture}
  \hfill\vrule\hfill
  \begin{tikzpicture}[baseline={(current bounding box.center)}]
    \node (ab) at (0.4,0) {$c$};
    \node (x) at (0.5,0) {};
    \node (c0) at (0,1) {$a\otimes b$};
    \node (c01) at (1.5,0) {};
    \node (ab1) at (1.5,1) {};
    \node (eq) at (2,0.5) {$=$};
    \node (a) at (2.5,1.2) {$a$};
    \node (b) at (2.5,0.4) {$b$};
    \node (c) at (2.5,-0.2) {$c$};
    \node (a1) at (3.5,1.2) {};
    \node (c1) at (3.5,0.4) {};
    \node (b1) at (3.5,-0.2) {};
    \node (c2) at (4.5,1.2) {};
    \node (a2) at (4.5,0.4) {};
    \node (b2) at (4.5,-0.2) {};
    \draw (x) to[out=0, in=180] (ab1);
    \draw (c0) to[out=0, in=180] (c01);
    \draw (a) to[out=0, in=180] (a1.center) to[out=0, in=180] (a2);
    \draw (b) to[out=0, in=180] (b1.center) to[out=0, in=180] (b2);
    \draw (c) to[out=0, in=180] (c1.center) to[out=0, in=180] (c2);
  \end{tikzpicture}
  \end{center}
  \smallskip\hrule\smallskip
  \begin{center}
  \begin{tikzpicture}[baseline={(current bounding box.center)}, every node/.style={inner sep=0, outer xsep=0, circle}]
    \node (i) at (0,0) {$e$};
    \node[inner xsep=.2em, draw, fill, black] (del) at (0.9,0) {};
    \node (eq) at (1.5,0) {$=$};
    \node (a) at (2,0) {$e$};
    \node[inner xsep=.2em, draw, fill, black] (N) at (2.5,0) {};
    \node (b) at (3,0.4) {};
    \node (c) at (3,-0.4) {};
    \node (eq0) at (3.5,0) {$=$};
    \node (e) at (4,0) {};
    \node (e1) at (4.5,0) {};
    \draw (i) -- (del);
    \draw (a) -- (N) to[out=60, in=180] (b);
    \draw (N) to[out=-60, in=180] (c);
    \draw (e) -- node[above, inner sep=2pt]{$e$} (e1);
  \end{tikzpicture}
  \hfill\vrule\hfill
  \begin{tikzpicture}[baseline={(current bounding box.center)}, every node/.style={inner sep=0, outer xsep=0, circle}]
    \node (ab) at (-0.3,0) {$a\otimes b$};
    \node (x) at (0.2,0) {};
    \node[inner xsep=.2em, draw, fill, black] (del) at (0.9,0) {};
    \node (eq) at (1.5,0) {$=$};
    \node (a) at (2,-0.4) {$a$};
    \node (b) at (2,0.4) {$b$};
    \node[inner xsep=.2em, draw, fill, black] (del1) at (3,-0.4) {};
    \node[inner xsep=.2em, draw, fill, black] (del2) at (3,0.4) {};
    \draw (x) -- (del);
    \draw (a) -- (del1);
    \draw (b) -- (del2);
  \end{tikzpicture}
  \hfill\vrule\hfill
  \begin{tikzpicture}[baseline={(current bounding box.center)}, every node/.style={inner sep=0, outer xsep=0, circle}]
    \node (a) at (0,-0.4) {$a$};
    \node[inner xsep=.2em, draw, fill, black] (Na) at (0.6,-0.4) {};
    \node (a1) at (1.1,0) {};
    \node (ia1) at (2,0) {};
    \node (a2) at (2,-1.2) {};
    \node (b) at (0,-1.4) {$b$};
    \node[inner xsep=.2em, draw, fill, black] (Nb) at (0.6,-1.4) {};
    \node (b1) at (2,-0.6) {};
    \node (b2) at (1.1,-1.8) {};
    \node (ib2) at (2,-1.8) {};
    \node (eq) at (2.5,-0.9) {$=$};
    \node (ab) at (3.5,-0.9) {$a\otimes b$};
    \node (x) at (4,-0.9) {};
    \node[inner xsep=.2em, draw, fill, black] (N) at (4.5,-0.9) {};
    \node (ab1) at (5.5,-0.5) {};
    \node (ab2) at (5.5,-1.4) {};
    \draw (a) -- (Na) to[out=60, in=180] (a1.center) -- (ia1);
    \draw (Na) to[out=-60, in=180] (a2);
    \draw (Nb) to[out=60, in=180] (b1);
    \draw (b) -- (Nb) to[out=-60, in=180] (b2.center) -- (ib2);
    \draw (x) -- (N) to[out=60, in=180] (ab1);
    \draw (N) to[out=-60, in=180] (ab2);
  \end{tikzpicture}
  \end{center}
  \smallskip\hrule\smallskip
  \begin{center}
  \begin{tikzpicture}[baseline={(current bounding box.center)}, every node/.style={inner sep=0, outer xsep=0, circle}]
    \node (a) at (0,-0.4) {};
    \node[inner xsep=.2em, draw, fill, black] (N1) at (0.5,-0.4) {};
    \node (b) at (1,0) {};
    \node (c) at (1,-0.8) {};
    \node[inner xsep=.2em, draw, fill, black] (N2) at (1.5,-0.8) {};
    \node (b1) at (2,0) {};
    \node (c1) at (2,-0.4) {};
    \node (c2) at (2,-1.2) {};
    \node (eq) at (2.5,-0.6) {$=$};
    \node (a0) at (3,-0.8) {};
    \node[inner xsep=.2em, draw, fill, black] (N) at (3.5,-0.8) {};
    \node (c0) at (4,-0.4) {};
    \node (b0) at (4,-1.2) {};
    \node[inner xsep=.2em, draw, fill, black] (N0) at (4.5,-0.4) {};
    \node (c01) at (5,0) {};
    \node (c02) at (5,-0.8) {};
    \node (b01) at (5,-1.2) {};
    \draw (a) -- (N1) to[out=60, in=180] (b.center) -- (b1);
    \draw (N1) to[out=-60, in=180] (c.center) -- (N2) to[out=60, in=180] (c1);
    \draw (N2) to[out=-60, in=180] (c2);
    \draw (a0) -- (N) to[out=-60, in=180] (b0.center) -- (b01);
    \draw (N) to[out=60, in=180] (c0.center) -- (N0) to[out=60, in=180] (c01);
    \draw (N0) to[out=-60, in=180] (c02);
  \end{tikzpicture}
  \hfill\vrule\hfill
  \begin{tikzpicture}[baseline={(current bounding box.center)}, every node/.style={inner sep=0, outer xsep=0, circle}]
    \node (a0) at (0.5,0) {};
    \node (a1) at (1.5,0) {};
    \node (eq) at (2,0) {$=$};
    \node (a) at (2.5,0) {$a$};
    \node[inner xsep=.2em, draw, fill, black] (N) at (3,0) {};
    \node (b) at (3.5,0.4) {};
    \node (c) at (3.5,-0.4) {};
    \node (b1) at (4.5,0.4) {};
    \node[inner xsep=.2em, draw, fill, black] (c1) at (4.4,-0.4) {};
    \draw (a0) -- node[above, inner sep=2pt]{$a$} (a1);
    \draw (a) -- (N) to[out=60, in=180] (b.center) -- (b1);
    \draw (N) to[out=-60, in=180] (c.center) -- (c1);
  \end{tikzpicture}
  \hfill\vrule\hfill
  \begin{tikzpicture}[baseline={(current bounding box.center)}, every node/.style={inner sep=0, outer xsep=0, circle}]
    \node (a) at (0,-0.4) {};
    \node[inner xsep=.2em, draw, fill, black] (N) at (0.5,-0.4) {};
    \node (b) at (1,0) {};
    \node (c) at (1,-0.8) {};
    \node (c1) at (2,0) {};
    \node (b1) at (2,-0.8) {};
    \node (eq) at (2.5,-0.4) {$=$};
    \node (a0) at (3,-0.4) {};
    \node[inner xsep=.2em, draw, fill, black] (N0) at (3.5,-0.4) {};
    \node (b0) at (4,0) {};
    \node (c0) at (4,-0.8) {};
    \draw (a) -- (N) to[out=60, in=180] (b.center) to[out=0, in=180] (b1);
    \draw (N) to[out=-60, in=180] (c.center) to[out=0, in=180] (c1);
    \draw (a0) -- (N0) to[out=60, in=180] (b0);
    \draw (N0) to[out=-60, in=180] (c0);
  \end{tikzpicture}
  }}
  \end{center}
  \caption{String diagrams for the axioms of gs-monoidal categories.}
  \label{fig:gs-axioms}
\end{figure}

In fact, in order to obtain strict gs-monoidal categories with
disjoint union as tensor, we will work with the skeleton categories
where every finite set $Y$ is represented by an isomorphic copy
$\{1,\dots,|Y|\}$. This enables us to make disjoint union strict,
i.e., associativity holds on the nose and not just up to
isomorphism. In particular for finite sets $Y,Z$, we define disjoint
union as $Y+Z = \{ 1,\dots ,|Y|,|Y|+1,\dots, |Y|+|Z|\}$.

\begin{theoremrep}[$\Cf$ is gs-monoidal]
  \label{thm:c-gs-monoidal}
  The category $\Cf$ with the following operators is gs-monoidal:
\begin{enumerate}
\item The tensor $\otimes$ on objects $a\in \monM^Y$ and
  $b\in \monM^Z$ is defined as
  \[
    a\otimes b = a+b \in \monM^{Y+Z}
  \]
  where for $k\in Y+Z$ we have $(a+b)(k) = a(k)$ if $k\leq |Y|$ and
  $(a+b)(k) = b(k-|Y|)$ if $|Y|<k\leq |Y|+|Z|$.

  On arrows $f\colon a\arr b$ and $g\colon a'\arr b'$ (with
  $a'\in \monM^{Y'}$, $b'\in \monM^{Z'}$) tensor is given by
  \[
    f\otimes g \colon \monM^{Y+ Y'} \to \monM^{Z+ Z'}, \quad
    (f\otimes g) (u) = f(\cev{u}_Y)+g(\vec{u}_{Y})
  \]
  for $u\in \monM^{Y+ Y'}$ where $\cev{u}_Y\in \monM^Y$ and
  $\vec{u}_Y\in \monM^{Y'}$ are defined as
  \[ \cev{u}_Y(k) = u(k) \mbox{ for $1\leq k \leq |Y|$ and }
    \vec{u}_Y(k) = u(|Y|+k) \mbox{ for $1\leq k \leq |Y'|$.} \]
\item The symmetry $\rho_{a,b}\colon a\otimes b\arr b\otimes a$ for
  $a\in \monM^Y$, $b\in \monM^Z$ is defined for $u\in \monM^{Y+Z}$ as
  \[
    \rho_{a,b}(u) = \vec{u}_Y+\cev{u}_Y.
  \]
\item The unit $e$ is the unique mapping
  $e \colon\emptyset \to \monM$.

\item The duplicator $\nabla_a \colon a\arr a\otimes a$ for
  $a\in \monM^Y$ is defined for $u\in \monM^Y$ as
  \[\nabla_{a}(u) =  u+u.\]

\item The discharger $!_a\colon a\arr e$ for $a\in \monM^Y$ is defined
  for $u \in \monM^Y$ as $!_a(u) = e$.

\end{enumerate}
\end{theoremrep}

\begin{proof}
  In the following let $a\in \monM^Y$, $a'\in \monM^{Y'}$,
  $b\in \monM^Z$, $b'\in \monM^{Z'}$, $c\in \monM^W$,
  $c'\in \monM^{W'}$ be objects in $\Cf$.

  We know that $\Cf$ is a well-defined category from
  Lemma~\ref{le:cate}. We also note that disjoint unions of
  non-expansive functions are non-expansive. Moreover, given
  $f\colon a\arr b$ and $g\colon a'\arr b'$, that
  \begin{align*}
    & (f\otimes g) (a\otimes a') 
    = (f\otimes g) (a+ a') \\
    &= f(\overleftarrow{(a+a')}_Y)  +  g(\overrightarrow{(a+a')}_Y)
    = f(a) + g(a')\\
    & = b+b' 
    = b\otimes b'.
  \end{align*}
  Thus  $f\otimes g$ is a well-defined arrow
  $a\otimes a'\arr b\otimes b'$.

We next verify all the axioms of gs-monoidal categories given in
Definition~\ref{def:gs-monoidal}. In general the calculations are
straightforward, but they are provided here for completeness.

In the sequel we will often use the fact that
$\cev{u}_Y + \vec{u}_Y = u$ whenever $Y$ is a subset of the domain of
$u$.

\begin{enumerate}
\item functoriality of tensor: 

  \begin{itemize}
  \item 
    $\mathit{id}_{a\otimes b} = \mathit{id}_a \otimes \mathit{id}_b$

    \smallskip
    
    Let $u\in \monM^{Y+Z}$. Then
    \begin{eqnarray*}
      &&(\mathit{id}_a \otimes \mathit{id}_b) (u) 
      = \mathit{id}_a (\cev{u}_Y) +  \mathit{id}_b (\vec{u}_Y)
      = \cev{u}_Y+\vec{u}_Y
      = u 
      = \mathit{id}_{a\otimes b}(u) 
    \end{eqnarray*}  
    
    \smallskip

  \item
    $(g\otimes g') \circ (f\otimes f') = (g\circ f) \otimes (g'\circ
    f')$

    \smallskip

    This is required to hold when both sides are defined.
    Hence let $f\colon a\arr b$, $g\colon b\arr c$,
    $f'\colon a'\arr b'$, $g'\colon b'\arr c'$ and
    $u\in \monM^{Y+ Y'}$. We obtain:
    \begin{align*}
      & (g\otimes g') \circ (f\otimes f') (u) 
      = (g\otimes g') (f(\cev{u}_Y)+f'(\vec{u}_Y)) \\
      &= g(f(\cev{u}_Y)) + g'(f'(\vec{u}_Y))
      = ((g\circ f) \otimes (g'\circ f')) (\cev{u}_Y+\vec{u}_Y) \\
      &= ((g\circ f) \otimes (g'\circ f'))(u)
    \end{align*}  
  \end{itemize}
  
\item monoidality:
  \begin{itemize}
  \item $f\otimes \mathit{id}_e=f = \mathit{id}_e\otimes f$

    \smallskip
    
    Let $f\colon a\arr b$ and $u\in \monM^Y$. It holds that:
    \begin{align*}
      & (f\otimes \mathit{id}_e) (u) = (f\otimes \mathit{id}_e) (u+e) 
      = f(u) + \mathit{id}_e(e) \\
      &= f(u) + e = f(u) = e + f(u) \\
      &=   \mathit{id}_e(e) + f(u) 
      = ( \mathit{id}_e \otimes f) (e+u) = (\mathit{id}_e \otimes f)(u)
    \end{align*}

    \smallskip
            
  \item $(f\otimes g) \otimes h = f \otimes (g\otimes h)$

    \smallskip
    
    Let $f\colon a\arr a'$, $g\colon b\arr b'$ and $h\colon c\arr c'$
    and $u\in \monM^{Y+ Z+ W}$, then
    \begin{align*}
      & ((f\otimes g)\otimes h)(u) 
      = (f\otimes g) (\cev{u}_{Y+Z}) + h(\vec{u}_{Y+Z}) \\
      &= (f(\overleftarrow{(\cev{u}_{Y+Z})}_Y) +
      g(\overrightarrow{(\cev{u}_{Y+Z})}_Y)) + h(\vec{u}_{Y+Z}) \\
      &= f(\cev{u}_Y) + (g(\overleftarrow{(\vec{u}_Y)}_Z) +
      h(\overrightarrow{(\vec{u}_Y)}_Z)) \\
      &= f(\cev{u}_Y) + (g\otimes h) (\vec{u}_Y) 
      = (f\otimes (g\otimes h)) (u)
    \end{align*} 
    where we use the fact that $\overrightarrow{(\cev{u}_{Y+Z})}_Y =
    \overleftarrow{(\vec{u}_Y)}_Z$. 
  \end{itemize}

  \smallskip

\item naturality:
\begin{itemize}
\item $(f'\otimes f) \circ \rho_{a,a'} = \rho_{b,b'} \circ (f\otimes
  f')$

  \smallskip

  Let $f\colon a\arr b$ and $f'\colon a'\arr b'$. Then for $u\in \monM^{Y+ Y'}$
  \begin{align*}
    & (\rho_{b,b'} \circ (f\otimes f'))(u)  
    = \rho_{b,b'} (f(\cev{u}_Y)+f'(\vec{u}_Y)) \\
    &= f'(\vec{u}_Y) + f(\cev{u}_Y) 
    = (f'\otimes f)(\vec{u}_Y+\cev{u}_Y)
    = ((f'\otimes f) \circ \rho_{a,a'})(u)
 \end{align*}

\end{itemize}
\item symmetry:
\begin{itemize}
\item $\rho_{e,e} = \mathit{id}_e$

  \smallskip

  We note that $e$ is the unique function from $\emptyset$ to $\monM$
  and furthermore $e\otimes e = e+e = e$. Then
  \[ \rho_{e,e} (e) = \rho_{e,e} (e+e) = e+e = e = \mathit{id}_e(e)\]

\item $\rho_{b,a}\circ \rho_{a,b} = \mathit{id}_{a\otimes b}$

  \smallskip

Let $u\in \monM^{Y+Z}$, then:
\[ (\rho_{b,a} \circ \rho_{a,b})(u) = \rho_{b,a} (\vec{u}_Y+\cev{u}_Y)
  =  \cev{u}_Y + \vec{u}_Y = u = \mathit{id}_{a\otimes b}(u) \]

\item $\rho_{a\otimes b,c} = (\rho_{a,c}\otimes \mathit{id}_b)
  \circ(\mathit{id}_a\otimes \rho_{b,c})$

  \smallskip

Let $u\in \monM^{Y+Z+W}$, then:
\begin{align*}
& ((\rho_{a,c} \otimes \mathit{id}_b) \circ (\mathit{id}_a \otimes  \rho_{b,c}))(u)\\
&= (\rho_{a,c} \otimes \mathit{id}_b) (\mathit{id}_a (\cev{u}_Y) + \rho_{b,c}(\vec{u}_Y)) \\
&= (\rho_{a,c} \otimes
\mathit{id}_b)(\cev{u}_Y+\overrightarrow{(\vec{u}_Y)}_Z +
\overleftarrow{(\vec{u}_Y)}_Z) \\
&= \rho_{a,c}(\cev{u}_Y+\vec{u}_{Y+Z}) + \mathit{id}_b (\overleftarrow{(\vec{u}_Y)}_Z) 
= \vec{u}_{Y+Z} + \cev{u}_Y+\overrightarrow{(\cev{u}_{Y+Z})}_Y \\
&= \vec{u}_{Y+Z} + \cev{u}_{Y+Z} 
= \rho_{a\otimes b,c} (u)
\end{align*}
where we use the fact that $\overrightarrow{(\vec{u}_Y)}_Z =
\vec{u}_{Y+Z}$ and $\overleftarrow{(\vec{u}_Y)}_Z =
\overrightarrow{(\cev{u}_{Y+Z})}_Y$. 

\end{itemize}

    \smallskip

\item gs-monoidality:
\begin{itemize}
\item $!_e = \nabla_e = \mathit{id}_e$

  \smallskip

  Since $e$ is the unique function of type $\emptyset\to\monM$ and $e+e=e$, we obtain:
  \[ !_e(e) = e = \mathit{id}_e (e) = e = e + e  = \nabla_e(e) \]

\item coherence axioms:

  \smallskip
  For $u\in \monM^Y$, we note that $\overleftarrow{(u+u)}_Y = \overrightarrow{(u+u)}_Y = u$.
  
\begin{itemize}
\item $(\mathit{id}_a\otimes \nabla_a) \circ \nabla_a =
  (\nabla_a\otimes \mathit{id}_a)\circ \nabla_a$

  \smallskip

Let $u\in \monM^Y$, then:
\begin{align*}
& ((\mathit{id}_{a} \otimes \nabla_{a})\circ \nabla_a)(u) 
= (\mathit{id}_a\otimes \nabla_a)(u+u)\\
&= \mathit{id}_a(u) + \nabla_a(u) 
= u + u + u 
= \nabla_a(u) + \mathit{id}_a(u) \\
&= (\nabla_a \otimes \mathit{id}_a) (u+u) 
= (\nabla_a \otimes \mathit{id}_a) (\nabla_a(u))
\end{align*}

\item $\mathit{id}_a=(\mathit{id}_a\otimes !_a) \circ \nabla_a$

  \smallskip
  
Let $u\in \monM^Y$, then:
\begin{align*}
&((\mathit{id}_{a} \otimes !_{a}) \circ \nabla_a)(u) 
= (\mathit{id}_{a} \otimes !_{a}) (u+u)\\
&= \mathit{id}_a(u) + !_a(u) 
= \mathit{id}_a(u) + e 
= \mathit{id}_a(u)
\end{align*}

\item $\rho_{a,a}\circ \nabla_a = \nabla_a$

  \smallskip

Let $u\in \monM^Y$, then:
\[ (\rho_{a,a} \circ \nabla_a)(u) = \rho_{a,a} (u + u) = u + u =
  \nabla_a (u) \]

\end{itemize}
\item monoidality axioms:
\begin{itemize}
\item $!_{a\otimes b} = !_a \otimes !_b$

  \smallskip

Let $u\in \monM^{Y+Z}$, then:
\[
  !_{a\otimes b} (u)  = e = e+e = !_a(\cev{u}_Y) + !_b(\vec{u}_Y) \\
  = (!_a\otimes !_b)(u)
\]

    \smallskip

\item $\nabla_a\otimes \nabla_b = (\mathit{id}_a\otimes \rho_{b,a}
  \otimes \mathit{id}_b) \circ \nabla_{a\otimes b}$

  \smallskip

Let $u\in \monM^{Y+Z}$, then:
\begin{align*}
&(\mathit{id}_a \otimes \rho_{b,a} \otimes \mathit{id}_b) ( \nabla_{a\otimes b} (u)) 
= (\mathit{id}_a \otimes \rho_{b,a} \otimes \mathit{id}_b) (u+u) \\
&= (\mathit{id}_a \otimes \rho_{b,a} \otimes \mathit{id}_b) (\cev{u}_Y+\vec{u}_Y+\cev{u}_Y+\vec{u}_Y) \\
&= \cev{u}_Y+\cev{u}_Y+\vec{u}_Y+\vec{u}_Y 
= \nabla_a(\cev{u}_Y) + \nabla_b(\vec{u}_Y) \\
&= (\nabla_a\otimes \nabla_b) (\cev{u}_Y+\vec{u}_Y) 
= (\nabla_a\otimes \nabla_b) (u) \tag*{\qedhere}
\end{align*}
\end{itemize}
\end{itemize}
\end{enumerate}
\end{proof}

We now turn to the category of approximations $\Af$. Note that here
functions have as parameters sets of the form
$U\subseteq [Y]^a\subseteq Y$. Hence, (the cardinality of) $Y$ cannot
be determined directly from $U$ and we need extra care with the
tensor.

\begin{theoremrep}[$\Af$ is gs-monoidal]
  \label{thm:a-gs-monoidal}
  The category $\Af$ with the following operators is gs-monoidal:
\begin{enumerate}
\item The tensor $\otimes$ on objects $a\in \monM^Y$ and
  $b\in \monM^Z$ is again defined as $a\otimes b = a+b$.

  On arrows $f\colon a\arr b$ and $g\colon a'\arr b'$ (where
  $a'\in \monM^{Y'}$, $b'\in \monM^{Z'}$ and
  $f\colon \Pow{[Y]^{a}} \to \Pow{[Z]^{b'}}$,
  $g\colon \Pow{[Y']^{a'}} \to \Pow{[Z']^{b'}}$ are the underlying
  functions), the tensor is given by
  \[ f\otimes g \colon \Pow{[Y+Y']^{a+a'}} \to \Pow{[Z+Z']^{b+ b'}},
    \quad 
    (f\otimes g) (U) = f(\cev{U}_Y) \cupa{Z} g(\vec{U}_Y)
  \]
  where \[ \cev{U}_Y = U\cap \{1,\dots,|Y|\} \mbox{ and }
  \vec{U}_Y = \{k \mid |Y|+k \in U\}. \] Furthermore:
  
  \[ U\cupa{Y} V = U\cup\{|Y| + k \mid k \in V\} \quad \text{(where
      $U\subseteq Y$)}  \]
      
\item The symmetry $\rho_{a,b}\colon a\otimes b\arr b\otimes a$ for
  $a\in \monM^Y$, $b\in \monM^Z$ is defined for
  $U\subseteq [Y+Z]^{a+b}$ as
  \[
    \rho_{a,b}(U)= \vec{U}_Y \cupa{Z} \cev{U}_Y \subseteq [Z+Y]^{b+a}
  \]
\item The unit $e$ is again the unique mapping
  $e \colon\emptyset \to \monM$.

\item The duplicator $\nabla_a \colon a\arr a\otimes a$ for
  $a\in \monM^Y$ is defined for $U\subseteq [Y]^a$ as
  \[\nabla_{a}(U)  = U \cupa{Y} U \subseteq [Y+Y]^{a+a}. \]
\item The discharger $!_a\colon a\arr e$ for $a\in \monM^Y$ is defined
  for $U \subseteq [Y]^a$ as $!_a(U) = \emptyset$.
\end{enumerate}
\end{theoremrep}

\begin{proof}
  Let $a\in \monM^Y$, $a'\in \monM^{Y'}$, $b\in \monM^Z$,
  $b'\in \monM^{Z'}$, $c\in \monM^W$, $c'\in \monM^{W'}$ be objects in
  $\Af$.

  We know that $\Af$ is a well-defined category from
  Lemma~\ref{le:cate}. We note that, disjoint unions of monotone
  functions are monotone, making the tensor well-defined.

  We now verify the axioms of gs-monoidal categories (see
  Definition~\ref{def:gs-monoidal}). The calculations are mostly
  straightforward.

  In the following we will often use the fact that
  $\cev{U}_Y\cupa{Y} \vec{U}_Y = U$ whenever $U\in \Pow{[Z]^b}$ and $Y\subseteq Z$.
\begin{enumerate}

\item functoriality of tensor:
\begin{itemize}
\item $\mathit{id}_{a\otimes b} = \mathit{id}_a \otimes
  \mathit{id}_b$

  \smallskip
  
  Let $U\subseteq [Y+Z]^{a+b}$, then:
\begin{align*}
  &(\mathit{id}_a \otimes \mathit{id}_b) (U) 
  = (\mathit{id}_a \otimes \mathit{id}_b) (\cev{U}_Y\cupa{Y} \vec{U}_Y) \\
  &= \mathit{id}_a (\cev{U}_Y) \cupa{Y}  \mathit{id}_b (\vec{U}_Y )
  = \cev{U}_Y\cupa{Y} \vec{U}_Y
  = U 
  = \mathit{id}_{a\otimes b}(U) 
\end{align*}  

\item $(g\otimes g') \circ (f\otimes f') = (g\circ f) \otimes (g'\circ
  f')$

  \smallskip

  Let $f\colon a\arr b$, $g\colon b\arr c$, $f'\colon a'\arr b'$,
  $g'\colon b'\arr c'$ and $u\in \monM^{Y+ Y'}$. We obtain:
 \begin{align*}
 & ((g\otimes g') \circ (f\otimes f')) (U) 
 = (g\otimes g') (f(\cev{U}_Y)\cupa{Z} f'(\vec{U}_Y)) \\
 &= g(f(\cev{U}_Y)) \cupa{W} g'(f'(\vec{U}_Y)) 
 = ((g\circ f) \otimes (g'\circ f'))(\cev{U}_Y\cupa{Y} \vec{U}_Y) \\
 &= ((g\circ f) \otimes (g'\circ f'))(U)
\end{align*}  

\end{itemize}
\item monoidality:
\begin{itemize}
\item $f\otimes \mathit{id}_e=f = \mathit{id}_e\otimes f$

  \smallskip

Let $f\colon a\arr b$ and $U\subseteq [Y]^a$. It holds that:
\begin{align*}
  &(f\otimes \mathit{id}_e) (U) = f(\cev{U}_Y) \cupa{Z}
  \mathit{id}_e(\vec{U}_Y)
  = f(U) \cupa{Z} \mathit{id}_e(\emptyset) 
  = f(U) \cupa{Z} \emptyset \\
  & =  f(U) = \emptyset \cupa{\emptyset} f(U)
  = \mathit{id}_e(\emptyset) \cupa{\emptyset} f(U) =
  \mathit{id}_e(\cev{U}_\emptyset)
  \cupa{\emptyset} f(\vec{U}_\emptyset) \\
  &= ( \mathit{id}_e \otimes f) (U)
\end{align*}
where we use $\cev{U}_Y = U$ and
$\vec{U}_Y = \emptyset$, since $U\subseteq Y$, as well as
$\cev{U}_\emptyset = \emptyset$ and $\vec{U}_\emptyset = U$.
  
  \medskip

\item $(f\otimes g) \otimes h = f \otimes (g\otimes h)$

  \smallskip

  Let $f\colon a\arr a'$, $g\in \colon b\arr b'$ and
  $h\colon c\arr c'$ and $U\subseteq [Y+ Z+ W]^{a+b+c}$. Then:
\begin{align*}
  & ((f\otimes g)\otimes h)(U) 
  = (f\otimes g) (\cev{U}_{Y+Z}) \cupa{Y'+Z'} h(\vec{U}_{Y+Z}) \\
  &= \big(f(\overleftarrow{(\cev{U}_{Y+Z})}_Y) \cupa{Y'} g(\overrightarrow{(\cev{U}_{Y+Z})}_Y)\big) \cupa{Y'+Z'} h(\vec{U}_{Y+Z}) \\
  &= \big(f(\cev{U}_Y) \cupa{Y'} g(\overleftarrow{(\vec{U}_{Y})}_Z)\big) \cupa{Y'+Z'} h(\vec{U}_{Y+Z}) \\
  &= f(\cev{U}_Y) \cupa{Y'} \big( g(\overleftarrow{(\vec{U}_{Y})}_Z)
  \cupa{Z'} h(\overrightarrow{(\vec{U}_{Y})}_Z) \big) \\
  &= f(\cev{U}_Y) \cupa{Y'} (g\otimes h) (\vec{U}_Y) 
  = (f\otimes (g\otimes h)) (U)
\end{align*} 
where we use
$\overleftarrow{(\cev{U}_{Y+Z})}_Y = \cev{U}_Y$,
$\overrightarrow{(\cev{U}_{Y+Z})}_Y = \overleftarrow{(\vec{U}_{Y})}_Z$
and $\vec{U}_{Y+Z} = \overrightarrow{(\vec{U}_{Y})}_Z$.
\end{itemize}

\item naturality:
\begin{itemize}
\item $(f'\otimes f) \circ \rho_{a,a'} = \rho_{b,b'} \circ (f\otimes
  f')$

  \smallskip

  Let $f\colon a\arr b$ and $f'\colon a'\arr b'$. Then for
  $U\subseteq [Y+Y']^{a+a'}$ it holds that:
 \begin{align*}
   & (\rho_{b,b'} \circ (f\otimes f')) (U)  \\
   &= \rho_{b,b'} (f(\cev{U}_Y) \cupa{Z} f'(\vec{U}_Y)) \\
   &= \overrightarrow{(f(\cev{U}_Y) \cupa{Z} f'(\vec{U}_Y))}_Z
   \cupa{Z'} \overleftarrow{(f(\cev{U}_Y) \cupa{Z}
     f'(\vec{U}_Y))}_Z \\
   &= f'(\vec{U}_Y) \cupa{Z'} f(\cev{U}_Y) \\
   &= f'(\overleftarrow{(\vec{U}_Y\cupa{Y'} \cev{U}_Y)}_{Y'}) \cupa{Z'}
   f(\overrightarrow{(\vec{U}_Y\cupa{Y'} \cev{U}_Y)}_{Y'}) \\
   &= (f'\otimes f)(\vec{U}_Y\cupa{Y'} \cev{U}_Y)
   = (f'\otimes f) (\rho_{a,a'}(U))
 \end{align*}
 where we use $\overleftarrow{(U\cupa{Y} V)}_Y = U$ and
 $\overrightarrow{(U\cupa{Y} V)}_Y = V$.
\end{itemize}

\item symmetry:
\begin{itemize}
\item $\rho_{e,e} = \mathit{id}_e$
    \smallskip

  Note that the only possibly argument is $\emptyset$ and hence:

  \[ \rho_{e,e} (\emptyset) = \vec{\emptyset}_\emptyset
    \cupa{\emptyset} \cev{\emptyset}_\emptyset = \emptyset
    \cupa{\emptyset} \emptyset = \emptyset =
    \mathit{id}_e(\emptyset) \]

\item $\rho_{b,a}\circ \rho_{a,b} = \mathit{id}_{a\otimes b}$

  \smallskip

  Let $U\subseteq [Y+Z]^{a+b}$, then:
  \begin{align*}
    & (\rho_{b,a} \circ \rho_{a,b}) (U) = \rho_{b,a}
    (\vec{U}_Y\cupa{Z}\cev{U}_Y) \\
    &=
    \overrightarrow{(\vec{U}_Y\cupa{Z}\cev{U}_Y)}_Z \cupa{Y}
    \overleftarrow{(\vec{U}_Y\cupa{Z}\cev{U}_Y)}_Z = \cev{U}_Y\cupa{Y}
    \vec{U}_Y =
    U \\
    &= \mathit{id}_{a\otimes b}(U)
  \end{align*}

\item
  $\rho_{a\otimes b,c} = (\rho_{a,c}\otimes \mathit{id}_b)
  \circ(\mathit{id}_a\otimes \rho_{b,c})$

  \smallskip
  
  Let $U\subseteq [Y+Z+W]^{a+b+c}$, then:
\begin{align*}
  &((\rho_{a,c} \otimes \mathit{id}_b) \circ (\mathit{id}_a \otimes  \rho_{b,c})) (U)\\
  &= (\rho_{a,c} \otimes \mathit{id}_b) (\mathit{id}_a (\cev{U}_Y) \cupa{Y} \rho_{b,c}(\vec{U}_Y)) \\
  &= (\rho_{a,c} \otimes \mathit{id}_b)(\cev{U}_Y\cupa{Y}
  (\overrightarrow{(\vec{U}_Y)}_Z \cupa{W} \overleftarrow{(\vec{U}_Y)}_Z)) \\
  &= (\rho_{a,c} \otimes \mathit{id}_b)(\cev{U}_Y\cupa{Y}
  (\vec{U}_{Y+Z} \cupa{W} \overleftarrow{(\vec{U}_Y)}_Z)) \\
  &= (\rho_{a,c} \otimes \mathit{id}_b)((\cev{U}_Y\cupa{Y}
  \vec{U}_{Y+Z}) \cupa{Y+W} \overleftarrow{(\vec{U}_Y)}_Z)) \\
  &= \rho_{a,c}(\cev{U}_Y\cupa{Y}
  \vec{U}_{Y+Z}) \cupa{W+Y} 
  \mathit{id}_b(\overleftarrow{(\vec{U}_Y)}_Z) \\
  &= (\vec{U}_{Y+Z}\cupa{W} \cev{U}_Y)
  \cupa{W+Y} \overleftarrow{(\vec{U}_Y)}_Z \\
  &= \vec{U}_{Y+Z}\cupa{W} (\cev{U}_Y
  \cupa{Y} \overleftarrow{(\vec{U}_Y)}_Z) \\
  &= \vec{U}_{Y+Z}\cupa{W} (\overleftarrow{(\cev{U}_{Y+Z})}_Y
  \cupa{Y} \overrightarrow{(\cev{U}_{Y+Z})}_Y) \\
  &= \vec{U}_{Y+Z}\cupa{W} \cev{U}_{Y+Z} 
  = \rho_{a\otimes b,c} (U)
\end{align*}
where we use
$\overrightarrow{(\vec{U}_Y)}_Z = \vec{U}_{Y+Z}$,
$\cev{U}_Y = \overleftarrow{(\cev{U}_{Y+Z})}_Y$ and
$\overleftarrow{(\vec{U}_Y)}_Z = \overrightarrow{(\cev{U}_{Y+Z})}_Y$.
\end{itemize}

\item gs-monoidality:
\begin{itemize}
\item $!_e = \nabla_e = \mathit{id}_e$

  \smallskip

  In this case $\emptyset$ is the only possible argument and we have:
  
  \[ !_e(\emptyset) = \emptyset = \mathit{id}_e (\emptyset) = \emptyset
    = \emptyset\cupa{\emptyset} \emptyset = \nabla_e(\emptyset) \]

\item coherence axioms:

  \smallskip
  For $U\subseteq [Y]^a$, we note that $\overleftarrow{(U\cup_Y U)}_Y = \overrightarrow{(U\cup_Y U)}_Y = U$.
\begin{itemize}
\item $(\mathit{id}_a\otimes \nabla_a) \circ \nabla_a =
  (\nabla_a\otimes \mathit{id}_a)\circ \nabla_a$ 

  \smallskip
  
Let $U\subseteq [Y]^a$, then:
\begin{align*}
  &((\mathit{id}_{a} \otimes \nabla_{a})\circ \nabla_a)(U)
    = (\mathit{id}_a\otimes \nabla_a)(U\cupa{Y}U)\\
  & = \mathit{id}_a(U) \cupa{Y} \nabla_a(U)
    = U\cupa{Y} (U\cupa{Y}U)\\
  &= (U\cupa{Y} U)\cupa{Y+Y}U
    = \nabla_a(U) \cupa{Y+Y} \mathit{id}_a(U) \\
  &= (\nabla_a \otimes \mathit{id}_a) (U\cupa{Y}U) 
    = (\nabla_a \otimes \mathit{id}_a) (\nabla_a(U))
\end{align*}

\item $\mathit{id}_a=(\mathit{id}_a\otimes !_a) \circ \nabla_a$

\smallskip
  
Let $U\subseteq [Y]^a$, then:
\begin{align*}
&((\mathit{id}_{a} \otimes !_{a}) \circ \nabla_a)(U) 
= (\mathit{id}_{a} \otimes !_{a}) (U\cupa{Y}U)\\
& = \mathit{id}_a(U) \cupa{Y} !_a(U) 
= \mathit{id}_a(U) \cupa{Y} \emptyset\\ 
& = \mathit{id}_a(U)
\end{align*}

\item $\rho_{a,a}\circ \nabla_a = \nabla_a$

  \smallskip

  Let $U\subseteq [Y]^a$, then:
  \[ (\rho_{a,a} \circ \nabla_a)(U) = \rho_{a,a} (U\cupa{Y}U) = U\cupa{Y}U
  = \nabla_a (U) \]

\end{itemize}
\item monoidality axioms:
\begin{itemize}
\item $!_{a\otimes b} = !_a \otimes !_b$

  \smallskip

Let $U\subseteq [Y+Z]^{a+b}$, then:
\begin{align*}
  !_{a\otimes b} (U)
  & =  \emptyset = \emptyset\cupa{\emptyset}
    \emptyset = !_a(\cev{U}_Y) \cupa{\emptyset} !_b(\vec{U}_Y) \\
  & =  (!_a\otimes !_b)(\cev{U}_Y\cupa{Y}\vec{U}_Y) = (!_a\otimes !_b)(U)
\end{align*}

\item $\nabla_a\otimes \nabla_b = (\mathit{id}_a\otimes \rho_{b,a}
  \otimes \mathit{id}_b) \circ \nabla_{a\otimes b}$

  \smallskip

Let $U\subseteq [Y+Z]^{a+b}$, then:
\begin{align*}
&(\mathit{id}_a \otimes \rho_{b,a} \otimes \mathit{id}_b) ( \nabla_{a\otimes b} (U))\\
& = (\mathit{id}_a \otimes (\rho_{b,a} \otimes \mathit{id}_b)) (U\cupa{Y+Z}U) \\
&= \mathit{id}_a (\overleftarrow{(U\cupa{Y+Z}U)}_Y) \cupa{Y}
(\rho_{b,a} \otimes \mathit{id}_b) (\overrightarrow{(U\cupa{Y+Z}U)}_Y) \\
&= \mathit{id}_a (\cev{U}_Y) \cupa{Y}
(\rho_{b,a} \otimes \mathit{id}_b) (\overrightarrow{(U\cupa{Y+Z}U)}_Y) \\
&= \cev{U}_Y \cupa{Y}
(\rho_{b,a} \otimes \mathit{id}_b) (\overrightarrow{(U\cupa{Y+Z}U)}_Y) \\
&= \cev{U}_Y \cupa{Y}
(\rho_{b,a} \otimes \mathit{id}_b) ((\vec{U}_Y\cupa{Z}\cev{U}_Y)
\cupa{Z+Y} \vec{U}_Y) \\
&= \cev{U}_Y\cupa{Y} ((\cev{U}_Y \cupa{Y}
\vec{U}_Y) \cupa{Y+Z}\vec{U}_Y)\\
& = (\cev{U}_Y\cupa{Y}\cev{U}_Y)\cupa{Y+Y} (\vec{U}_Y\cupa{Z}\vec{U}_Y) \\
&= \nabla_a(\cev{U}_Y) \cupa{Y+Y} \nabla_b(\vec{U}_Y) \\
& = (\nabla_a\otimes \nabla_b) (\cev{U}_Y\cupa{Y}\vec{U}_Y)
 = (\nabla_a\otimes \nabla_b) (U)
\end{align*}
where we use the fact that
\begin{align*} 
  \overleftarrow{(U\cupa{Y+Z}U)}_Y = \cev{U}_Y \text{ and }
  \overrightarrow{(U\cupa{Y+Z}U)}_Y = (\vec{U}_Y\cupa{Z}\cev{U}_Y) \cupa{Z+Y} \vec{U}_Y. \tag*{\qedhere}
\end{align*}
\end{itemize}
\end{itemize}
\end{enumerate}
\end{proof}

Finally, the approximation $\#$ is indeed gs-monoidal,
i.e., it preserves all the additional structure (tensor, symmetry, unit,
duplicator and discharger).

\begin{theoremrep}[$\#$ is gs-monoidal]
\label{thm:gs-monoidal-functor}
$\#\colon \Cf\to \Af$ is a gs-monoidal functor.

\end{theoremrep}

\begin{proof}
  We write $e',\otimes',!',\nabla',\rho'$ for the corresponding
  operators in category $\Af$. Note that by definition $e=e'$ and
  $\otimes$, $\otimes'$ agree on objects.

  First, categories $\Cf$ and $\Af$ are gs-monoidal by
  Theorem~\ref{thm:c-gs-monoidal} and~\ref{thm:a-gs-monoidal}.

  Furthermore we verify that:
  
\begin{enumerate}
\item monoidality:

  \smallskip
  
\begin{itemize}
\item $\#(e) = e'$

  \smallskip
  
  We have $\#(e) = e = e'$
\item $\#(a\otimes b) = \#(a) \otimes' \#(b)$

  \smallskip
  
  We have:
  \[ \#(a\otimes b) = a\otimes b = \#(a) \otimes' \#(b)\]
\end{itemize}
\item symmetry:

  \smallskip
  
\begin{itemize}
\item $\#(\rho_{a,b}) = \rho'_{\#(a),\#(b)}$

  \smallskip
  
  Let $U\subseteq [Y+Z]^{a+b}$, then for sufficiently small
  $\delta\sqsupset 0$ (note that such $\delta$ exists due to finiteness):
  \begin{align*}
    &\#(\rho_{a,b})(U) \\
    &= (\rho_{a,b})_\#^{a+b,\delta}(U) \\
    &= \{ w\in [Z+Y]^{b+a}\mid \rho_{a,b}(a+b)(w) \ominus \rho_{a,b}((a+b)\ominus \delta_U)(w) \sqsupseteq \delta\}\\
    &= \{ w\in [Z+Y]^{b+a}\mid (b+a)(w) \ominus ((b+a)\ominus \delta_{\rho'_{a,b}(U)})(w) \sqsupseteq \delta\}\\
    &= \rho'_{a,b}(U) = \rho'_{\#(a),\#(b)} (U)
  \end{align*}
  since $\rho_{a,b}$ distributes over componentwise subtraction and
  $\rho_{a,b}(\delta_U) = \delta_{\rho'_{a,b}(U)}$. The second-last
  equality holds since for all $w$ in the set we have
  $(b+a)(w) \sqsupset 0$.
\end{itemize}
\item gs-monoidality:

  \smallskip
  
  \begin{itemize}
  \item $\#(!_a) = !'_{\#(a)}$

    \smallskip
    
    Let $U\subseteq [Y]^a$, then for some $\delta$:
    \[ \#(!_a)(U) = (!_a)_\#^{a,\delta}(U) = \emptyset = !'_a(U) =
      !'_{\#(a)}(U) \] since the codomain of $(!_a)_\#^{a,\delta}(U)$
    is $\Pow{\emptyset}$ and hence the only possible value for
    $(!_a)_\#^{a,\delta}(U)$ is $\emptyset$.
  \item $\#(\nabla_a) = \nabla'_{\#(a)}$
    
    \smallskip
    
Let $U\subseteq [Y]^a$, then for sufficiently small $\delta\sqsupset 0$:
\begin{align*}
  & \#(\nabla_a)(U) \\
  &= (\nabla_a)^{a,\delta}_\#(U)\\
  &= \{ w\in [Y+Y]^{a+a}\mid \nabla_a(a)(w) \ominus \nabla_a(a\ominus \delta_U)(w) \sqsupseteq \delta\}\\
  &= \{ w\in [Y+Y]^{a+a}\mid (a+a)(w) \ominus ((a+a)\ominus \delta_{\nabla'_a(U)})(w) \sqsupseteq \delta\}\\
  &= \nabla'_a(U) = \nabla'_{\#(a)}(U)
\end{align*}
since $\nabla_{a}$ distributes over componentwise subtraction and
$\nabla_{a}(\delta_U) = \delta_{\nabla'_a(U)}$. The second-last
equality holds since for all $w$ in the set we have
$(a+a)(w) \sqsupset 0$. \qedhere
\end{itemize}
\end{enumerate}
\end{proof}

\section{\udefix: A Tool for Fixpoints Checks}
\label{sec:tool}

\subsection{Overview}

We present a tool, called $\udefix$, which exploits gs-monoidality as
discussed before and allows the user to construct functions
$f\colon \monM^Y \to \monM^Y$ (with $Y$ finite) as a circuit. As basic
components, $\udefix$ can handle all functions presented in
Section~\ref{sec:categories-functor}
(Table~\ref{tab:basic-functions-approximations}) and
addition/subtraction by a fixed constant $w$, denoted
$\mathit{add}_w$/$\mathit{sub}_w$ (both are non-expansive functions).
Since the approximation functor $\#$ is gs-monoidal on the finitary
subcategories, this circuit can then be transformed automatically, in
a compositional way, into the corresponding approximation $f_\#^a$,
for some given $a\in \monM^Y$. By computing the greatest fixpoint of
$f_\#^a$ and checking for emptiness, $\udefix$ can check whether
$a = \mu f$.
In addition, it is possible to check whether a given post-fixpoint $a$
is below the least fixpoint $\mu f$ (recall that in this case the
check is sound but not complete). The dual checks (for greatest
fixpoint and pre-fixpoints) are implemented as well.

The tool is shipped with pre-defined functions implementing examples
concerning case studies on termination probability, bisimilarity,
simple stochastic games, energy games, behavioural metrics and Rabin
automata.

{\udefix} is a Windows-Tool created in Python, which can be obtained
from \url{https://github.com/TimoMatt/UDEfix}.

Building the desired function $f\colon \monM^Y \to \monM^Y$ requires
three steps: 
\begin{itemize}
\item Choosing the MV-algebra $\monM$ of interest.
\item Creating the required basic functions by specifying their
  parameters.
\item Assembling $f$ from these basic functions. 
\end{itemize}

\subsection{Tool Areas}

Concretely, the GUI of \udefix\ is separated into three areas: the
\textsf{Content} area, \textsf{Building} area and
\textsf{Basic-Functions} area. Under \textsf{File}-\textsf{Settings}
the user can set the MV-algebra.  Currently
the MV-chains $[0,k]$ (\textsf{algebra 1}) and $\{0,\dots,k\}$
(\textsf{algebra 2}) for arbitrary $k$ are supported
(see Example~\ref{ex:mv-chains})

\paragraph*{\textsf{Basic-Functions} Area:}
The \textsf{Basic-Functions} area contains the basic functions,
encompassing those listed in
Table~\ref{tab:basic-functions-approximations}
(Section~\ref{sec:preliminaries}), as well as addition and subtraction
by a constant $w$.  Via drag-and-drop (or right-click) these basic
functions can be added to the \textsf{Building} area to create a
\textsf{Function} box. Each such box requires three (in the case of
$\tilde{\mathcal{D}}$ two) \textsf{Contents}: The \textsf{Input} set,
the \textsf{Output} set and an additional required parameter (see
Table~\ref{tab:additional-para}). These \textsf{Contents} can be
created in the \textsf{Content} area.

\begin{table}[h]
\begin{center}
\begin{tabular}{|c|c|c|c|c|c|}
  \hline \textbf{Basic Function} & $c_k$ & $u^*$ &
  $\min_\mathcal{R}/\max_\mathcal{R}$& $\mathit{add}_w$ & $\mathit{sub}_w$ \\
  \hline
  \textbf{Req. Parameter} & $k\in \monM^Z$ & $u\colon Z\to Y$ & $R\subseteq Y\times Z$ & $w\in \monM^Y$ & $w\in \monM^Y$ \\
  \hline
\end{tabular} 
\end{center}
\caption{Additional parameters for the basic functions from
  Table~\ref{tab:basic-functions-approximations}.}
\label{tab:additional-para}
\end{table}

Additionally the \textsf{Basic-Functions} area offers functionalities
for composing functions via disjoint union (more concretely, this is
handled by the auxiliary \textsf{Higher-Order Function}) and the
\textsf{Testing} functionality for fixpoint checks which we will
discuss in the next paragraph.

\paragraph*{\textsf{Building} Area:}
The user can connect the created \textsf{Function} boxes to construct
the function $f$ of interest. Composing functions is as simple as
connecting two \textsf{Function} boxes in the correct order by
mouse-click. Disjoint union is achieved by connecting two boxes to the
same box. Note that \textsf{Input} and \textsf{Output} sets of
connected \textsf{Function} boxes need to match.  As an example, in
Figure~\ref{fig:assembling-beh-metrics} we show how the function
$\mathcal{B}$, discussed in Section~\ref{sec:case-study}, for
computing the behavioural distance of a labeled Markov chain can be
assembled.  The construction exactly follows the structure of the
diagram in Figure~\ref{fig:decomp-beh-metric}. Here, the parameters
are instantiated for the labeled Markov chain displayed in
Figure~\ref{fig:prob-ts} (left-hand side).

\begin{figure}[t]
\begin{center}
\includegraphics[scale=0.4]{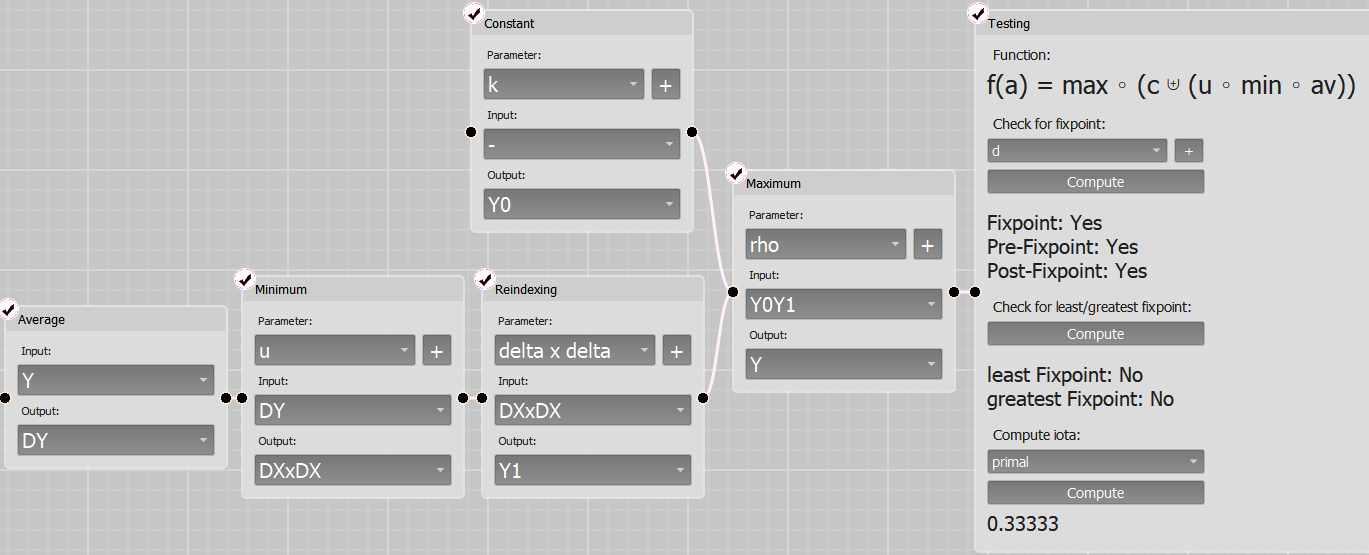}
\end{center}
\vspace{-12pt}
\caption{Assembling the function $\mathcal{B}$ from Section~\ref{sec:case-study}.}
\label{fig:assembling-beh-metrics}
\end{figure}

The special box \textsf{Testing} is always required to appear at the
end. Here, the user can enter some mapping $a\colon Y \to \monM$, test
if $a$ is a fixpoint of the function $f$ of interest and then verify
if $a = \mu f$. As explained before, this is realised by computing the
greatest fixpoint of the approximation $\nu f_\#^a$ . In case this is
not empty and thus $a \neq \mu f$, the tool can produce a suitable
value which can be used for decreasing $a$, needed for iterating to the least fixpoint
from above (respectively increasing $a$ for iterating to the greatest
fixpoint from below). There is also support for comparison with pre-
and post-fixpoints.

\begin{example}
  As an example, consider the left-hand system in
  Figure~\ref{fig:prob-ts} and consider the function $\mathcal{B}$
  from Section~\ref{sec:case-study} whose least fixpoint corresponds to the
  behavioural distance for a labeled Markov chain.

  We now define a fixpoint of $\mathcal{B}$, which is not the least,
  namely $d\colon Y \to [0,1]$ with $d(3,3) = 0$,
  $d(1,1) = \nicefrac{1}{2}$,
  $d(1,2) = d(2,1) = d(2,2) = \nicefrac{2}{3}$ and $1$ for all other
  pairs. Note that $d$ is not a pseudo-metric, but it is a fixpoint of
  $\mathcal{B}$. For understanding why $d$ is a fixpoint, consider,
  for instance, the pair $(1,2)$. Due to the labels, an optimal
  coupling for the successors of $1,2$ assigns
  $(3,3)\mapsto \nicefrac{1}{3}$, $(4,4)\mapsto \nicefrac{1}{2}$,
  $(3,4)\mapsto \nicefrac{1}{6}$. Hence, by the fixpoint equation, we
  have
  \[ d(1,2) = \nicefrac{1}{3}\cdot d(3,3) + \nicefrac{1}{2}\cdot
    d(4,4) + \nicefrac{1}{6} \cdot d(3,4) = \nicefrac{1}{3}\cdot 0 +
    \nicefrac{1}{2}\cdot 1 + \nicefrac{1}{6} \cdot 1 = \nicefrac{1}{2}
    + \nicefrac{1}{6} = \nicefrac{2}{3}. \] A similar argument can be
  made for the remaining pairs.
  
  By clicking \textsf{Compute} in the \textsf{Testing}-box, $\udefix$
  displays that $d$ is a fixpoint and tells us that $d$ is in fact not
  the least and not the greatest fixpoint. It also computes the
  greatest fixpoints of the approximations step by step (via Kleene
  iteration) and displays the results to the user.  In this case
  $\nu f_\#^a = \{ (4,4)\}$, indicating that the distance $d(4,4) = 1$
  over-estimates the true value and can be decreased. The fact that
  also other pairs over-estimate their value (but to a lesser degree),
  will be detected in later steps in the iteration to the least
  fixpoint from above.
\end{example}

\paragraph*{\textsf{Content} Area:}
Here the user can create sets, mappings and relations which are used
to specify the basic functions. 
The user can create a variety of different types of sets, such as
$X=\{1,2,3,4\}$, which is a basic set of numbers, and the set
$D=\{p_1,p_2,p_3,p_4\}$ which is a set of mappings representing
probability distributions. These objects are called \textsf{Content}.

Once \textsf{Input} and \textsf{Output} sets are created we can
specify the required parameters (cf. Table~\ref{tab:additional-para})
for a function. Here, the created sets can be chosen as domain and
co-domain. Relations can be handled in a similar fashion: Given the
two sets one wants to relate, creating a relation can be easily
achieved by checking some boxes. Some useful in-built relations
like ``is-element-of''-relation and projections to the $i$-th component
are pre-defined.

By clicking on the icon ``+'' in a \textsf{Function} box, a new
function with the chosen \textsf{Input} and \textsf{Output} sets is
created.  The additional parameters
(cf. Table~\ref{tab:additional-para}) have domains and co-domains
which need to be created by the user or are provided by the chosen
MV-algebra.

The \textsf{Testing} function $a$ (i.e., the candidate
(pre/post-)fixpoint) is a mapping as well and can be created as all
other functions.

In Figure~\ref{fig:contents} we give examples of how
contents can be created: we show the creation of a set ($Y = X\times X$), a
distance function ($d$) and a relation ($\rho$).

\begin{figure}[t]
\begin{center}
\includegraphics[scale=0.6]{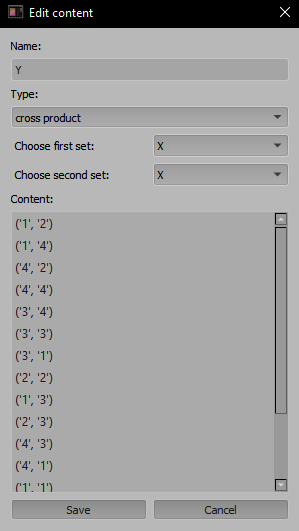} \quad
\includegraphics[scale=0.6]{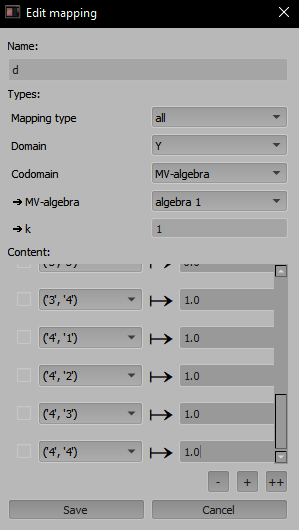} \quad
\includegraphics[scale=0.6]{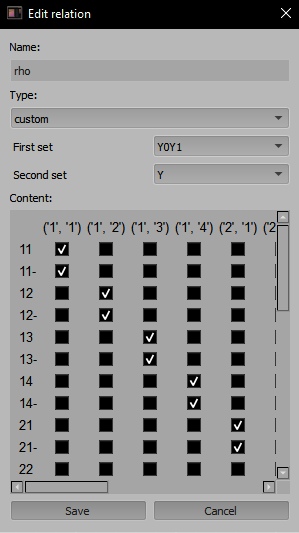}
\end{center}
\caption{Contents: Set $Y$, Mapping $d$, Relation $\rho$.}
\label{fig:contents}
\end{figure}

\subsection{Tutorial}

We now provide a small tutorial intended to clarify the use of the
tool. It deals with a simple example: a function whose least fixpoint
provides the termination probability of a Markov chain.
Specifically, we continue the example from Section~\ref{sec:motivation} and
Example~\ref{ex:termination-mc-comp}, and we consider the Markov chain
$(S,T,\eta)$ in Figure~\ref{fig:markov-chain2}.

Remember that the function $\mathcal{T}$, whose least fixpoint is the
termination probability, can be decomposed as follows:
\[ \mathcal{T}= (\eta^*\circ \tilde{\mathcal{D}})\otimes c_k\]

In order to start, one first chooses the correct MV-algebra under
\textsf{Settings} and creates the sets $S,T,S\setminus T$
(\textsf{Input} and \textsf{Output} sets). The creation of set $S$ is
exemplified in Figure~\ref{fig:tutorial-set-distr} (left-hand side). The
tool supports several types of sets and operators on sets (such as
complement, which makes it easy to create $S\backslash T$. Next, we
create the set $D$ of probability distributions, consisting of the
mappings $p_x,p_y,p_z$. In Figure~\ref{fig:tutorial-set-distr}
(right-hand side) we show the creation of $p_x$.

\begin{figure}
  \centering
  \includegraphics[scale=0.6]{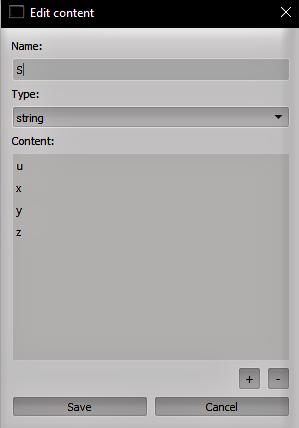} \quad
  \includegraphics[scale=0.448]{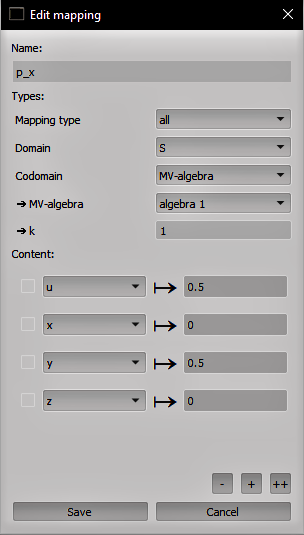} 
  \caption{Creation of set $S$ and probability distribution $p_x$ for
    the example.}
  \label{fig:tutorial-set-distr}
\end{figure}

Now we create the basic function boxes and connect them in the correct
way (see Figure~\ref{fig:tutorial-diagram}). The additional parameters
according to Table~\ref{tab:additional-para} -- in this case the map
$c_k$ and the reindexing $\eta^*$ based on the successor map -- can be
created by clicking the icon ``+'' in the corresponding box (see
Figure~\ref{fig:tutorial-k-eta-a1}) (left-hand side for $k$ and middle
for $\eta$).

\begin{figure}
  \centering
  \includegraphics[scale=0.4]{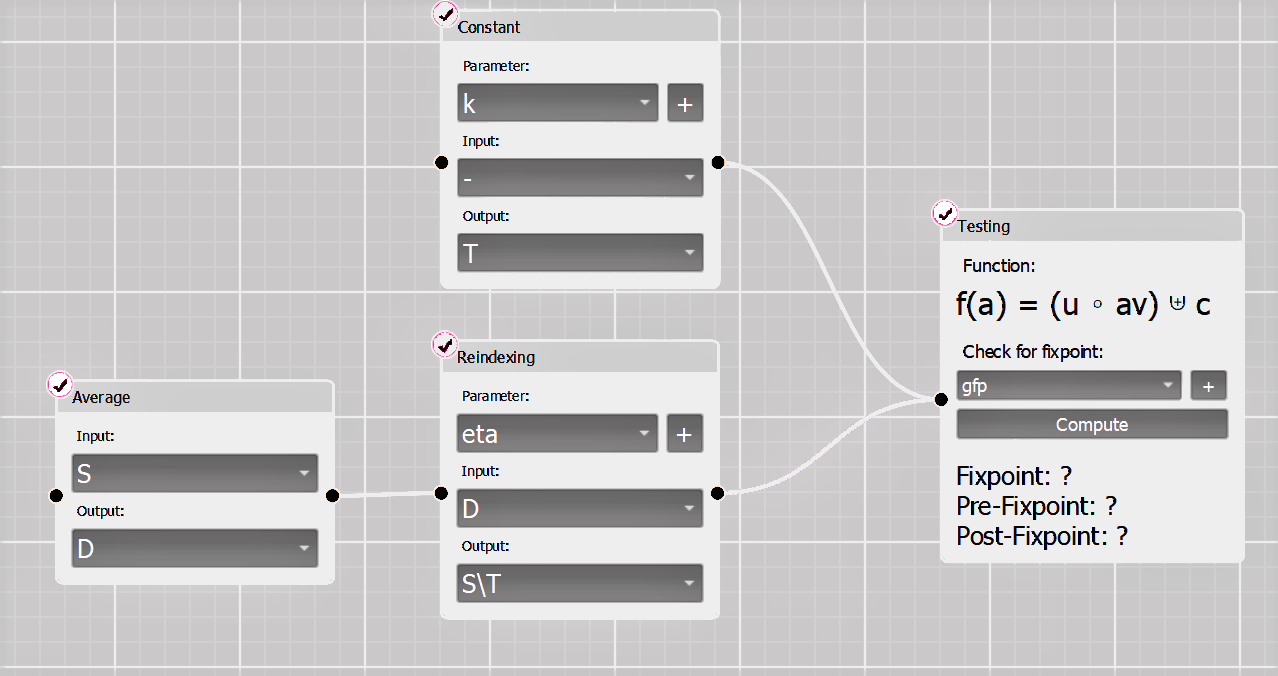}
  \caption{Assembling the function $\mathcal{T}$.}
  \label{fig:tutorial-diagram}
\end{figure}

\begin{figure}
  \centering
  \includegraphics[scale=0.60]{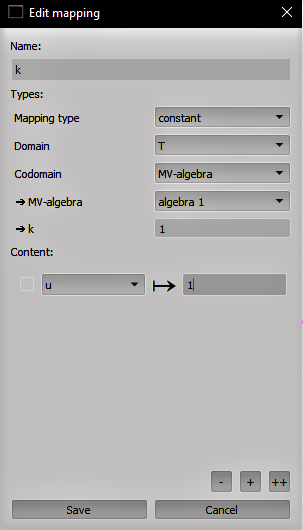} \quad
  \includegraphics[scale=0.60]{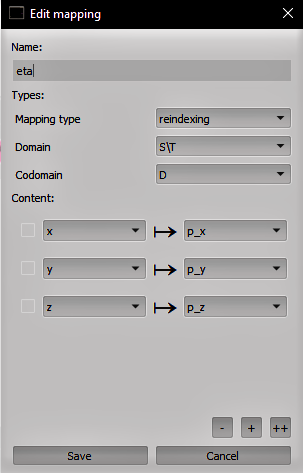} \quad
  \includegraphics[scale=0.55]{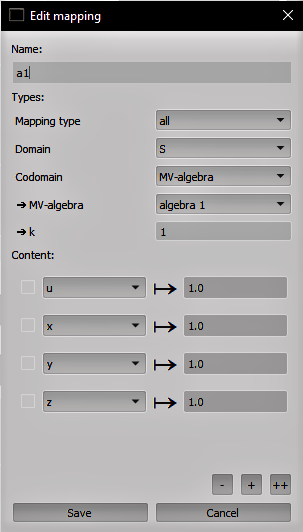}
  
  \caption{Creation of the parameters $k$ and $\eta^*$ and the greatest
    fixpoint~$a_1 = \mu \mathcal{T}$.}
  \label{fig:tutorial-k-eta-a1}
\end{figure}

We can also assemble several test functions (e.g., possible candidate
fixpoints), among them the greatest fixpoint $a_1 = \nu \mathcal{T}$
(see Figure~\ref{fig:tutorial-k-eta-a1}, right-hand side).

When testing $a_1$ we obtain the results depicted in
Figure~\ref{fig:tutorial-testing-a1}. In fact
$\nu \mathcal{T}_\#^{a_1} = \{y,z\} \neq \emptyset$, which tells us
that $a_1$ is not the least fixpoint. 
\begin{figure}
  \centering
  \includegraphics[scale=0.55]{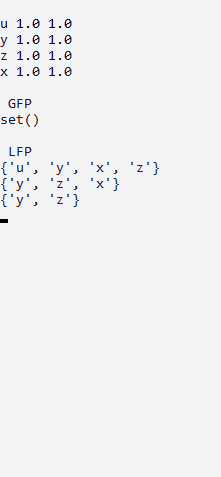} \quad
    \includegraphics[scale=0.55]{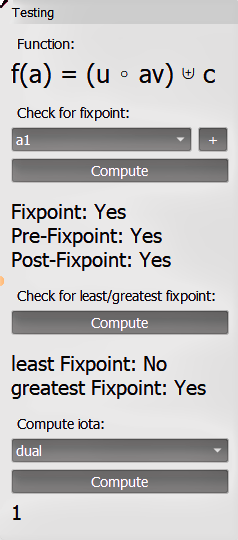}
  \caption{Checking the candidate fixpoint $a_1$.}
  \label{fig:tutorial-testing-a1}
\end{figure}

Similarly, one can test $a_2 = \mu \mathcal{T}$, the least fixpoint,
obtaining $\nu \mathcal{T}_\#^{a_2} = \emptyset$. This allows the user
to deduce that $a_2$ is indeed the least fixpoint. As mentioned before,
one can also test whether a pre-fixpoint is below the greatest
fixpoint or a post-fixpoint above the least fixpoint, although such
tests are sound but not complete.

\section{Conclusion, Related  and Future Work}
\label{sec:conclusion}

We have shown that a framework originally introduced
in~\cite{bekp:fixpoint-theory-upside-down,BEKP:FTUD-journal} for
analysing fixpoint of non-expansive functions over MV-algebras can be
naturally cast into a gs-monoidal setting. When considering the finitary categories, both the non-expansive
functions of interest and their approximations live in two gs-monoidal
categories and are related by a gs-monoidal functor $\#$.  We also
developed a general theory for constructing approximations of
predicate liftings, which find natural application in the definition
of behavioural metrics over coalgebras.

Graph compositionality has been studied from several angles and
has always been an important part of the theory of graph
rewriting. For instance, one way to explain the double-pushout
approach~\cite{eps:gragra-algebraic} is to view the graph to be
rewritten as a composition of a left-hand side and a context, then
compose the context with the right-hand side to obtain the result of
the rewriting step.

Several algebras have been proposed for a compositional view on
graphs, see for
instance~\cite{bc:graph-expressions,k:hypergraph-construction-journal,bk:axiomatization-graphs}.
For the compositional modelling of graphs and graph-like structures it
has in particular proven useful to use the notion of monoidal
categories~\cite{ml:categories}, i.e., categories equipped with a
tensor product. There are several variants of such categories, such as
gs-monoidal categories, that have been shown to be suitable for
specifying term graph rewriting (see
e.g.~\cite{g:concurrent-term-rewriting,gh:inductive-graph}).

In our work, the compositionality properties arising from the
gs-monoidal view of the theory are at the basis of the development of
the prototypical tool {\udefix}. The tool allows one to build the
concrete function of interest out of some basic components. Then the
approximation of the function can be obtained compositionally from the
approximations of the components and one can check whether some
fixpoint is the least or greatest fixpoint of the function of
interest. Additionally, one can use the tool to show that some
pre-fixpoint is above the greatest fixpoint or some post-fixpoint is
below the least fixpoint.

\paragraph*{\emph{Related work:}} This paper is based on fixpoint theory,
coalgebras, as well as on the theory of monoidal categories. Monoidal
categories~\cite{ml:categories} are categories equipped with a
tensor. It has long been realised that monoidal categories can have
additional structure such as braiding or symmetries. Here we base our
work on so called gs-monoidal categories~\cite{cg:term-graphs-gs-monoidal,gh:inductive-graph}, called
s-monoidal in~\cite{g:concurrent-term-rewriting}. These are symmetric
monoidal categories, equipped with a discharger and a duplicator. Note
that ``gs'' originally stood for ``graph substitution'' as such
categories were first used for modelling term graph rewriting.

We view gs-monoidal categories as a means to compositionally build
monotone non-expansive functions on complete lattices, for which we
are interested in the (least or greatest) fixpoint.  Such fixpoints
are ubiquitous in computer science, here we are in particular
interested in applications in concurrency theory and games, such as
bisimilarity~\cite{s:bisimulation-coinduction}, behavioural
metrics~\cite{dgjp:metrics-labelled-markov,b:prob-bisimilarity-distances,cbw:complexity-prob-bisimilarity,bbkk:coalgebraic-behavioral-metrics}
and simple stochastic games~\cite{c:algorithms-ssg}. In recent work we
have considered strategy iteration procedures inspired by games for
solving fixpoint equations~\cite{bekp:lattice-strategy-iteration}.

Fixpoint equations also arise in the context of coalgebra~\cite{r:universal-coalgebra}, a general framework for investigating
behavioural equivalences for systems that are parameterised -- via a
functor -- over their branching type (labelled, non-deterministic,
probabilistic, etc.). Here in particular we are concerned with
coalgebraic behavioural metrics~\cite{bbkk:coalgebraic-behavioral-metrics}, based on a generalisation
of the Wasserstein or Kantorovich lifting~\cite{v:optimal-transport}. Such liftings require the notion of
predicate liftings, well-known in coalgebraic modal logics~\cite{s:coalg-logics-limits-beyond-journal}, lifted to a quantitative
setting~\cite{bkp:up-to-behavioural-metrics-fibrations}.

\paragraph*{\emph{Future work:}}
One important question is still open: we defined an approximation
$\#$, relating the concrete category $\C$ of functions of type
$\monM^Y\to\monM^Z$ -- where $Y,Z$ might be infinite -- to their
approximations, living in $\A$. It is unclear whether $\#$ is
a lax or even proper functor, i.e., whether it (laxly) preserves
composition. For finite sets functoriality derives from a non-trivial
result in~\cite{BEKP:FTUD-journal} and it is unclear whether it can be
extended to the infinite case. If so, this would be a valuable step to
extend the theory to infinite sets.

In this paper we illustrated the approximation for predicate liftings
via the powerset and the distribution functor. It would be worthwhile
to study more functors and hence broaden the applicability to other
types of transition systems.

Concerning $\udefix$, we plan to extend the tool to compute fixpoints,
either via Kleene iteration or strategy iteration (strategy iteration
from above and below), as detailed
in~\cite{bekp:lattice-strategy-iteration}. Furthermore for convenience
it would be useful to have support for generating fixpoint functions
directly from a given coalgebra respectively transition system.

\medskip

\noindent\emph{Acknowledgements:} We want to thank Ciro Russo who
helped us with a question on the connection of MV-algebras and
quantales.

\bibliographystyle{alphaurl}
\bibliography{references}
\end{document}